\documentclass[12pt,a4paper]{report}

\makeindex

\usepackage{mathphdthesis}

\usepackage{amssymb, amsthm, amsmath}
\usepackage{pstricks}

\def\pstlw{.8pt}

\begin{document}


\titlepgtrue 
\signaturepagetrue 
\copyrighttrue 
\coauthortrue
\abswithesistrue 
\acktrue 
\tablecontentstrue 
\tablespagetrue 
\figurespagetrue 

\title{Entanglement, Invariants, and Phylogenetics} 
\author{Jeremy G Sumner} 
\prevdegrees{B.Sc. Hons (Tas)} 
\advisor{Dr Peter Jarvis} 
\dept{Mathematics and Physics} 
\submitdate{December, 2006} 

\newcommand{\abstextwithesis}
{\noindent This thesis develops and expands upon known techniques of mathematical physics relevant to the analysis of the popular Markov model of phylogenetic trees required in biology to reconstruct the evolutionary relationships of taxonomic units from biomolecular sequence data.\\
The techniques of mathematical physics are plethora and have been developed for some time. The Markov model of phylogenetics and its analysis is a relatively new technique where most progress to date has been achieved by using discrete mathematics. This thesis takes a group theoretical approach to the problem by beginning with a remarkable mathematical parallel to the process of scattering in particle physics. This is shown to equate to branching events in the evolutionary history of molecular units. The major technical result of this thesis is the derivation of existence proofs and computational techniques for calculating polynomial group invariant functions on a multi-linear space where the group action is that relevant to a Markovian time evolution. The practical results of this thesis are an extended analysis of the use of invariant functions in distance based methods and the presentation of a new reconstruction technique for quartet trees which is consistent with the most general Markov model of sequence evolution.}

\newcommand{\acknowledgement}
{\noindent First and foremost my thanks go to my supervisor Peter Jarvis. Not only for having the insight to take on this novel work and his outstanding knowledge of mathematical physics, but also for being a true friend and good bloke.

\noindent
These people have all played their own special role in bringing this thesis to fruition: Michael Sumner, Robert Delbourgo, Patrick McLean; William Joyce and the Physics department of the University of Cantebury; Mike Steel and the organisers of the New Zealand phylogenetics meeting; Rex Lau, Lars Jermiin, Michael Charleston and SUBIT; Alexei Drummond; Simon Wotherspoon (for giving me such a hard time), Malgorzata O'Reilly, Jim Bashford, Giuseppe Cimo, Stuart Morgan, Isamu Imahori and Graham Legg; Mum, Dad and Kate; Keith, Tim, Sarah, Wazza and Beans. 

\noindent
A special mention for my high school maths teacher Mr. Rush, who used to laugh when I continually interrupted his classes with: ``That's all very well, Mr. Rush, but how is this going to help me lay bricks?''

}

\newcommand{\quotations}
{\noindent\textit{Here, as it draws to its last Halt, if anywhere, might both Gentlemen take joy of a brief Holiday from Reason. Yet, ``Too busy,'' Mason insists, and ``Far too cheerful for thah','' supposes Dixon.}

\textbf{Mason and Dixon}\\\indent
Thomas Pynchon

}

\beforepreface
\afterpreface



\newcommand{\fra}[2]{\textstyle{\frac{#1}{#2}}}
\newcommand{\beqn}{\begin{eqnarray}\begin{aligned}}
\newcommand{\eqn}{\end{aligned}\end{eqnarray}}

\newtheorem{thm}{Theorem}[section]
\newtheorem{cor}[thm]{Corollary}
\newtheorem{lem}[thm]{Lemma}
\newtheorem{prop}[thm]{Proposition}
\newtheorem{remark}[thm]{}
\theoremstyle{definition}
\newtheorem{define}[thm]{Definition}
\newtheorem{example}{Example}[thm]
\newtheorem{postulate}{Postulate}


\pagenumbering{arabic}

\chapter{Introduction}\label{chap1}

The rationale of this thesis is taken from a remarkable analogy between the stochastic models used to infer phylogenetic relationships in mathematical biology and the structure of multiparticle quantum physics. There is a direct relationship between Feynman diagrams that describe the interactions of sub-atomic particles and phylogenetic trees that graphically represent the evolutionary relationship between taxonomic units. A Feynman diagram gives the graphical representation of creation and annihilation events of particle interactions. A taxonomic unit may be any biomolecular unit such as a gene, an amino acid or base pair, and the time evolution of these molecular units is modelled stochastically under a Markov assumption. Techniques which reconstruct the evolutionary history of molecular units from present observations are based on these models. Given the correct framework, these Markov models and the formalism of multiparticle quantum mechanics can be put into a mathematical correspondence. This is a very useful observation because phylogenetics is a relatively new mathematical problem (for example see the classic paper by Felsenstein \cite{felsenstein1981}) whereas the mathematics of particle physics has been studied for over a century. (For an outstanding introduction to the history of theoretical particle physics see \cite{pais1988}, and for a comprehensive introduction to mathematical physics see \cite{szekeres2004}.) Given that there is a mathematical connection between the two problems it would certainly be unfortunate to see results that have been obtained in physics re-derived independently in the context of phylogenetics. This thesis looks at a particular aspect of quantum systems known as \textit{entanglement} and shows that measures of entanglement can be utilized to improve the reconstruction of phylogenetic relationships.   
\\
We will need to be clear that the probabilities associated with quantum systems and those of phylogenetic models arise in quite a different scientific way. Quantum mechanics is a probabilistic theory because the theoretical predictions give the correct statistical behaviour regarding the outcomes of particular experiments. The theoretical predictions can be used to infer (incredibly accurately) the distribution of results for many repetitions of the same experiment. (For a popular discussion of the amazing accuracy of quantum theory see Feynman's discussion of the magnetic moment on the electron as predicted from quantum electrodynamics \cite{feynman1988}.) Since quantum theory is (and should be) seen as a \textit{theory} of nature there has been argument for many decades on how to interpret this probabilistic aspect of quantum theory. This argument raises quite profound scientific and philosophical issues which, thankfully, we will not be concerned with in this thesis. Models of phylogenetics are exactly that -- \textit{models}, and should not be seen as being theories of nature. No one would argue that the time evolution of molecular units follow the Markov model of phylogenetics in detail, but rather that these models are the best (tractable) approximation that give us recourse to establishing properties of phylogenetic history. Primarily the points of interest are the branching structure of the evolutionary history and also the evolutionary distance (or time) between branching events. \\
After we have made the mathematical analogy between quantum theory and the Markov model of phylogenetics, we will concentrate on only a small part of what can be done using techniques known in mathematical physics. We will focus on the study of entanglement invariants and their generalization to the phylogenetic case \cite{sumner2005,sumner2006}. There is potential for concentrating on other techniques such as Lie algebra symmetries \cite{bashford2004} and the analysis of the path integral formulation \cite{jarvis2001,jarvis2005}, but these techniques will not be explored here.\\
The distance based  technique has been used in phylogenetics as a tree building algorithm following the discovery that it is possible to calculate a distance from the observed sequences that is consistent with the Markov model. This distance function is a well defined mathematical object known as a group invariant function and is used in quantum physics to quantify and test for the phenomenon of entanglement. Entanglement is a general property that can exist in many different physical systems and the invariant function used as a distance measure in phylogenetics is used to quantify entanglement for only the most elementary case. Hence, it seems astute to investigate what the next most complicated types of entanglement correspond to in phylogenetics.  
\subsubsection{Theoretical outcomes of the thesis} 
We present a group representation theoretic analysis of the Markov model of phylogenetic trees. Specifically this formalism is used to construct all the one-dimensional representations of the (appropriately defined) Markov semigroup. These one-dimensional representations occur as polynomials in the (discrete) probability distributions predicted from the Markov model which we coin \textit{Markov invariants}. We establish the connection between these one-dimensional representations and that of \textit{phylogenetic invariants} \cite{cavender1987,evans1993,felsenstein1991,steel1993} and pairwise distance measures \cite{gu1996,lockhart1994}. This representation theoretical approach touches upon existing techniques and can be incorporated into known algorithms to give novel results and insights to the problem of phylogenetic reconstruction. The main theoretical outcome of the thesis is this use of representation theory. We will also develop the theory of invariants of the general linear group on a tensor product space and show how to infer existence of these invariants in different cases. We develop a procedure for computing the explicit form of these invariant functions, firstly developed for the general linear group and then generalized to the Markov semigroup. 
\subsubsection{Practical outcomes of the thesis}
We study a group invariant function, well known in quantum physics as the \textit{tangle}, in the context of phylogenetics. The tangle is used in physics to give a measure of the amount of entanglement between three qubits. Qubits are two state objects in quantum physics and correspond in phylogenetics to a probability distribution on two states. In phylogenetics the classic example is to use the DNA as a state space and hence the case of four state objects is of interest. To this end we have generalized the tangle to the case of three and four character states. This is a new result that to the best of the author's knowledge was previously unknown. Having successfully generalized the tangle we investigate how the tangle can be used to construct improved phylogenetic distance matrices. Additionally we study a set of Markov invariants which exist for the case of phylogenetic quartet tree. In the case of the evolution of four taxa there are three possible historical evolutionary relationships. We show that these Markov invariants can be used to distinguish these three cases under the assumption of the most general Markov model. It is expected that the use of the tangle to construct distance matrices and using the Markov invariants to distinguish the three possible quartets will lead to improvements of the reconstruction of phylogenetic relationships from observed biomolecular data.\\    
\subsubsection{Structure of the thesis}
Chapter \ref{chap2} begins by introducing the mathematical material needed to understand the results presented in this thesis. This includes a short introduction to group representation theory, group characters and tensor product; a presentation of the Schur/Weyl duality and the Schur functions; a definition of group invariant functions and their relation to one-dimensional representations. The chapter ends with several relevant examples of invariants of the general linear group.\\
Chapter \ref{chap3} begins with a light speed introduction to the formalism of quantum mechanics, the concept of entanglement and mathematical analysis thereof using group invariant functions. The Markov model of phylogenetic trees is then developed in its usual presentation, followed by a change of formalism which makes apparent the analogy between phylogenetic trees and multiparticle quantum systems. The chapter ends with a detailed analysis of the mathematical analysis of the invariant functions when evaluated upon a phylogenetic tree.\\
Chapter \ref{chap4} gives a review of phylogenetic distance measures and shows how the tangle invariant function used to analyse three qubit entanglement can be generalized to the phylogenetic case and used to improve popular distance measures. This is done by defining the branch lengths of a phylogenetic tree, reviewing the standard measure known as the \textit{$\log\det$} and then using the tangle invariant to give a consistent distance measure for the case of quartets.\\
Chapter \ref{chap5} returns to the mathematical detail of Chapter \ref{chap2} and derives invariant functions that are more closely relevant to the Markov model of a phylogenetic tree. This is done by first defining the Markov semigroup. The invariant functions of the general linear group are rederived using a technique which is generalized to derive the Markov invariants. Finally we examine the structure of the Markov invariants on a phylogenetic tree. In particular we concentrate on the quartet case where there exists four Markov invariants which can be used to distinguish between the three possible quartet trees.\\


\chapter{Mathematical background}\label{chap2}

In this chapter we will present the requisite mathematical background for developing the results presented in this thesis. It will be assumed that the reader is familiar with elementary concepts of algebra, most importantly the theory of groups and finite dimensional vector spaces (for example see \cite{kunze1971}) and the theory of Lie groups and the classical groups (see \cite{miller1972}). The presentation will be brief and the reader interested in proofs is referred to the relevant literature as the discussion progresses. Our aim is to show how representation theory of groups -- most notably the Schur/Weyl duality -- can be used to count and construct the group invariant functions on a multi-linear (tensor product) space. We will develop some explicit invariants for the general linear group using a method which is known intuitively to many mathematical physicists and we formalize the technique.

\section{Group representations}
Throughout this thesis we will be interested in the vector spaces $\mathbb{C}^n$ and $\mathbb{R}^n$. Almost all of the results presented will be equally valid whether one considers the complex or real space. Hence, we will simply refer to the vector space $V$, making the distinction between the real and complex case only when confusion may arise. For proofs of theorems that will be presented and further discussion of group representation theory the reader is referred to the excellent texts \cite{hammermesh1964,keown1975,miller1972}.
\begin{define}\label{representation}
A group \textit{representation} $\rho$ on the vector space $V$ is a homomorphism from a group $G$ to the set of invertible, linear transformations $GL(V)$. The image element of $g\in G$ is denoted by $\rho(g)$ and the $dimension$ of the representation is taken to be the dimension of the corresponding vector space.
\end{define}\noindent
A simple example of a group representation is constructed from the symmetric group on $n$ elements, $S_n$, by taking a given group element $\sigma\in S_n$ to simply permute the basis vectors of the $n$ dimensional vector space $V$:
\beqn
\rho(\sigma)e_i:=e_{\sigma i}.\nonumber
\eqn
It is clear that we have $\rho(\sigma\sigma')=\rho(\sigma)\rho(\sigma')$ so that $\rho$ is indeed a homomorphism from $S_n$ to $GL(V)$.\\
We will often be interested in the case where the abstract group is a matrix group such as the general linear group $GL(V)$ which is, of course, defined by its \textit{action} on the vector space $V$. To avoid confusion, we will refer to this representation as the \textit{defining} representation. To increase confusion we will write elements of the defining representation simply as $g$.\\
Given a matrix group $G$, there is always a one-dimensional representation defined by the determinant function:
\beqn
\det: G\rightarrow \mathbb{C}^*,\nonumber
\eqn    
where $\mathbb{C^*}\cong \mathbb{C}\setminus\{0\}$ is the group of multiplications of non-zero complex numbers. The multiplicative property of the determinant:
\beqn
\det(g_1g_2)=\det(g_1)\det(g_2),\nonumber
\eqn
ensures that the determinant function defines a group homomorphism.
\begin{define}
A subspace $U\subseteq V$ is \textit{invariant} under the group representation $\rho$ if for all $u\in U$ it follows that $\rho(g)u\in U$ for all $g\in G$.
\end{define}\noindent
The notion of invariant subspaces allows us to break a given representation into its essential parts. That is, we can simplify the representation by considering its action upon the invariant subspaces alone.
\begin{define}\label{def:reps}
A representation is \textit{reducible} if there exists a non-trivial invariant subspace $U$. An \textit{irreducible} representation is one which has no non-trivial invariant subspaces. A representation is \textit{decomposable} if there exist non-trivial invariant subspaces $U$ and $W$ such that $V\cong U\oplus W$, and \textit{indecomposable} otherwise. A representation is \text{completely reducible} if whenever there exists a non-trivial invariant subspace $U$, then there exists a second non-trivial invariant subspace $W$ such that $V\cong U\oplus W$. 
\end{define}\noindent
The matrix interpretation of a completely reducible representation is that there exists a basis where the matrix representation of each group element takes on a block-diagonal form. We will be exclusively interested in \textit{integral} representations of the general linear group and its subgroups. Integral representations are those in which the entries of the representation matrix are polynomials in the matrix entries of $GL(V)$ with respect to a particular basis. The integral representations of $GL(V)$ are completely reducible \cite{keown1975}.
\begin{define}
The representations $\rho_1$ and $\rho_2$ are said to be $equivalent$ if there exists an invertible linear transformation $S$ on $V$ such that 
\beqn
S\rho_1(g)S^{-1}=\rho_2(g)\nonumber
\eqn
for all $g\in G$.
\end{define}\noindent
From these considerations we can conclude that a given integral representation of the general linear group can be decomposed as
\beqn
\rho=\bigoplus_a\rho_a,\nonumber
\eqn
where each $\rho_a$ is an irreducible representation.
\subsection{Group characters}
\begin{define}
The \textit{character} of a representation $\rho$ is defined as the trace function: 
\beqn\chi(g)=tr(\rho(g)).\nonumber
\eqn
\end{define}\noindent
It follows immediately that the character is unaffected by similarity transformations:
\beqn
tr(S\rho(g)S^{-1})=tr(\rho(g) S^{-1}S)=tr(\rho(g)),\nonumber
\eqn
and is hence the same for equivalent representations. \\
The problem of classifying irreducible representations reduces to identifying the characters. Although the following result is valid only for finite groups, we will see that understanding the representation theory of $S_n$ (a finite group) is crucial to constructing the irreducible representations of $GL(V)$ (an infinite group). 
\begin{remark}
For a finite group, the number of non-equivalent irreducible representations of a group $G$ is equal to the number of conjugacy classes of G.
\end{remark}\noindent
For example the conjugacy classes of the symmetric group can be found by considering the cycle notation which presents an element of $S_n$ as a product of disjoint cycles. The lengths of these cycles adds to $n$ and hence we get the well known result that the conjugacy classes of $S_n$ are labelled by the partitions of $n$. (We will discuss partitions in more detail in the next section.) To illustrate this, consider that any element of the symmetric group can be written in the following form:
\beqn
\sigma=(i_1i_2\ldots i_{\alpha_1})(j_1j_2\ldots j_{\alpha_2})\ldots (l_1l_2\ldots l_{\alpha_p}).\nonumber
\eqn
This element belongs to the conjugacy class which is specified by the partition $\{\alpha_1,\alpha_2,\ldots ,\alpha_p\}$ where $\alpha_1+\alpha_2+\ldots +\alpha_p=n$. The fundamental result follows:
\begin{remark}
The irreducible representations of the symmetric group $S_n$ can be labelled by the partitions of $n$.
\end{remark}\noindent
For example we consider the representation on the $n$-dimensional vector space $V$ of the symmetric group $S_n$ defined, as above, by
\beqn
\rho(\sigma) e_i=e_{\sigma i}.\nonumber
\eqn
Introducing the change of basis
\beqn\label{zbasis}
z_0&=\fra{1}{\sqrt{n}}\sum_{i=1}^ne_i,\\
z_a&=\fra{1}{\sqrt{a(a+1)}}\sum_{i=1}^a(e_i-\sqrt{a}e_{a+1}),\quad a=1,2,\ldots ,n-1.
\eqn
It is clear that $z_0$ spans a one-dimensional invariant subspace
\beqn
\rho(\sigma)z_0=\fra{1}{\sqrt{n}}\sum_{i=1}^ne_{\sigma i}=z_0;\nonumber
\eqn
and we have
\beqn
\rho(\sigma) z_a=\fra{1}{\sqrt{a(a+1)}}\sum_{i=1}^a(e_{\sigma i}-\sqrt{a}e_{\sigma(a+1)}),\nonumber
\eqn
which itself belongs to the span of $\{z_1,z_2,\ldots ,z_{n-1}\}$ which is consequently a complementary invariant space. To prove this consider the standard inner product:
\beqn
(e_i,e_j):=\delta_{ij}.\nonumber
\eqn
and show that
\beqn
(\rho(\sigma)z_a,z_0)=0,\quad\forall\sigma\in S_n.\nonumber
\eqn
The representation of the symmetric group on the subspace $z_0$ corresponds to the partition of $n$ consisting of a single element: $\{n\}$.\\
Another one-dimensional representation of the symmetric group can be constructed by taking the sign of the permutation
\beqn
sgn(\sigma)=\pm 1,\nonumber
\eqn
with the representation space $\mathbb{C}$. This representation corresponds to the partition $\{1,1,\ldots ,1\}$ with $1+1+\ldots +1=n$.
\subsection{Tensor product}
The \textit{dual} of the vector space, $V$, is denoted as $V^*$ and defined to be the set of linear functionals $\{f: V\rightarrow \mathbb{C}\}$:
\beqn
f(cv)&=cf(v),\\
f(v+v')&=f(v)+f(v'),\nonumber
\eqn
for all $c\in \mathbb{C}$ and $v,v'\in V$. Of course $V^*$ itself forms a vector space and we use the basis $\xi_1,\xi_2,\ldots ,\xi_n$ such that $\xi_i(e_j)=\delta_{ij}$. Since $V$ and $V^*$ are complex vector spaces of identical dimension they must be isomorphic and we define the linear functional $\overline{v}$ as
\beqn
\overline{v}(u)=\sum_{i=1}^n v_{i}^*u_{i},\nonumber
\eqn
so that
\beqn
\overline{v}=\sum_{i=1}^n v_{i}^*\xi_{i}.\nonumber
\eqn
 With these definitions in hand we consider $bi$-$linear$ functionals on the ordered product of two vector spaces $V_1$ and $V_2$ with bases $\{e_i^{(1)}\}$ and $\{e_j^{(2)}\}$ respectively. Such functionals map $V_1\times V_2$ to $\mathbb{C}$ and satisfy
\beqn
f(cv_1,v_2)&=cf(v_1,v_2)=f(v_1,cv_2),\\
f(v_1+v_1',v_2)&=f(v_1,v_2)+f(v_1',v_2),\\
f(v_1,v_2+v_2')&=f(v_1,v_2)+f(v_1,v_2'),\nonumber
\eqn 
for all $c\in \mathbb{C},$ $v_1,v_1'\in V_1$ and $v_2,v_2'\in V_2$. Again this set of functionals forms a vector space which we denote as $(V_1\otimes V_2)^*$ with basis given by the set of functionals $\xi_i^{(1)}\otimes \xi_j^{(2)}$ defined as
\beqn
\xi_i^{(1)}\otimes \xi_j^{(2)}(e_k^{(1)},e_l^{(2)}):=\xi_i(e_k^{(1)})\xi_j(e_l^{(2)}).\nonumber
\eqn
From which it follows that the bi-linear functional $f$ can be written as
\beqn
f=\sum_{i,j} f_{ij}\xi_i^{(1)}\otimes \xi_j^{(2)},\nonumber
\eqn
where $f_{ij}=f(e_i^{(1)},e_j^{(2)})$. From this we can induce the definition of the \textit{tensor product} of $V_1$ and $V_2$ to be the vector space $V_1\otimes V_2$. A given element $\psi\in V_1\otimes V_2$ is referred to as a tensor and can be expressed uniquely in the form
\beqn
\psi=\sum_{i,j}\psi_{ij}e^{(1)}_{i}\otimes e^{(2)}_{j}.\nonumber
\eqn  
This process can be iterated to the tensor product of multiple vector spaces $\mathcal{H}=V_1\otimes V_2 \otimes \ldots \otimes V_m$ where a given element $\psi\in \mathcal{H}$ can be expressed as 
\beqn
\psi=\sum_{i_1,i_2,\ldots ,i_m}\psi_{i_1i_2\ldots i_m}e^{(1)}_{i_1}\otimes e^{(2)}_{i_2}\otimes\ldots \otimes e^{(m)}_{i_m}.\nonumber
\eqn
The tensor product space satisfies the axioms of a vector space with addition and scalar multiplication defined in the obvious way:
\beqn
c\cdot \psi&=\sum_{i_1,i_2,\ldots ,i_m}c \psi_{i_1i_2\ldots i_m}e^{(1)}_{i_1}\otimes e^{(2)}_{i_2}\otimes\ldots \otimes e^{(m)}_{i_m},\nonumber\\
\psi+\varphi&=\sum_{i_1,i_2,\ldots ,i_m}(\psi_{i_1i_2\ldots i_m}+\varphi_{i_1i_2\ldots i_m})e^{(1)}_{i_1}\otimes e^{(2)}_{i_2}\otimes\ldots \otimes e^{(m)}_{i_m}.\nonumber
\eqn
When one is taking the tensor product of a single vector space we use the notation
\beqn
V^{\otimes m}:=V\otimes V\otimes\ldots \otimes V.\nonumber
\eqn
Again, $\mathcal{H}:=V^{\otimes m}$ must be isomorphic to $\mathcal{H}^*\cong (V^*)^{\otimes m}$ and we define
\beqn
\overline{\psi}=\sum_{i_1,i_2,\ldots ,i_m}\psi_{i_1i_2\ldots i_m}^*\xi_{i_1}\otimes\xi_{i_2}\otimes\ldots \otimes\xi_{i_m},\nonumber
\eqn
so that
\beqn
\overline{\psi}(\varphi)=\sum_{i_1,i_2,\ldots ,i_m}\psi_{i_1i_2\ldots i_m}^*\varphi_{i_1i_2\ldots i_m}.\nonumber
\eqn
\subsection{Group action on a tensor product space}
Given a set of representations of a group 
\beqn
\rho_a: G\rightarrow GL(V_a),\qquad a=1,2,\ldots ,m,\nonumber
\eqn
it is possible to construct a new representation $\rho$ by taking the tensor product
\beqn
\mathcal{H}=V_1\otimes V_2 \otimes\ldots  \otimes V_m\nonumber
\eqn
and define the \textit{tensor product representation} on the vector space $\mathcal{H}$ to act as:
\beqn
\rho(g)\psi:&=\rho_1(g)\otimes \rho_2(g) \otimes\ldots \otimes \rho_m(g)\psi,\nonumber\\
&=\sum_{i_1,i_2,\ldots ,i_m} \psi_{i_1i_2\ldots i_m} \rho_1(g)e^{(1)}_{i_1}\otimes \rho_2(g)e^{(2)}_{i_2}\otimes\ldots \otimes \rho_m(g)e^{(m)}_{i_m}.
\eqn
In contrast to this we consider another important case which occurs when we have the direct (cartesian) product of $m$ groups: 
\beqn
G=G_1\times G_2\times\ldots \times G_m,\nonumber
\eqn
with representations $\rho_1,\rho_2,\ldots ,\rho_m$ and associated representation spaces
\beqn
V_1,V_2,\ldots .,V_m.\nonumber
\eqn
It is again possible to define a representation $\bar{\rho}$ on $\mathcal{H}$ as 
\beqn
\bar{\rho}(g)\psi=\bar{\rho}(g_1\times g_2\times\ldots \times g_m) \psi=\rho_1(g_1)\otimes \rho_2(g_2)\otimes \ldots \otimes \rho_m(g_m) \psi.\nonumber
\eqn
\footnote{Interestingly, in quantum physics the appropriate description of a multi-particle system is given by taking the tensor product of different representations of a single group, such as the orthogonal or Lorentz groups, where the choice of each representation is fixed by the individual particle types. Whereas in the case of stochastic models of phylogenetics the reverse is the case; the system is described by taking the group action on the tensor product space as the direct product of the defining representation of the Markov semigroup.} For future use we define the notation
\beqn
\times^m G:&=G\times G\times\ldots \times G\nonumber\\
\otimes^mg:&=g\otimes g \otimes\ldots \otimes g. 
\eqn
Presently we will recall the character theory of the general linear group to enable us to decompose such representations into their irreducible parts.
\section{Irreducible representations of the general linear group}
It is well known from group representation theory that the finite-dimensional irreducible representations of the general linear and the symmetric group can be put into a correspondence. This result is known as the Schur/Weyl duality. As we saw above, the irreducible representations of the symmetric group on $n$ elements can be labelled by the partitions of $n$. Additionally, there exist algorithms for explicitly constructing these irreducible representations once a partition has been specified. Here we will show how the irreducible representations of the general linear group on $V$ occur as subspaces of the tensor product space $V^{\otimes m}$. These projections are constructed using operators known as \textit{Young's operators} which are computed from the partitions of $m$.
\subsection{Partitions}
A finite sequence of positive integers 
\beqn
\lambda=\{\lambda_1,\lambda_2,\ldots \}\nonumber
\eqn
with $\lambda_1\geq \lambda_2\geq\ldots$, is an (ordered) \textit{partition} of the integer $n$ if the \textit{weight} of the partition,
\beqn
|\lambda|:=\lambda_1+\lambda_2+\ldots \nonumber,
\eqn
satisfies $|\lambda|=n$.   
\\
It is usual to use a notation which indicates the number of times each integer occurs as a part:
\beqn
\lambda=\{\ldots ,r^{m_r},\ldots ,2^{m_2},1^{m_1}\}\nonumber
\eqn
so that $m_i$ of the parts of $\lambda$ are equal to $i$. It is useful to represent a given partition as a \textit{Ferrers} \textit{diagram} by drawing a row of squares for each part of the partition, and placing these rows upon each other sequentially such that the rows decrease in length down the page. For example the partition $\lambda=\{5,3^2,2,1\}$ is represented by:
\begin{figure}[h]
\centering
\pspicture[](0,0)(0,3)
\psline(-1.25,2.5)(1.25,2.5)
\psline(-1.25,2)(1.25,2)
\psline(-1.25,1.5)(.25,1.5)
\psline(-1.25,1)(.25,1)
\psline(-1.25,.5)(-.25,.5)
\psline(-1.25,0)(-.75,0)
\psline(-1.25,0)(-1.25,2.5)
\psline(-.75,0)(-.75,2.5)
\psline(-.25,.5)(-.25,2.5)
\psline(.25,1)(.25,2.5)
\psline(.75,2)(.75,2.5)
\psline(1.25,2)(1.25,2.5)
\endpspicture 
\end{figure}
\\
\begin{define}
A \textit{Young tableau}, $T$, of shape $\lambda$ with $|\lambda|=n$ is an assignment of the integers $1,2,\ldots ,n$ to a Ferrers diagram such that the rows and columns are strictly increasing. A \textit{semi-standard tableau}, $T'$, requires that only the rows need to be increasing.
\end{define}\noindent
For example, the canonical Young tableau of shape $\{5,3^2,2,1\}$ is:
\begin{figure}[h]
\centering
\pspicture[](0,0)(0,3)
\psline(-1.25,2.5)(1.25,2.5)
\psline(-1.25,2)(1.25,2)
\psline(-1.25,1.5)(.25,1.5)
\psline(-1.25,1)(.25,1)
\psline(-1.25,.5)(-.25,.5)
\psline(-1.25,0)(-.75,0)
\psline(-1.25,0)(-1.25,2.5)
\psline(-.75,0)(-.75,2.5)
\psline(-.25,.5)(-.25,2.5)
\psline(.25,1)(.25,2.5)
\psline(.75,2)(.75,2.5)
\psline(1.25,2)(1.25,2.5)
\rput(-1,2.25){1}
\rput(-.5,2.25){2}
\rput(0,2.25){3}
\rput(.5,2.25){4}
\rput(1,2.25){5}
\rput(-1,1.75){6}
\rput(-.5,1.75){7}
\rput(0,1.75){8}
\rput(-1,1.25){9}
\rput(-.5,1.25){10}
\rput(0,1.25){11}
\rput(-1,.75){12}
\rput(-.5,.75){13}
\rput(-1,.25){14}
\endpspicture 
\end{figure}
\\
while a semi-standard tableau of the same shape is:
\begin{figure}[h]
\centering
\pspicture[](0,0)(0,3)
\psline(-1.25,2.5)(1.25,2.5)
\psline(-1.25,2)(1.25,2)
\psline(-1.25,1.5)(.25,1.5)
\psline(-1.25,1)(.25,1)
\psline(-1.25,.5)(-.25,.5)
\psline(-1.25,0)(-.75,0)
\psline(-1.25,0)(-1.25,2.5)
\psline(-.75,0)(-.75,2.5)
\psline(-.25,.5)(-.25,2.5)
\psline(.25,1)(.25,2.5)
\psline(.75,2)(.75,2.5)
\psline(1.25,2)(1.25,2.5)
\rput(-1,2.25){1}
\rput(-.5,2.25){1}
\rput(0,2.25){1}
\rput(.5,2.25){2}
\rput(1,2.25){3}
\rput(-1,1.75){3}
\rput(-.5,1.75){4}
\rput(0,1.75){4}
\rput(-1,1.25){4}
\rput(-.5,1.25){5}
\rput(0,1.25){6}
\rput(-1,.75){6}
\rput(-.5,.75){6}
\rput(-1,.25){7}
\endpspicture. 
\end{figure}
\\
\begin{define}
The ring of \textit{symmetric functions}, $\Lambda_n=\mathbb{Z}[x_1,\ldots ,x_n]^{S_n}$, is the set of polynomials in $n$ independent variables $x_1,\ldots ,x_n$ which are invariant under the representation of $S_n$ defined by permutations of the variables.
\end{define}\noindent
That is, $f$ is a symmetric function if and only if:
\beqn
f(x_1,x_2,\ldots ,x_n)=f(x_{\sigma 1},x_{\sigma 2},\ldots ,x_{\sigma n}),\quad\forall \sigma\in S_n.\nonumber
\eqn
It is clear that $\Lambda_n$ is a graded ring:
\beqn
\Lambda_n=\oplus_{d\geq 0}\Lambda_n^d\nonumber
\eqn
where $\Lambda_n^d\subset \Lambda_n$ consists of the homogeneous symmetric polynomials of degree $d$. 
Various bases exist for the ring of symmetric functions (see \cite{macdonald1979}). The basis which will be of use to us is given by the Schur functions.
\subsection{The Schur functions}
For a given partition $\lambda$ define the monomial $x^{\lambda}=x_1^{\lambda_1}x_2^{\lambda_2}\ldots x_{n}^{\lambda_n}$. Consider the polynomial which is obtained by anti-symmetrizing:
\beqn
a_{\lambda}=a_{\lambda}(x_1,\ldots ,x_n)=\sum_{\sigma\in S_n}sgn(\sigma)\sigma(x^{\lambda})\nonumber,
\eqn 
where
\beqn
\sigma(x^{\lambda}):=x_{\sigma 1}^{\lambda_1}x_{\sigma 2}^{\lambda_2}\ldots x_{\sigma{n}}^{\lambda_n}.\nonumber
\eqn
By considering the partition $\delta=\{n-1,n-2,\ldots ,1\}$ it follows that
\beqn
a_{\delta}=\prod_{1\leq i<j\leq n}(x_i-x_j),\nonumber
\eqn
which is called the \textit{Vandermonde determinant}. The Schur functions are then defined as the quotient
\beqn
s_{\lambda}=s_{\lambda}(x_1,x_2,\ldots ,x_n)=a_{\lambda+\delta}/a_{\delta}\nonumber,
\eqn
which is clearly symmetric. A more intuitive and constructive way of defining the Schur functions is to take:
\beqn
s_{\lambda}=\sum_{T'}x^{T'},\nonumber
\eqn
where the summation is over all semi-standard $\lambda$ tableaux $T'$. For example, for $\lambda=\{2,1\}$ the semi-standard tableaux are displayed in Figure \ref{fig:semi}.
\begin{figure}[t]
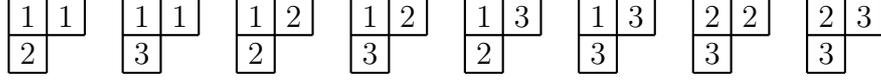

\centering
\pspicture[](0,0)(0,1)
\psline(-.5,0)(-.5,1)
\psline(-.5,1)(.5,1)
\psline(0,0)(0,1)
\psline(-.5,0)(0,0)
\psline(.5,.5)(.5,1)
\psline(-.5,.5)(.5,.5)
\rput(-.25,.75){1}
\rput(.25,.75){1}
\rput(-.25,.25){2}
\endpspicture
\pspicture[](-1.5,0)(0,1.5)
\psline(-.5,0)(-.5,1)
\psline(-.5,1)(.5,1)
\psline(0,0)(0,1)
\psline(-.5,0)(0,0)
\psline(.5,.5)(.5,1)
\psline(-.5,.5)(.5,.5)
\rput(-.25,.75){1}
\rput(.25,.75){1}
\rput(-.25,.25){3}
\endpspicture 
\pspicture[](-1.5,0)(0,1.5)
\psline(-.5,0)(-.5,1)
\psline(-.5,1)(.5,1)
\psline(0,0)(0,1)
\psline(-.5,0)(0,0)
\psline(.5,.5)(.5,1)
\psline(-.5,.5)(.5,.5)
\rput(-.25,.75){1}
\rput(.25,.75){2}
\rput(-.25,.25){2}
\endpspicture 
\pspicture[](-1.5,0)(0,1.5)
\psline(-.5,0)(-.5,1)
\psline(-.5,1)(.5,1)
\psline(0,0)(0,1)
\psline(-.5,0)(0,0)
\psline(.5,.5)(.5,1)
\psline(-.5,.5)(.5,.5)
\rput(-.25,.75){1}
\rput(.25,.75){2}
\rput(-.25,.25){3}
\endpspicture 
\pspicture[](-1.5,0)(0,1.5)
\psline(-.5,0)(-.5,1)
\psline(-.5,1)(.5,1)
\psline(0,0)(0,1)
\psline(-.5,0)(0,0)
\psline(.5,.5)(.5,1)
\psline(-.5,.5)(.5,.5)
\rput(-.25,.75){1}
\rput(.25,.75){3}
\rput(-.25,.25){2}
\endpspicture 
\pspicture[](-1.5,0)(0,1.5)
\psline(-.5,0)(-.5,1)
\psline(-.5,1)(.5,1)
\psline(0,0)(0,1)
\psline(-.5,0)(0,0)
\psline(.5,.5)(.5,1)
\psline(-.5,.5)(.5,.5)
\rput(-.25,.75){1}
\rput(.25,.75){3}
\rput(-.25,.25){3}
\endpspicture 
\pspicture[](-1.5,0)(0,1.5)
\psline(-.5,0)(-.5,1)
\psline(-.5,1)(.5,1)
\psline(0,0)(0,1)
\psline(-.5,0)(0,0)
\psline(.5,.5)(.5,1)
\psline(-.5,.5)(.5,.5)
\rput(-.25,.75){2}
\rput(.25,.75){2}
\rput(-.25,.25){3}
\endpspicture 
\pspicture[](-1.5,0)(0,1.5)
\psline(-.5,0)(-.5,1)
\psline(-.5,1)(.5,1)
\psline(0,0)(0,1)
\psline(-.5,0)(0,0)
\psline(.5,.5)(.5,1)
\psline(-.5,.5)(.5,.5)
\rput(-.25,.75){2}
\rput(.25,.75){3}
\rput(-.25,.25){3}
\endpspicture 
\caption{Semi-standard tableaux}
\label{fig:semi}
\end{figure}
In this case each tableau corresponds to a monomial $x^T$ to give
\beqn
s_{21}(x_1,x_2,x_3)&=x_1^2x_2+x_1^2x_3+x_1x_2^2+2x_1x_2x_3+x_1x_3^2+x_2^2x_3+x_2x_3^2\nonumber
\eqn
which is easily seen to be a symmetric polynomial.
\subsection{Group characters of $GL(n)$}
For a given matrix $g\in GL(n)$ it is possible to use the Jordan decomposition to put it in upper triangular form and hence the character is simply the sum of the eigenvalues:
\beqn
\chi(g)=tr(g)=x_1+x_2+\ldots +x_n.\nonumber
\eqn
This corresponds to the Schur function
\beqn
s_{\{1\}}(x_1,\ldots ,x_n)=x_1+x_2+\ldots +x_n.\nonumber
\eqn
By considering the tensor product representation of $GL(V)$ on $V\otimes V$ we have 
\beqn
V^{\{2\}}:=\{\psi^{(s)}|\psi^{(s)}_{i_1i_2}=\psi_{i_2i_1}\},\qquad
V^{\{1^2\}}:=\{\psi^{(a)}|\psi^{(a)}_{i_1i_2}=-\psi_{i_2i_1}\},\nonumber
\eqn
as irreducible subspaces known as the \textit{symmetric} and \textit{anti-symmetric} tensors with dimensions $\frac{1}{2}n(n+1)$ and $\frac{1}{2}n(n-1)$ respectively. We have
\beqn
\psi=\fra{1}{2}(\psi^{(s)}+\psi^{(a)}),\nonumber
\eqn
where 
\beqn
\psi^{(s)}_{i_1i_2}=\fra{1}{2}(\psi_{i_1i_2}+\psi_{i_2i_1}),\qquad\psi^{(a)}_{i_1i_2}=\fra{1}{2}(\psi_{i_1i_2}-\psi_{i_2i_1}),\nonumber
\eqn
so that the decomposition of $V\otimes V$ under the action of $GL(n)$ into irreducible subspaces is given by
\beqn
V\otimes V=V^{\{2\}}\oplus V^{\{1^2\}}.\nonumber
\eqn
Now suppose we take the group element $g\in GL(n)$. It follows from an elementary calculation that the character of this group element on the representation $V^{\otimes 2}$ is simply the product:
\beqn
\chi(g\otimes g)=(x_1+x_2+\ldots +x_n)(x_1+x_2+\ldots +x_n).\nonumber
\eqn
In terms of the Schur functions it follows that we have the decomposition
\beqn
\chi(g\otimes g)&=s_{\{1\}}(x)s_{\{1\}}(x),\nonumber\\
&=s_{\{2\}}(x)+s_{\{1^2\}}(x)
\eqn
where
\beqn
s_{\{2\}}(x)&=(x_1^2+x_1x_2+x_1x_3+\ldots +x_2^2+x_2x_3+\ldots +x_n^2),\nonumber\\
s_{\{1^2\}}(x)&=(x_1x_2+x_1x_3+\ldots +x_2x_3+\ldots +x_{n-1}x_n).
\eqn
Thus we see that the decomposition of the tensor product representation into irreducible parts can be inferred by using the Schur functions as a basis for the ring of symmetric functions. This is the archetypal example from physics and leads to the full Schur/Weyl duality which allows us to classify the irreducible representations of $GL(n)$ (and its subgroups) by simply using the character formulas and the Schur functions.
\subsection{The Schur/Weyl duality}
In this section we will construct the Schur/Weyl duality which states that the irreducible representations of the general linear group and that of the symmetric group can be put into correspondence.
\begin{remark}
If $V$ decomposes into the direct sum $V=U\oplus W$ where $U$ and $W$ are invariant subspaces under the group representation $\rho$, then the projection operator $P$, defined by $PV\cong U$, satisfies
\beqn\label{project}
P\rho(g)=\rho(g)P,\quad \forall g\in G,
\eqn
and similarly for the orthogonal projection $(1-P)$. Conversely, if $P$ is a projection operator satisfying (\ref{project}) then the subspace it projects to is invariant under $\rho$.
\end{remark}\noindent
Consider the representation of the symmetric group on $V^{\otimes m}$ defined by
\beqn
\sigma (e_{i_1}\otimes e_{i_2}\otimes\ldots \otimes e_{i_m}):=e_{i_{\sigma 1}}\otimes e_{i_{\sigma 2}}\otimes\ldots \otimes e_{i_{\sigma m}}.\nonumber
\eqn
It should be clear that the action of any such element of the symmetric group will commute with the tensor product representation of $GL(n)$. In addition to this the \textit{algebra} generated from this action will commute with $GL(n)$ and hence can be used to construct projection operators which satisfy (\ref{project}). Presently we will discuss how to construct such projection operators such that the corresponding invariant subspaces are in fact irreducible.\\
Consider a Young tableau with shape $\lambda$ and $|\lambda|=m$. Consider the permutations $p$ which interchange the integers in the same row, and, conversely, permutations $q$ which interchange numbers in the same column. In the algebra of the symmetric group action defined above, consider the quantities
\beqn
P&=\sum_{p}p,\nonumber\\
Q&=\sum_{q}sgn(q)q.\nonumber
\eqn 
The Young operator corresponding to the standard tableau $T$ is then defined to be
\beqn
Y=QP,\nonumber
\eqn
and we have the fundamental result:
\begin{remark}
For a given partition $\lambda$, $Y$ projects onto an irreducible subspace of $V^{\otimes m}$ under the tensor product representation of $GL(n)$. Young tableaux of the same shape label equivalent representations. 
\end{remark}\noindent
Now suppose $Y_{\lambda}$ is the Young operator corresponding to the partition $\lambda$. We define the subspace
\beqn
V^{\lambda}:=Y_{\lambda}V^{\otimes m}.\nonumber
\eqn
It is possible to prove that the group character of the tensor product representation of the general linear group on the subspace $V^{\lambda}$ is none other than the Schur function $s_{\lambda}(x_1,x_2,\ldots ,x_n)$.\\\noindent
For example we consider the standard tableau:
\begin{figure}[h]
\centering 
\pspicture[](0,0)(3,1)
\psline(-.5,0)(-.5,1)
\psline(-.5,1)(.5,1)
\psline(0,0)(0,1)
\psline(-.5,0)(0,0)
\psline(.5,.5)(.5,1)
\psline(-.5,.5)(.5,.5)
\rput(-.25,.75){1}
\rput(.25,.75){3}
\rput(-.25,.25){2}
\endpspicture .
\end{figure}
\\
With corresponding Young operator given by 
\beqn
P&=e+(13),\nonumber\\
Q&=e-(12),\\
Y&=e+(13)-(12)-(123).
\eqn
We also note that the dimension of the invariant subspaces are given by setting the characteristic values in the Schur function equal to the identity:
\beqn
\textit{dimension of $V^\lambda$}=s_{\lambda}(1,1,\ldots ,1).\nonumber
\eqn 
\subsection{More representations}
From this construction we can build more representations such as
\beqn
V^\mu\otimes V^\nu,\nonumber
\eqn
with group character which corresponds to the \textit{outer} product of two Schur functions which is defined as the pointwise product:
\beqn
s_{\mu}s_{\nu}(x):=s_{\mu}(x)s_{\nu}(x)=\sum_{\lambda}c^{\lambda}_{\mu\nu}s_{\lambda},\nonumber
\eqn
where $|\lambda|=|\mu|+|\nu|$ and the $c^{\lambda}_{\mu\nu}$ are integer coefficients which can be determined by the Littlewood-Richardson rule \cite{littlewood1940,macdonald1979}.
\\
Another way of constructing representations is to consider
\beqn
(V^\mu)^\nu\nonumber.
\eqn
The group character of this representation is given by another type of multiplication of Schur functions known as the \textit{plethysm} (defined formally in Macdonald \cite{macdonald1979}). Here we use Young's tableaux to give a constructive definition. Recall that we have
\beqn
s_{\mu}(x)=\sum_{T'}x^{T'}\nonumber
\eqn
which is a summation of monomials in $x_1,x_2,\ldots ,x_n$. If there are $m$ such monomials in $s_{\mu}(x)$ and these are denoted by $y_i$, $1\leq i\leq m$, then the plethysm is given by
\beqn
s_{\lambda}[s_{\mu}](x)=s_{\lambda}(y)=\sum_{T'}y^{T'}.\nonumber
\eqn
The plethysm $s_{\lambda}[s_{\mu}]$ can be interpreted as giving the character of the representation $(V^\mu)^\lambda$. That is we take $V^\mu$ as the defining representation and symmetrize this representation with $\lambda$.
\\
Finally the \textit{inner} product of two Schur functions is defined as
\beqn
s_{\mu}(x)\ast s_{\nu}(x)=\sum_{\lambda}\gamma^{\lambda}_{\mu\nu}s_{\lambda}(x),\nonumber
\eqn
where $|\mu|=|\nu|=|\lambda|=n$ and the $\gamma^{\lambda}_{\mu\nu}$ are the integer multiplicities of the $\lambda$ representation of $S_n$ occurring in the decomposition of the tensor product of the $\mu$ and $\nu$ representations of $S_n$. The inner product comes into play if we wish to compute the character of $GL(n)\times GL(n')$ on $V^\lambda \otimes V^{\lambda'}$ with $|\lambda|=m,|\lambda'|=m'$. The character of this representation is $s_{\lambda}(x)s_{\lambda'}(y)$ where $x_1,x_2,\ldots ,x_n$ and $y_1,y_2,\ldots ,y_n'$ are the eigenvalues of the relevant group elements in $GL(n)$ and $GL(n')$ respectively. The decomposition of characters is given by the formula
\beqn
s_{\lambda}(x)s_{\lambda'}(y)=\sum_{\rho}\gamma^{\rho}_{\lambda\lambda'}s_{\rho}(xy),\nonumber
\eqn
where $(xy)=(x_1y_1,x_1y_2,\ldots ,x_2y_1,\ldots ,x_ny_n)$ \cite{macdonald1979}.\\
We will often write the Schur function $s_{\lambda}$ simply as $\{\lambda\}$ and the plethysm will sometimes be written as
\beqn
s_{\lambda}[s_{\mu}]=\{\mu\}^{\otimes \{\lambda\}}.\nonumber
\eqn
In practice we compute Schur multiplications by using the group theory software \textbf{Schur} \cite{schur}. For further discussion of Schur functions and their various multiplications see \cite{baker1994,carvalho2001,fauser2005,fauser2005a,littlewood1940}.
\\
As an example consider the defining representation of $GL(n)$ on the tensor product $V^{\otimes m}$. That is
\beqn
\psi\rightarrow g\otimes g\otimes\ldots \otimes g\psi:=\otimes^mg\psi,\nonumber
\eqn 
for $\psi\in V^{\otimes m}$, $g\in GL(n)$. The character of this representation is given by the pointwise product of $m$ copies of $s_{\{1\}}(x)$ and can be decomposed into irreducible characters by using the Littlewood-Richardson coefficients $c^{\lambda}_{\mu\nu}$. In the case where $m=4$, \textbf{Schur} gives
\beqn
s_{\{1\}}(x)s_{\{1\}}(x)s_{\{1\}}(x)&s_{\{1\}}(x)=\\\nonumber
s_{\{4\}}(x)&+3s_{\{31\}}(x)+2s_{\{2^2\}}(x)+3s_{\{21^2\}}(x)+s_{\{1^4\}}(x).\nonumber
\eqn
This tells us that under the action of $GL(n)$ the tensor product $V^{\otimes 4}$ decomposes into irreducible subspaces:
\beqn
V^{\otimes 4}=V^{\{4\}}+3V^{\{31\}}+2V^{\{2^2\}}+3V^{\{21^2\}}+V^{\{1^4\}},\nonumber
\eqn
where the multiplicities account for the number of legal standard tableaux for each partition. 
\subsection{One-dimensional representations}\label{onedrep}
Recall that the dimension of the irreducible representation $\lambda$ is given by $s_{\lambda}(1,1,\ldots ,1)$. It follows that the one-dimensional representations occur when there is only a single semi-standard tableau with shape $\lambda$. In the case when $V$ is $n$-dimensional it should be clear that the one-dimensional representations occur when we have $\lambda=\{k^n\}$ for some $k$.\\
Consider the character of $GL(n)$ on $V^{\{k^n\}}$:
\beqn
s_{\{k^n\}}(x)&=(x_1x_2\ldots x_n)^k,\nonumber\\
&=\det(g)^k.
\eqn
Thus for any $\eta\in V^{\{k^n\}}$ we have
\beqn
\eta\mapsto\det(g)^k\eta,\nonumber
\eqn
under the $\{k^n\}$ representation of $GL(n)$.

\section{Invariant theory}
Given the defining representation of a group $G$ on a vector space $V$, it is possible to define a representation which acts on the vector space of functions $f:V\rightarrow \mathbb{C}$ as 
\beqn\label{induced}
gf:=f\circ g^{-1}.
\eqn
(It is necessary to take the inverse of the group element to ensure that the induced representation satisfies the properties of a group homomorphism.) An \textit{invariant} with \textit{weight} $k\in\mathbb{N}$ is then defined as any function which satisfies 
\beqn\label{inv}
g^{-1}f=f\circ g=\det(g)^kf.
\eqn
We will be exclusively interested in the case where $f$ is a polynomial in the dual vector space $V^*$ with basis elements $\{\xi_1,\xi_2,\ldots ,\xi_n\}$. In order to generate polynomials in this space multiplication is defined pointwise:
\beqn
\xi_a\xi_b(x):=\xi_a(x)\xi_b(x),\qquad 1\leq a,b\leq n.\nonumber
\eqn
The full set of polynomials generated from this construction is denoted as $\mathbb{C}[V]$. A \textit{homogeneous} polynomial satisfies
\beqn
f\circ c1=c^df,\quad c\in\mathbb{C},\nonumber
\eqn
for some positive integer $d$ which is referred as the \textit{degree}. From elementary considerations it follows that $\mathbb{C}[V]$ has the structure of a graded algebra over the degree:
\beqn
\mathbb{C}[V]=\oplus_d\mathbb{C}[V]_{d},\nonumber
\eqn
where $\mathbb{C}[V]_d$ is the set of homogeneous polynomials of degree $d$.\\
By counting the degree of the various algebraic quantities we see that $d=nk$, and we denote
\beqn
\mathbb{C}[V]_d^G=\{f\in \mathbb{C}[V]_d|f\circ g=\det(g)^kf, \forall g\in G\}.\nonumber
\eqn
Of course we have already studied invariant functions on a finite group! The symmetric functions are none other than the set $\mathbb{C}[V]^{S_n}$ with $k=0$. Another example comes from the classical groups which are defined by imposing invariant functions. For example the orthogonal group $O(n)$ acting on $\mathbb{R}^n$ can be considered to be defined by the invariant function
\beqn
r^2(x)=\sum_{i=1}^n x^2_i.\nonumber
\eqn
Consider the tensor product space $\mathbb{C}^2\otimes \mathbb{C}^2$ with associated group action $GL(2)\times GL(2)$. The following relation holds for any $\psi=\sum_{i,j} \psi_{ij}e_{i}\otimes e_{j}\in \mathbb{C}^2\otimes \mathbb{C}^2$:
\beqn
(\psi'_{11}\psi'_{22}-\psi'_{12}\psi'_{21})=\det{g}(\psi_{11}\psi_{22}-\psi_{12}\psi_{21}),\nonumber
\eqn
where $\psi':=g_1\otimes g_2 \psi$ and $\det{g}=\det{g_1}\det{g_2}$. So $(\psi_{11}\psi_{22}-\psi_{12}\psi_{21})$ is an invariant. 
\subsection{Invariants as irreducible representations}
In this section we will show that the group invariant polynomials $\mathbb{C}[V^{\otimes m}]_d^G$ occur exactly as the one-dimensional representations $V^{\{k^n\}}$ in the decomposition of ${(V^{\otimes m})}^{\{d\}}$ with $md=kn$. As a first step consider a vector space $U$. We establish the vector space isomorphism:
\beqn
U^{\{d\}}\cong\mathbb{C}[U]_d.\nonumber
\eqn
This follows by observing that if $U$ has basis $\{u_i\}$, then $U^{\{d\}}$ consists of all tensors of the form
\beqn
\psi=\sum_{i_1,i_2,\ldots,i_d}\psi_{i_1i_2\ldots i_d}u_{i_1}\otimes u_{i_2}\otimes\ldots \otimes u_{i_d},\nonumber
\eqn
where $\psi_{i_1i_2\ldots i_d}$ is invariant under permutations of indices. Now if $U^*$ has basis $\{\zeta_i\}$, consider an arbitrary element of $\mathbb{C}[U]_d$:
\beqn
f=\sum_{i_1,i_2,\ldots,i_d} f_{i_1i_2\ldots i_d}\zeta_{i_1}\zeta_{i_2}..\zeta_{i_d}.\nonumber
\eqn
Clearly $f_{i_1i_2\ldots i_d}$ is also invariant under permutations of indices. This identification establishes the isomorphism. We define the canonical isomorphism $\omega:U^{\{d\}}\rightarrow \mathbb{C}[U]_d$ as
\beqn\label{omegax}
\omega(\psi)=\sum_{i_1,i_2,\ldots,i_d} \psi_{i_1i_2\ldots i_d}\zeta_{i_1}\zeta_{i_2}\ldots \zeta_{i_d},
\eqn
with inverse
\beqn
\omega^{-1}(f)=\sum_{i_1,i_2,\ldots,i_d} f_{i_1i_2\ldots i_d}u_{i_1}\otimes u_{i_2}\otimes\ldots \otimes u_{i_d}.\nonumber
\eqn
\\
By explicit computation
\beqn\label{omegamap1}
\omega(\otimes^dg^{t}\psi)=\omega(\psi)\circ g=g^{-1}\omega(\psi),
\eqn
and
\beqn\label{omegamap2}
\omega^{-1}(g^{-1}f)=\omega^{-1}(f\circ g)=\otimes^dg^{t}\omega^{-1}(f)
\eqn
for all $g\in GL(U)$, $f\in\mathbb{C}[U]_d$ and $\psi\in U^{\{d\}}$.\\
From these considerations we generalize to the case where $U=V^{\otimes m}$ and establish the main result of this section:
\begin{thm}
Consider integers $m,d,k,n$ with $md=kn$ and label the occurrences of $V^{\{k^n\}}$ in the decomposition of $(V^{\otimes m})^{\{d\}}$ by an integer $a$. It follows that
\beqn
\mathbb{C}[V^{\otimes m}]^{GL(n)}_d\cong\oplus_a V^{\{k^n\}}_a.\nonumber
\eqn
\end{thm}
\begin{proof}
Suppose $f\in\mathbb{C}[V^{\otimes m}]^{GL(n)}_d$. We have
\beqn
\omega^{-1}(g^{-1}f)=\omega^{-1}(\det(g)^kf)&=\det(g)^k\omega^{-1}(f)\nonumber\\
&=\otimes^dg^t\omega^{-1}(f),
\eqn
and hence the representation space $span[\omega^{-1}(f)]$ provides a one-dimensional representation of $GL(n)$. Conversely, suppose that $span[\psi]$ with $\psi\in (V^{\otimes m})^{\{d\}}$ provides a one-dimensional representation of $GL(n)$ such that 
\beqn
\otimes^dg\psi=\det(g)^k\psi.\nonumber 
\eqn
Noting that $\det(g)=\det(g^t)$ it follows that
\beqn
\omega(\otimes^d{g}\psi)=\omega(\det(g)^k\psi)&=\det(g)^k\omega(\psi)\nonumber\\
&=\omega(\otimes^d{g^t}\psi)=g^{-1}\omega(\psi),
\eqn
so we can conclude that $\omega(\psi)\in\mathbb{C}[V^{\otimes m}]^{GL(n)}_d$.
\end{proof}
\subsection{Using Schur functions to count invariants}
By the preceding theorems we conclude the following:
\begin{thm}\label{thm:pleth}
The number of invariants in $\mathbb{C}[V^{\otimes m}]^{GL(n)}_d$ of weight $k$ is equal to the number of occurrences of $\{k^n\}$ in the decomposition of $(\times^m \{1\})^{\otimes \{d\}}$.
\end{thm}\noindent
We now consider the character of $\times^mGL(n)$ on $(V^{\otimes m})^\lambda$:
\beqn\label{directchar}
s_{\lambda}&(x^{(1)}x^{(2)}\ldots x^{(m)})\\&=\sum_{\stackrel{\mu_1,\ldots,\mu_m,}{\nu_1,\ldots,\nu_{m-1}}}\gamma_{\mu_1\nu_1}^\lambda\gamma_{\mu_2\nu_2}^{\nu_1}\ldots \gamma_{\mu_{m-1}\mu_m}^{\nu_{m-2}}s_{\mu_1}(x^{(1)})s_{\mu_2}(x^{(2)})\ldots s_{\mu_m}(x^{(m)}),
\eqn
where $(x^{(1)}x^{(2)}\ldots x^{(m)})=(x^{(1)}_{i_1}x^{(2)}_{i_2}\ldots x^{(m)}_{i_m})_{1\leq i_a\leq n}$. Now each term in (\ref{directchar}) is an irreducible character
\beqn\label{directirrep}
s_{\sigma_1}(x^{(1)})s_{\sigma_2}(x^{(2)})\ldots s_{\sigma_{m}}(x^{(m)}),\nonumber
\eqn
with $|\sigma_i|=|\lambda|$ and multiplicity
\beqn
q(\sigma_1,\sigma_2,\ldots ,\sigma_m;\lambda):=\sum_{\nu_1,\ldots,\nu_{m-2}}\gamma_{\sigma_1\nu_1}^{\lambda}\gamma_{\sigma_2\nu_2}^{\nu_1}\ldots \gamma_{\sigma_{m-1}\sigma_{m}}^{\nu_{m-2}}.\nonumber
\eqn
From the definition of the inner product
\beqn
q(\sigma_1,\sigma_2,\ldots ,\sigma_m;\lambda)=\{\textit{multiplicity of $\lambda$ in $\sigma_1\ast \sigma_2\ldots \ast\sigma_m$}\}.\nonumber
\eqn
The dimension of each of the irreducible representations (\ref{directirrep}) is equal to the product of the dimensions of each component irreducible representations. To identify invariant functions we are led to the following theorem:
\begin{thm}\label{thm:inner}
The number of weight $k$ invariants in $\mathbb{C}[V^{\otimes m}]^{\times^m GL(n)}_d$ is equal to the number of occurrences of the Schur function $\{d\}$ in $\ast^m\{k^n\}$.
\end{thm}

\section{Invariants of the general linear group}
We have established that any one-dimensional representation of $GL(n)$ occurs as a partition of the form $\{k^n\}$. This is because the columns of the partitions correspond to the anti-symmetrization process of Young operators and it is clear that if we anti-symmetrize $n$ elements $n$ times then there will only be a single independent element remaining. Presently we will present a generic scheme which allows us to generate the exact polynomial form of these representations.
\\
Consider the definition of the determinant of a matrix $g$:
\beqn
\det(g)=\sum_{\sigma\in S_n}sgn(\sigma)g_{1\sigma(1)}g_{2\sigma(2)}\ldots g_{n\sigma(n)}.\nonumber
\eqn 
By defining the (anti-symmetric) Levi-Civita tensor
\beqn
\epsilon:=\sum_{i_1,i_2,\ldots,i_n} \epsilon_{i_1i_2\ldots i_n}e_{i_1}\otimes e_{i_2}\otimes\ldots \otimes e_{i_n},\nonumber
\eqn
where $\epsilon_{\sigma(1)\sigma(2)\ldots \sigma(n)}:=sgn(\sigma)$, it follows that the determinant can be expressed as
\beqn\label{epsdet}
\det(g)={\frac{1}{n!}}\sum_{\stackrel{i_1,i_2,\ldots,i_n,}{j_1,j_2,\ldots,j_n}}^{} g_{i_1j_1}g_{i_2j_2}\ldots g_{i_nj_n}\epsilon_{i_1i_2\ldots i_n}\epsilon_{j_1j_2\ldots j_n}.
\eqn
Presently we will show that
\beqn
\epsilon':=g\otimes g\otimes\ldots \otimes g\epsilon=\det(g)\epsilon,\nonumber
\eqn
for all matrices $g$. In components we have
\beqn
\epsilon'_{i_1i_2\ldots i_n}=\sum_{j_1,j_2,\ldots,j_n} g_{i_1j_1}g_{i_2j_2}\ldots g_{i_nj_n}\epsilon_{j_1j_2\ldots j_n},\nonumber
\eqn
and it is clear that $\epsilon'_{i_1i_2\ldots i_n}$ is completely anti-symmetric under interchange of indices and hence must be proportional to $\epsilon_{i_1i_2\ldots i_n}$. Finally we use (\ref{epsdet}) to conclude that
\beqn
\epsilon'=\det(g)\epsilon.\nonumber
\eqn
\begin{thm}\label{invariantcond}
Consider a function $f: V^{\otimes n}\times V^{\otimes m}\rightarrow\mathbb{C}$ which satisfies the conditions:
\begin{enumerate}
\item {For fixed $\chi\in V^{\otimes n}$ we have $f\in\mathbb{C}[V^{\otimes m}]_d$. That is $f(\chi,c\psi)=c^df(\chi,\psi)$.}
\item {For fixed $\psi\in V^{\otimes m}$ we have $f\in\mathbb{C}[V^{\otimes n}]_k$. That is $f(c\chi,\psi)=c^kf(\chi,\psi)$.}
\item {$f(\chi,\otimes^mg\psi)=f(\otimes^ng^t\chi,\psi)$.}
\end{enumerate}
The function $f_{\epsilon}:V^{\otimes m}\rightarrow\mathbb{C}$ given by
\beqn
f_{\epsilon}(\psi):=f(\epsilon,\psi)\nonumber
\eqn
then satisfies $f_{\epsilon}(\otimes^mg\psi)=\det(g)^kf_{\epsilon}(\psi)$.
\end{thm}
\begin{proof}
We have
\beqn
f_{\epsilon}(\otimes^mg\psi)=f(\epsilon,\otimes^mg\psi)=f(\otimes^ng^t\epsilon,\psi)=f(\det(g)\epsilon,\psi)&=\det(g)^kf(\epsilon,\psi)\\
&=\det(g)^kf_{\epsilon}(\psi).\nonumber
\eqn
\end{proof}\noindent
This theorem gives us some idea of how to explicitly construct invariants for the general linear group. The rest of this chapter will be devoted to the illustration of several examples.
\subsection{Invariants of $GL(n)$ on $V^{\otimes m}$}
For this case the number of invariants of $GL(n)$ on $V^{\otimes m}$ is given by the multiplicity of $\{k^n\}$ in $(\times^m\{1\})^{\otimes\{d\}}$ with $nk=md$. Here we will consider $m=2$ and the cases $n=2,3,4.$ 
\subsubsection{The case of $GL(2)$}
In the case that $n=m=2$, the possible degrees of the invariants are 
\beqn
d=1,2,3,4,\ldots \nonumber
\eqn
and using \textbf{Schur} we find 
\beqn
(\{1\}\times \{1\})^{\otimes\{1\}}&\ni \{1^2\},\nonumber\\
(\{1\}\times \{1\})^{\otimes\{2\}}&\ni 2\{2^2\},\\
(\{1\}\times \{1\})^{\otimes\{3\}}&\ni 2\{3^2\},\\
(\{1\}\times \{1\})^{\otimes\{4\}}&\ni 3\{4^2\},\\
(\{1\}\times \{1\})^{\otimes\{5\}}&\ni 3\{5^2\},\\
(\{1\}\times \{1\})^{\otimes\{6\}}&\ni 4\{6^2\}.
\eqn
At each degree the correct number of invariants can be built from 
\beqn\label{eq:inv1}
f_1(\psi)&:=\sum_{i_1,i_2} \psi_{i_1i_2}\epsilon_{i_1i_2}=(\psi_{12}-\psi_{21}),\\
f_2(\psi)&:=\sum_{i_1,i_2,j_1,j_2} \psi_{i_1i_2}\psi_{j_1j_2}\epsilon_{i_1j_1}\epsilon_{i_2j_2}=(\psi_{11}\psi_{22}-\psi_{11}\psi_{22}),
\eqn
and are non-zero, algebraically independent, and by inspection satisfy (\ref{invariantcond}).
\subsubsection{The case of $GL(3)$}
In the case that $n=3,m=2$, the possible degrees of invariants are 
\beqn
d=3,6,9,12,\ldots \nonumber
\eqn
Computing plethysms in \textbf{Schur} gives
\beqn
(\{1\}\times \{1\})^{\otimes\{3\}}&\ni 2\{2^3\},\nonumber\\
(\{1\}\times \{1\})^{\otimes\{6\}}&\ni 3\{4^3\},\\
(\{1\}\times \{1\})^{\otimes\{9\}}&\ni 4\{6^3\}.
\eqn
At each degree the correct number of invariants can be built from the two $d=3$ invariants:
\beqn\label{eq:inv2}
f_1(\psi):&=\sum_{i_1,i_2,j_1,j_2,k_1,k_2}\psi_{i_1i_2}\psi_{j_1j_2}\psi_{k_1k_2}\epsilon_{i_1j_1j_2}\epsilon_{i_2j_2k_2}\\
&=-\psi_{13} \psi_{22} \psi_{31}+\psi_{12} \psi_{23} \psi_{31}+\psi_{13} \psi_{21} \psi_{32}\\&\hspace{20mm}-\psi_{11} \psi_{23} \psi_{32}-\psi_{12} \psi_{21} \psi_{33}+\psi_{11} \psi_{22} \psi_{33},\\
f_2(\psi):&=\sum_{i_1,i_2,j_1,j_2,k_1,k_2}\psi_{i_1i_2}\psi_{j_1j_2}\psi_{k_1k_2}\epsilon_{i_1i_2j_1}\epsilon_{j_2k_1k_2}\\
&=\psi_{1 3}^2 \psi_{2 2} - \psi_{1 2} \psi_{1 3} \psi_{2 3} - 
    \psi_{1 3} \psi_{2 1} \psi_{2 3} + \psi_{1 1} \psi_{2 3}^2\\& - 
    2 \psi_{1 3} \psi_{2 2} \psi_{3 1} + 3 \psi_{1 2} \psi_{2 3} \psi_{3 1} - 
    \psi_{2 1} \psi_{2 3} \psi_{3 1} + \psi_{2 2} \psi_{3 1}^2\\& - 
    \psi_{1 2} \psi_{1 3} \psi_{3 2} + 3 \psi_{1 3} \psi_{2 1} \psi_{3 2} - 
    2 \psi_{1 1} \psi_{2 3} \psi_{3 2} - \psi_{1 2} \psi_{3 1} \psi_{3 2}\\& - 
    \psi_{2 1} \psi_{3 1} \psi_{3 2} + \psi_{1 1} \psi_{3 2}^2 + \psi_{1 2}^2 \psi_{3 3}\\& - 
    2 \psi_{1 2} \psi_{2 1} \psi_{3 3} + \psi_{2 1}^2 \psi_{3 3},
\eqn
which are non-zero, linearly independent and satisfy (\ref{invariantcond}).
\subsubsection{The case of $GL(4)$}
In the case that $n=4,m=2$, the possible degrees of the invariants are
\beqn
d=2,4,6,\ldots \nonumber
\eqn 
and \textbf{Schur} gives
\beqn
(\{1\}\times \{1\})^{\otimes\{2\}}&\ni \{1^4\},\nonumber\\
(\{1\}\times \{1\})^{\otimes\{4\}}&\ni 3\{2^4\},\\
(\{1\}\times \{1\})^{\otimes\{6\}}&\ni 3\{3^4\},\\
(\{1\}\times \{1\})^{\otimes\{8\}}&\ni 6\{4^4\}.
\eqn
The correct number of invariants can be constructed from three invariants of degree $d=2,4,4$ respectively:
\beqn\label{sumconvention}
f_1(\psi):&=\sum_{i_1,i_2,j_1,j_2}\psi_{i_1i_2}\psi_{j_1j_2}\epsilon_{i_1i_2j_1j_2},\\
f_2(\psi):&=\sum_{i_1,i_2,j_1,j_2,k_1,k_2,l_1,l_2}\psi_{i_1i_2}\psi_{j_1j_2}\psi_{k_1k_2}\psi_{l_1l_2}\epsilon_{i_1j_1k_1l_1}\epsilon_{i_2j_2k_2l_2},\\
f_3(\psi):&=\sum_{i_1,i_2,j_1,j_2,k_1,k_2,l_1,l_2}\psi_{i_1i_2}\psi_{j_1j_2}\psi_{k_1k_2}\psi_{l_1l_2}\epsilon_{i_1i_2j_1k_1}\epsilon_{j_2k_2l_1l_2},\eqn
which by explicit expansion (either by hand or using a computer algebra package)  are non-zero, algebraically independent and satisfy (\ref{invariantcond}).\footnote{As the number of indices in these expressions is becoming prohibitively large, we will adopt a convention from now until the end of the thesis that, unless otherwise indicated, any indices that appear after a summation sign are to be summed over appropriate bounds.}
\subsection{Invariants of $\times^mGL(n)$ on $V^{\otimes m}$}
We consider the existence of invariants $q:V^{\otimes m}\rightarrow \mathbb{C}$ which take the form
\beqn
q(gx)=\det(g)^kq(x),\nonumber
\eqn
for all $g=g_{1}\otimes g_{2}\otimes\ldots \otimes g_{m}$ with $g_{a}\in GL(n)$ for $1\leq a\leq m$. We mimic the construction of the previous section and give sufficient conditions for the existence of such functions.
\begin{thm}\label{invariantcond2}
Consider a function $q: (\times^m V^{\otimes n})\times V^{\otimes m}\rightarrow\mathbb{C}$ which satisfies the conditions:
\begin{enumerate}
\item {For fixed $\psi\in V^{\otimes m}$ we have \beqn q(\chi_1,\ldots ,c\chi_a,\ldots ,\chi_m;\psi)=c^kq(\chi_1,\ldots ,\chi_a,\ldots ,\chi_m;\psi),\nonumber\eqn for each $1\leq a\leq m$.}
\item {For fixed $\chi_a\in V^{\otimes n}$, $1\leq a\leq m$, we have \beqn q(\chi_1,\ldots ,\chi_m;c\psi)=c^dq(\chi_1,\ldots ,\chi_m;\psi).\nonumber\eqn}
\item {For all $g=g_1\otimes g_2\otimes\ldots \otimes g_m$ we have\beqn q(\chi_1,\ldots ,\chi_m;g\psi)=q(\otimes^ng_1^t\chi_1,\ldots ,\otimes^ng_m^t\chi_m;\psi).\nonumber\eqn}
\end{enumerate}
The function $q_{\epsilon}:V^{\otimes m}\rightarrow\mathbb{C}$ given by
\beqn
q_{\epsilon}(\psi):=q(\epsilon,\epsilon,\ldots ,\epsilon;\psi),\nonumber
\eqn
satisfies $q_{\epsilon}(g\psi)=\det(g)^kf_{\epsilon}(\psi)$ for all $g=g_1\otimes g_2\otimes\ldots \otimes g_m$.
\end{thm}
\begin{proof}
We have
\beqn
q_{\epsilon}(g\psi)&=q(\epsilon,\ldots ,\epsilon;g\psi)\nonumber\\
&=q(\otimes^ng_1^t\epsilon,\ldots ,\otimes^ng_m^t\epsilon;\psi)\\
&=q(\det(g_1)\epsilon,\ldots ,\det(g_m)\epsilon;\psi)\\
&=\det(g_1)^k\det(g_2)^k\ldots \det(g_m)^kq(\epsilon,\ldots ,\epsilon;\psi)\\
&=\det(g)^kq_{\epsilon}(\psi).
\eqn
\end{proof}\noindent
With these sufficient conditions in mind we will use the Schur functions to ascertain existence of these invariants and give examples of their exact form. 
\subsubsection{The case of $k=1$}
From (\ref{thm:inner}) the existence of such invariants requires that for the $m$-fold inner product we have:
\beqn
\{1^n\}\ast \{1^n\}\ast \ldots \ast \{1^n\}\ni \{n\}\nonumber.
\eqn
Now for even $m$ we have
\beqn
\{1^n\}\ast \{1^n\}\ast \ldots \ast \{1^n\}=\{n\}\nonumber
\eqn
and for odd $m$
\beqn
\{1^n\}\ast \{1^n\}\ast \ldots \ast \{1^n\}=\{1^n\}\nonumber.
\eqn
So that there exists a single invariant for each even $m$ and no invariants for odd $m$.\\
For $m=2$ and $n=2,3,4$ these invariants are
\beqn\label{concurrence}
{\det}_2(\psi)&=\sum\psi_{i_1i_2}\psi_{j_1j_2}\epsilon_{i_1j_1}\epsilon_{i_2j_2},\\
{\det}_3(\psi)&=\sum\psi_{i_1i_2}\psi_{j_1j_2}\psi_{k_1k_2}\epsilon_{i_1j_1k_1}\epsilon_{i_2j_2k_2},\\
{\det}_4(\psi)&=\sum\psi_{i_1i_2}\psi_{j_1j_2}\psi_{k_1k_2}\psi_{l_1l_2}\epsilon_{i_1j_1k_1l_1}\epsilon_{i_2j_2k_2l_2},
\eqn
which can be seen to satisfy $(\ref{invariantcond2})$ and can be generalized in the obvious manner for any $n$. (These polynomials should be distinguished from the determinant of a matrix; although their functional form is identical to that of the determinant, they arise as invariant functions on the linear space $V\otimes V$.) For $m=4$ and $n=2,3,4$ we can define:
\beqn\label{eq:quangles}
Q_2(\psi)&=\sum\psi_{i_1i_2i_3i_4}\psi_{j_1j_2j_3j_4}\epsilon_{i_1j_1}\epsilon_{i_2j_2}\epsilon_{i_3j_3}\epsilon_{i_4,j_4},\\
Q_3(\psi)&=\sum\psi_{i_1i_2i_3i_4}\psi_{j_1j_2j_3j_4}\psi_{k_1k_2k_3k_4}\epsilon_{i_1j_1k_1}\epsilon_{i_2j_2k_2}\epsilon_{i_3j_3k_3}\epsilon_{i_4j_4k_4},\\
Q_4(\psi)&=\sum\psi_{i_1i_2i_3i_4}\psi_{j_1j_2j_3j_4}\psi_{k_1k_2k_3k_4}\psi_{l_1l_2l_3l_4}\epsilon_{i_1j_1k_1l_1}\epsilon_{i_2j_2k_2l_2}\epsilon_{i_3j_3k_3l_3}\epsilon_{i_4j_4k_4l_4},
\eqn
which can also be seen to satisfy $(\ref{invariantcond2})$ and can be generalized in the obvious way for arbitrary $n$. We refer to these invariants as \textit{quangles}.
\subsubsection{The case of $\times^mGL(2)$ and $k=2$}
For $m=2,3,4$ \textbf{Schur} shows that
\beqn
\{2^2\}&\ast\{2^2\}\ni\{4\},\nonumber\\
\{2^2\}&\ast\{2^2\}\ast\{2^2\}\ni\{4\},\nonumber\\
\{2^2\}&\ast\{2^2\}\ast\{2^2\}\ast\{2^2\}\ni 3\{4\}.
\eqn
At $m=2$ the required invariant is the pointwise product of ${\det}_2$ with itself. Whereas at $m=3$ we have the \textit{tangle}\footnote{The tangle is known and used in physics to analyse multiparticle entanglement in quantum mechanics. This will be reviewed in Chapter \ref{chap3}} 
\beqn\label{tangle2}
\mathcal{T}_2(\psi)=\sum \psi_{i_1i_2i_3}\psi_{j_1j_2j_3}\psi_{k_1k_2k_3}\psi_{l_1l_2l_3}\epsilon_{i_1j_1}\epsilon_{i_2j_2}\epsilon_{k_1l_1}\epsilon_{k_2l_2}\epsilon_{i_3l_3}\epsilon_{j_3k_3}.
\eqn
At $m=4$, the pointwise product of $Q_2$ with itself forms a $k=2$ invariant and we have the additional invariants: 
\beqn
\mathcal{I}_1:=\sum\psi_{i_1i_2i_3i_4}\psi_{j_1j_2j_3j_4}\psi_{k_1k_2k_3k_4}\psi_{l_1l_2l_3l_4}\epsilon_{i_1j_1}\epsilon_{i_2j_2}\epsilon_{k_1l_1}\epsilon_{k_2l_2}\epsilon_{i_3k_3}\epsilon_{i_4k_4}\epsilon_{j_3l_3}\epsilon_{j_4l_4},\nonumber\\
\mathcal{I}_2:=\sum\psi_{i_1i_2i_3i_4}\psi_{j_1j_2j_3j_4}\psi_{k_1k_2k_3k_4}\psi_{l_1l_2l_3l_4}\epsilon_{i_1j_1}\epsilon_{i_2l_2}\epsilon_{i_3l_3}\epsilon_{i_4k_4}\epsilon_{j_2k_2}\epsilon_{j_3k_3}\epsilon_{j_4l_4}\epsilon_{k_1l_1},\nonumber\\
\eqn
which satisfy (\ref{invariantcond2}) and can be shown to be non-zero and algebraically independent.
\subsubsection{The case of $GL(3)^{\times m}$ and $k=2$}
For $m=2,3,4$ \textbf{Schur} shows that
\beqn
\{2^3\}\ast\{2^3\}&\ni\{6\},\nonumber\\
\{2^3\}\ast\{2^3\}\ast\{2^3\}&\ni \{6\},\nonumber\\
\{2^3\}\ast\{2^3\}\ast\{2^3\}\ast\{2^3\}&\ni 4\{6\}.
\eqn
At $m=2$, the pointwise product of ${\det}_3$ with itself forms a $k=2$ invariant and at $m=3$ the tangle can be generalized to the $n=3$ case:
\beqn\label{tangle3}
\mathcal{T}_3(\psi)=\sum\psi_{i_1i_2i_3}\psi_{j_1j_2j_3}\psi_{k_1k_2k_3}\psi_{l_1l_2l_3}\psi_{m_1m_2m_3}\psi_{n_1n_2n_3}\\\cdot\epsilon_{i_1j_1k_1}\epsilon_{j_2k_2l_2}\epsilon_{k_3l_3m_3}\epsilon_{l_1m_1n_1}\epsilon_{m_2n_2i_2}\epsilon_{n_3i_3j_3},
\eqn
which by explicit expansion can be shown to be non-zero.\\
The invariants at $m=4$ remain uninvestigated.
\subsubsection{The case of $\times^mGL(4)$ and $k=2$}
For $m=2,3,4$ \textbf{Schur} shows that
\beqn
\{2^4\}\ast\{2^4\}&\ni\{8\},\nonumber\\
\{2^4\}\ast\{2^4\}\ast\{2^4\}&\ni \{8\},\nonumber\\
\{2^4\}\ast\{2^4\}\ast\{2^4\}\ast\{2^4\}&\ni 7\{8\}.
\eqn
At $m=2$ the pointwise product of ${\det}_4$ with itself is a $k=2$ invariant and at $m=3$ the tangle can again be generalized:
\beqn\label{tangle4}
\mathcal{T}_4=\sum\psi_{i_1j_1k_1}\psi_{i_2j_2k_2}\psi_{i_3j_3k_3}\psi_{i_4j_4k_4}\psi_{i_5j_5k_5}\psi_{i_6j_6k_6}\psi_{i_7j_7k_7}\psi_{i_8j_8k_8}\\
\cdot\epsilon_{i_1i_2i_3i_4}\epsilon_{i_5i_6i_7i_8}\epsilon_{j_1j_5j_4j_8}\epsilon_{j_2j_6j_3j_7}\epsilon_{k_1k_5k_2k_6}\epsilon_{k_3k_7k_4k_8},
\eqn
and shown to be non-zero by explicit expansion.\\
The invariants at $m=4$ remain uninvestigated.
\section{Closing remarks}
In this chapter we have given a review of the use of the character theory to build the irreducible representations of the general linear group. We have demonstrated the concrete connection between the one-dimensional representations and the classical invariants, and have presented theorems that allow us to count these invariants at given degree $d$ and weight $k$.

\chapter{Entanglement and phylogenetics}\label{chap3}
Stochastic methods that model character distributions in aligned sequences are part of the standard armoury of phylogenetic analysis \cite{felsenstein1981,felsenstein2004,nei2000,rodriguez1990,steel1998}. The evolutionary relationships are usually represented as a bifurcating tree directed in time. It is remarkable that there is a strong conceptual and mathematical analogy between the construction of phylogenetic trees using stochastic methods, and the process of scattering in particle physics \cite{jarvis2001}. It is the purpose of the present chapter to show that there is much potential in taking an algebraic, group theoretical approach to the problem where the inherent symmetries of the system can be fully appreciated and utilized.\\
Entanglement is of considerable interest in physics and there has been much effort to elucidate the nature of this curious physical phenomenon \cite{bernevig2003,dur2000,guhne2003,linden1998,werner2001}. Entanglement has its origin in the manner in which the state probabilities of a quantum mechanical system must be constructed from the individual state probabilities of its various subsystems. Whenever there are global conserved quantities, such as spin, there exist entangled states where the choice of measurement of one subsystem can affect the measurement outcome of another subsystem no matter how spatially separated the two subsystems are. This curious physical property is represented mathematically by nonseparable tensor states. Remarkably, if the pattern frequencies of phylogenetic analysis are interpreted in a tensor framework it is possible to show that the branching process itself introduces entanglement into the state. In the context of phylogenetics this element of entanglement corresponds to nothing other than that of phylogenetic relation. This is a mathematical curiosity that can be studied using methods from quantum physics. This is a novel way of approaching phylogenetic analysis which has not been explored before.\\
This chapter will begin by establishing the formalism of quantum mechanics and introducing the concept of entanglement through an elementary example. A short review of the use of group invariant functions to analyse  entanglement will be presented. The stochastic model of a phylogenetic tree will then be developed in its standard form, followed by a discussion which establishes a presentation of this model in the form of a group action on a tensor product space as used in quantum mechanics. The invariant functions used to study entanglement will then be examined in the context of phylogenetic trees. 

\noindent
\textbf{Note:} Elements of this chapter are extracted from \cite{sumner2005}.

\section{Quantum mechanics}
The formalism of the quantum mechanical description of physical systems amounts to four fundamental postulates.
\begin{postulate}
The mathematical description of any physical system occurs as a \textit{state vector} $\psi$ in a complex vector space, $V$, together with an inner product known as a \textit{Hilbert space} $\mathcal{H}=(V,.)$. 
\end{postulate}\noindent
For a given physical system it is not \textit{a priori} apparent exactly how the Hilbert space should be chosen. As will be elaborated later, a basic property of quantum mechanics is that it is not possible to determine (in practice or in principle) the exact and complete configuration of a physical system. Thus, the Hilbert space is chosen not to represent all possible configurations of the system, but rather to represent whichever part is observable and under consideration in a given experimental setup. For example the full description of an electron is given by the tensor product of the representation space of the spin, $\mathbb{C}^2$, with that of the representation space of spatial position, square integrable functions $\{f:\mathbb{R}^3\rightarrow \mathbb{C}\}$. However, one is often only interested in the spin degrees of freedom of the system and simply ignores the position component of the state vector.\\
For our purposes it will be enough to consider only the case where $\mathcal{H}$ is the finite dimensional vector space with inner product given in terms of notation from Chapter \ref{chap2} as
\beqn
(\psi,\varphi)=\overline{\psi}(\varphi).\nonumber
\eqn
\begin{postulate}
The dynamical evolution of any physical system is governed by the linear equation
\beqn\label{schrodinger}
i\hbar\frac{\partial\psi(t)}{\partial t}=H(t)\psi(t),
\eqn
where $\hbar$ is Planck's constant and $H(t)$ is a \textit{Hermitian} operator:
\beqn
(\psi,H\varphi)=(H\psi,\varphi),\nonumber
\eqn
known as the \textit{Hamiltonian}. Completely equivalently, the dynamical evolution is described by solutions of (\ref{schrodinger}):
\beqn
\psi(t_2)=U(t_2,t_1)\psi(t_1),\nonumber
\eqn
where $U(t_2,t_1)$ is a unitary operator
\beqn
(U\psi,U\varphi)=(\psi,\varphi).\nonumber
\eqn
\end{postulate}\noindent
From this postulate it is not apparent how the Hamiltonian should be chosen in any particular case. Historically, Dirac formalized the idea of \textit{classical analogy} where the Hamiltonian is interpreted as the total energy of the system \cite{dirac1958}. However, this procedure is limited to systems which have a classical counterpart and the general case is left to the modern quantum physicist.
\begin{postulate}
An \textit{observable} of a physical system is described by an Hermitian operator $A$ with associated eigenvalues $\{\alpha_1,\alpha_2,\ldots\}$ and eigenspaces defined by the projection operators $\{P_1,P_2,\ldots\}$. If the state vector before measurement is $\psi$, then the probability of the result $\alpha_i$ is given by
\beqn
\frac{(\psi,P_i\psi)}{(\psi,\psi)},\nonumber
\eqn
and the state after measurement is $\psi'=P_i\psi$.
\end{postulate}\noindent
From this definition it is apparent that $U$ must be unitary to preserve total probability. We will follow the standard procedure of \textit{normalizing} the state vector:
\beqn
(\psi,\psi)=1.\nonumber
\eqn
\begin{postulate}
The state space of a composite of $m$ quantum systems with individual state spaces $\mathcal{H}_1,\mathcal{H}_2,\ldots,\mathcal{H}_m$ is given by the tensor product:
\beqn
\mathcal{H}=\mathcal{H}_1\otimes \mathcal{H}_2 \otimes\ldots\otimes \mathcal{H}_m.\nonumber
\eqn
\end{postulate}\noindent
From this definition it may seem that the state vector of a composite system should be expressed as the \textit{product state}
\beqn\label{productstate}
\psi=\varphi^{(1)}\otimes \varphi^{(2)} \otimes\ldots\otimes \varphi^{(m)},
\eqn
where $\varphi^{(a)}\in\mathcal{H}_a$ is the state vector of each individual system. However, for the general case, there are physical reasons why there must exist states which cannot be written in the form (\ref{productstate}). We will explore these states and their curious properties in the next section.
\subsection{Spin $\fra{1}{2}$ and entanglement}
One way to proceed in the search for the appropriate state space $\mathcal{H}$ is to study the representation spaces of the irreducible representations of a symmetry group of a physical system. For the case of three dimensional Euclidean space, consider the symmetry group of proper rotations; the special orthogonal group $SO(3)$. The irreducible representations of $SO(3)$ are labelled by the \textit{spin} quantum numbers $s=\{0,\fra{1}{2},1,\fra{3}{2},2,\fra{5}{2},\ldots\}$ (see \cite{miller1972}). Here we will study the case $s=\fra{1}{2}$ where the representation is two-dimensional: $\mathcal{H}=\mathbb{C}^2$, and a state vector is referred to as a \textit{qubit}. The physics of the spin of a qubit is captured by considering an orthonormal basis for $\mathbb{C}^2$ as $\{z_+,z_-\}$ and introducing the observable $S_z$ satisfying 
\beqn
S_zz_+=\fra{\hbar}{2}z_+,\qquad S_zz_-=-\fra{\hbar}{2}z_-,\nonumber
\eqn
so that the states $\psi^+:=z_+$ and $\psi^-:=z_-$ are eigenvectors of the spin operator. Analogously, we can define the $x$ basis $\{x_+,x_-\}$ (or any other orthonormal basis) by rotating the $z$ basis using the group element of the two-dimensional representation of $SO(3)$ which corresponds to the appropriate physical rotation. In particular, we have
\beqn
x_+=\fra{1}{\sqrt{2}}(z_++z_-),\qquad x_-=\fra{1}{\sqrt{2}}(x_+-x_-).\nonumber
\eqn
The measurement operators are then defined as being the projection operators onto the appropriate basis vectors. For instance the projection operators for spin in the $z$ direction satisfy
\beqn
P_z^{+}z_+=z_+,\qquad P_z^{+}z_-=0,\qquad P_z^{-}z_+=0,\qquad P_z^{-}z_-=z_-.\nonumber
\eqn 
A generic qubit can be written as
\beqn
\psi=c_1z_++c_2z_-.\nonumber
\eqn
Introduce the random variables $A_z\in \{+1,-1\}$ to correspond to the value of the spin along the $z$ axis, and we have
\beqn
\mathbb{P}(A_z=1)=(\psi,P_z^+\psi)=|c_1|^2,\nonumber
\eqn
and 
\beqn
\mathbb{P}(A_z=-1)=(\psi,P_z^-\psi)=|c_2|^2.\nonumber
\eqn
Now we turn our attention to composite states of $m$ qubits where the state space becomes
\beqn
\mathcal{H}={(\mathbb{C}^2)}^{\otimes m}.\nonumber
\eqn
The most general state can be expressed as
\beqn
\psi=\sum \psi_{i_1i_2\ldots i_m}e_{i_1}\otimes e_{i_2} \otimes\ldots\otimes e_{i_m},\nonumber
\eqn
so that the state is specified by $2^m$ complex numbers $\psi_{i_1i_2\ldots i_m}$. In the case where $\psi$ can be expressed in the form of a product state, we have
\beqn
\psi_{i_1i_2\ldots i_m}=\varphi^{(1)}_{i_1}\varphi^{(2)}_{i_2} \ldots\varphi^{(m)}_{i_m},\nonumber
\eqn
and we see that the state is specified by $2m$ complex numbers. The difference in these parameter counts between the general state and the product state is the origin of \textit{entanglement}.\\
To illustrate the simplest example of entanglement consider the case of a spin zero particle splitting into two spin $\fra{1}{2}$ qubits labelled as $A$ and $B$. To ensure that the total spin is zero, it must be the case that the total state is
\beqn
\psi=\fra{1}{\sqrt{2}}(z_+\otimes z_--z_-\otimes z_+),\nonumber
\eqn 
which ensures that $(S_z\otimes 1+1\otimes S_z)\psi=0$. We introduce the random variables $A_z$ for particle $A$ and $B_z$ for particle $B$. For the state, $\psi$, the measurement of spins of $A$ and $B$ along the $z$ axis is associated with the probabilities
\beqn
\mathbb{P}(A_z=1,B_z=1)&=\mathbb{P}(A_z=-1,B_z=-1)=0,\nonumber\\
\mathbb{P}(A_z=1,B_z=-1)&=\mathbb{P}(A_z=-1,B_z=1)=1/2,
\eqn
\beqn
\mathbb{P}(A_z=1)=\mathbb{P}(A_z=-1)=1/2,\nonumber
\eqn
and
\beqn
\mathbb{P}(B_z=1)=\mathbb{P}(B_z=-1)=1/2.\nonumber
\eqn
Now if we consider the same state but with spin measurements taken along the $x$ axis, it is a simple exercise to show that
\beqn
\psi=\fra{1}{\sqrt{2}}(x_+\otimes x_--x_-\otimes x_+).\nonumber
\eqn
Now if we were to go ahead and compute the various probabilities associated with the observable $S_x$ we would come to the same probabilities as above. That is, the spins of $A$ and $B$ are always opposite to give $A_x=1,B_x=-1$ with probability $\fra{1}{2}$ and $A_x=-1,B=+1$ with probability $\fra{1}{2}$. One can go further and show that this is true for $any$ orthonormal basis of $\mathbb{C}^2$. This implies that no matter which axis the spins are measured along, the outcome at $A$ is always the negative of the outcome at $B$. These probabilities have been amply confirmed by experiment.\\
A problem arises if one wishes to interpret the probabilities of the formalism of quantum mechanics as representing our ignorance of the full state of the physical system. Such a description of these events would require that at the moment of splitting, each particle actually carries the requisite information as how to respond to a spin measurement on an arbitrary axis, and somehow this information is unobservable or hidden from us. This additional information over and above the state vector was historically coined the \textit{hidden variables}. However, Bell showed that it is actually impossible to specify the required hidden variables \cite{bell1964} and thus it is not possible to interpret the probabilities as simply representing our ignorance of the system. This implies that quantum mechanics requires that the physical world is probabilistic in an intrinsic way. An alternative way out of this predicament is to assume that there is a non-local communication between particles $A$ and $B$, which ensures that spins are opposite along any axis. However, at the moment of measurement, $A$ and $B$ could be separated by a very large distance! Thus the entanglement leads us to the dilemma of having to accept one of the following:
\begin{itemize}
\item Quantum systems have an essentially non-local property.
\item The probabilities in quantum mechanics do not just indicate our ignorance of the configuration of a physical system, but are an essential part of physical reality.
\end{itemize}\noindent
Einstein was unhappy with both options, and never made his peace with the quantum theory that he was so instrumental in constructing. This is because the first violates the spirit, if not the detail of special relativity grossly, and the second implies that Einstein's contention that ``God does not play dice" cannot be true.
\\
Recall that the conditional probability that the random variable $A=x$ given that $B=y$ is defined to be
\beqn
\mathbb{P}(A=x|B=y):=\frac{\mathbb{P}(A=x,B=y)}{\mathbb{P}(B=y)}.\nonumber
\eqn
The random variables $A$ and $B$ are said to be \textit{stochastically independent} \cite{feller1968} if and only if 
\beqn
\mathbb{P}(A=x,B=y)=\mathbb{P}(A=x)\mathbb{P}(B=y),\nonumber
\eqn
from which it would follow that
\beqn
\mathbb{P}(A=x|B=y)=\mathbb{P}(A=x),\nonumber
\eqn
which motivates the definition. (This notion of stochastic independence can be extended to multiple random variables. For details see Feller \cite{feller1968}.)\\
In quantum mechanics, stochastic independence is implied if the state is a product $\psi=\varphi^{(1)}\otimes\varphi^{(2)}$. For if the state is a product state, we have
\beqn
\mathbb{P}(A=i,B=j)&=\overline{\varphi^{(1)}}(P_i\varphi^{(1)})\overline{\varphi^{(2)}}(P_j\varphi^{(2)})\nonumber\\
&:=\mathbb{P}(A=i)\mathbb{P}(B=j).
\eqn
In what follows we will equate entanglement with this notion of stochastic dependence.
\subsection{Orbit classes and invariants}
We have seen that a quantum system exhibits entanglement if the state vector cannot be written as a product. Mathematically one would like to partition the set of entangled state vectors into equivalence classes which capture the essential property of entanglement. A systematic approach to the classification problem is to study the orbit classes of the tensor product space under a group action which is designed to preserve the essential non-local properties of entanglement. The orbit of an element $\psi\in\mathcal{H}$ under the group action $G$ is defined as the set of elements $\{\psi'=g\psi\text{ for some }g\in G\}$.\\
In quantum physics the appropriate group action is known to be the set of SLOCC operators, (Stochastic Local Operations with Classical Communication) \cite{dur2000,guhne2003,linden1998,miyake2003,nielsen2000}. Mathematically SLOCC operators correspond to the ability to transform the individual parts of the tensor product space $\mathcal{H}\cong \mathcal{H}_1\otimes \mathcal{H}_2\otimes \ldots\otimes \mathcal{H}_m$ with arbitrary invertible, linear operations. These operators are expressed by group elements of the form
\beqn
g=g_1\otimes g_2\otimes\ldots\otimes g_m,\nonumber
\eqn
where $m$ is the number of individual spaces making up the tensor product, and $g_i\in GL(\mathcal{H}_i)$.\\
The task is to identify the orbit classes of a given tensor product space under the general set of SLOCC operators. A powerful tool in this analysis is the construction of the invariant functions $\mathbb{C}[\mathcal{H}]^{G}$. By definition these invariants are relatively constant up to the determinant upon each orbit class of $\mathcal{H}$. It can be shown that there exists (under the action of the general linear group at least), a \textit{finite} set of elements which generate the full set of invariants on a given linear space. It can also be shown that the set of orbit classes of a given linear space can be completely classified given a full set of invariants on that space \cite{olver2003}.\\
In what follows we study the orbit class problem for the state space of two qubits and then that of three qubits.
\subsection{Two qubits and the concurrence}
Using the notation of Chapter $\ref{chap2}$, the concurrence is defined using (\ref{concurrence}): 
\beqn
\mathcal{C}=\det{}_2,\nonumber
\eqn
so that
\beqn
\mathcal{C}(\psi)=\sum \psi_{i_1i_2}\psi_{j_1j_2}\epsilon_{i_1j_1}\epsilon_{i_2j_2}.\nonumber
\eqn
We wish to construct the orbit classes of $\mathcal{H}=\mathbb{C}^2\otimes\mathbb{C}^2$ under the group action $GL(\mathbb{C}^2)\times GL(\mathbb{C}^2)$. Any state $\psi\in \mathcal{H}$ can be expressed using the four parameters $\psi_{i_1i_2}$ which in turn can be arranged as a matrix $M=[\psi_{i_1i_2}]$. Under the group transformation
\beqn
\psi\rightarrow g_1\otimes g_2\psi,\nonumber
\eqn
 the corresponding matrix transformation is 
\beqn
M\rightarrow g_1M g_2^{t}.\nonumber
\eqn
Hence we can answer the orbit class problem by taking a canonical $2\times2$ matrix $X$ and considering the set of matrices $\{M=AXB;A,B\in GL(\mathbb{C}^2)\}$.
\begin{thm}
The vector space $V\otimes V$ where $V\equiv\mathbb{C}^2$ has three orbits under the group action $GL(2)\times GL(2)$. Under the identification $M=[\psi_{i_1i_2}]$ for all $\psi$ the orbits are characterized by the following canonical forms:\\
\\
(i) Null-orbit $X=\left(\begin{array}{cc} 0 & 0 \\ 0 & 0\nonumber
\end{array}\right)
$,\\
(ii) Separable-orbit $Y=\left(\begin{array}{cc} 1 & 0 \\ 0 & 0\nonumber
\end{array}\right)
$,\\
(iii) Entangled-orbit $Z=\left(\begin{array}{cc} 1 & 0 \\ 0 & 1\nonumber
\end{array}\right)
$.\\
\\ The separable and entangled-orbits can be distinguished by the determinant function.
\end{thm}

\begin{proof}
(i) The null-orbit has only one member, the null vector; it is of course unchanged by the group action.\\
(ii) We are required to show that the set of $2\times 2$ matrices $\mathcal{M}=\{S:S=AYB;A,B\in GL(V)\}$ is all matrices such that $\det(S)=0$. We begin by taking a general member of $\mathcal{M}$, $S=\left(\begin{array}{cc} a & b \\ c & d\nonumber
\end{array}\right)$ with $ad-bc=0$. Clearly the matrices 
\beqn
S':&=\left(\begin{array}{cc} 0 & 1 \\ 1 & 0
\end{array}\right)S=\left(\begin{array}{cc} c & d \\ a & b
\end{array}\right),\nonumber
\quad
S'':&=S\left(\begin{array}{cc} 0 & 1 \\ 1 & 0
\end{array}\right)=\left(\begin{array}{cc} b & a \\ d & c
\end{array}\right)\nonumber,\quad\mbox{and}
\eqn
\beqn
S''':&=\left(\begin{array}{cc} 0 & 1 \\ 1 & 0
\end{array}\right)S\left(\begin{array}{cc} 0 & 1 \\ 1 & 0
\end{array}\right)=\left(\begin{array}{cc} d & c \\ b & a
\end{array}\right)\nonumber
\eqn
also belong to $\mathcal{M}$. So without loss of generality we can take $a\neq 0$ and it is an easy computation to show that
\beqn
S=\left(\begin{array}{cc} 1 & 0 \\ c/a & 1\nonumber
\end{array}\right)Y\left(\begin{array}{cc} a & b \\ 0 & 1\nonumber
\end{array}\right),
\eqn
so that $\mathcal{M}$ is the set of $2\times 2$ matrices with vanishing determinant.\\
(iii) Clearly any $2\times 2$ matrix $N$ with non-zero determinant can be written as $N=AZB$ where $A,B\in GL(\mathbb{C}^2)$.
\end{proof}
\begin{cor}
The orbits of $\mathcal{H}=\mathbb{C}^2\otimes \mathbb{C}^2$ under $SL(\mathbb{C}^2)\times SL(\mathbb{C}^2)$ are labelled by the determinant function $\det[\phi(h)].$
\end{cor}\noindent
For further discussion see \cite{bernevig2003,dur2000,linden1998}.
\subsection{Three qubits and the tangle}
It is known that there are six orbit classes of $\mathbb{C}^2\otimes \mathbb{C}^2\otimes \mathbb{C}^2$ under the action $GL(\mathbb{C}^2)\times GL(\mathbb{C}^2)\times GL(\mathbb{C}^2)$. These orbits classes can be distinguished by functions of the concurrence and another relative invariant known as the tangle \cite{dur2000,guhne2003}.\\
We begin by defining three partial concurrence operations as
\beqn\label{pconcurrence}
\mathcal{C}_1(\psi)=\sum \psi_{ijk}\psi_{lmn}\epsilon_{jm}\epsilon_{kn}e_i\otimes e_l,\\
\mathcal{C}_2(\psi)=\sum \psi_{ijk}\psi_{lmn}\epsilon_{il}\epsilon_{kn}e_j\otimes e_m,\\
\mathcal{C}_3(\psi)=\sum \psi_{ijk}\psi_{lmn}\epsilon_{il}\epsilon_{jm}e_k\otimes e_n.\\
\eqn
From these definitions it is easy to see that
\beqn
\mathcal{C}_1(\psi'):&=\mathcal{C}_1(g_1\otimes g_2 \otimes g_3 \psi)\nonumber\\
&=[\det(g_2)\det(g_3)]g_1\otimes g_1\mathcal{C}_1(\psi),
\eqn
with similar expressions for $\mathcal{C}_2$ and $\mathcal{C}_3$.\\
The tangle is an invariant satisfying
\beqn
\mathcal{T}(\psi)=[\det(g_1)\det(g_2)\det(g_3)]^2\mathcal{T}(\psi),\nonumber
\eqn
and from (\ref{tangle2})  can be written in the form
\beqn
\mathcal{T}(\psi)=\sum \psi_{a_1a_2a_3}\psi_{b_1b_2b_3}\psi_{c_1c_2c_3}\psi_{d_1d_2d_3}\epsilon_{a_1b_1}\epsilon_{a_2b_2}\epsilon_{c_1d_1}\epsilon_{c_2d_2}\epsilon_{b_3c_3}\epsilon_{a_3d_3}.\nonumber
\eqn
\\
The six orbit classes are described by the completely disentangled states 
\beqn
\psi=&\varphi^{(1)}\otimes \varphi^{(2)}\otimes\varphi^{(3)},\qquad \varphi^{(a)}\in\mathbb{C}^2;\nonumber
\eqn
the partially entangled states which form three orbit classes characterized by the separability of the canonical tensors
\beqn
\psi_p^{(1)}=\sum\varphi^{(1)}_i\varphi^{(23)}_{jk}e_{i}\otimes e_{j}\otimes e_k,\\\nonumber
\psi_p^{(2)}=\sum\varphi^{(13)}_{ik}\varphi^{(2)}_je_{i}\otimes e_{j}\otimes e_k,\\\nonumber
\psi_p^{(3)}=\sum\varphi^{(12)}_{ij}\varphi^{(3)}_ke_{i}\otimes e_{j}\otimes e_k;\nonumber
\eqn
the completely entangled states equivalent to the $GHZ$ state
\beqn
\psi_{ghz}=\fra{1}{\sqrt{2}}(e_0\otimes e_0\otimes e_0+e_1\otimes e_1\otimes e_1);\nonumber
\eqn
and the completely entangled states equivalent to the $W$ state
\beqn
\psi_w=\fra{1}{\sqrt{3}}(e_0\otimes e_0\otimes e_1+e_0\otimes e_1\otimes e_0+e_1\otimes e_0\otimes e_0).\nonumber
\eqn
The tangle and the concurrence and its partial counterparts can be used to fully distinguish these orbit classes.  For the completely disentangled tensors we have
\beqn
\mathcal{C}_a(\psi)=0,\quad \mathcal{T}(\psi)=0,\nonumber
\eqn 
for all $a=1,2,3$. Whereas for the first partially entangled state we have
\beqn
\mathcal{C}_1(\psi_p^{(1)})\neq 0,\quad\mathcal{C}_2(\psi_p^{(1)})=\mathcal{C}_3(\psi)=0,\quad \mathcal{T}(\psi_p^{(1)})=0,\nonumber
\eqn
and similar relations for the remaining two partially entangled states.\\
States in the $GHZ$ orbit satisfy
\beqn
\mathcal{C}_a(\psi_{ghz})\neq 0,\qquad\mathcal{T}(\psi_{ghz})\neq 0\nonumber
\eqn
for all $a=1,2,3$. Whereas states on the $W$ orbit satisfy
\beqn
\mathcal{C}_a(\psi_{w})\neq 0,\qquad\mathcal{T}(\psi_{w})= 0\nonumber
\eqn
for all $a=1,2,3$.\\
Notice that the $GHZ$ and $W$ orbits characterize different classes of three qubit entanglement. In the $GHZ$ orbit each qubit is entangled with the other two qubits \textit{and} the three qubits are entangled as a triplet. In the $W$ orbit the qubits are entangled as pairs but are \textit{not} entangled as a triplet.
\section{Stochastic evolution of biomolecular units}
It is standard to model sequence evolution as a stochastic process. A discrete set $\mathcal{K}$ is associated with biomolecular units which we refer to as  \textit{bases} and define $n:=|\mathcal{K}|$. For example, in the case of DNA sequences made up of the four nucleotides adenine, cytosine, guanine, thymine, we have $\mathcal{K}=\{A,G,C,T\}$ and $n=4$. The instance of a particular base in the sequence is equated with the time dependent random variable $X(t)\in\mathcal{K}$ and the stochastic time evolution is modelled as a continuous time Markov chain (CTMC) so that
\beqn\label{contmark}
\frac{d}{dt} \mathbb{P}(X(t)=i)=\sum_{j}q_{ij}(t)\mathbb{P}(X(t)=j),\qquad i,j\in\mathcal{K}.
\eqn
The $q_{ij}(t)$ are called \textit{rate parameters} and must satisfy the relations
\beqn\label{ratespars}
q_{ij}(t)\geq 0,\quad \forall i\neq j;\qquad q_{ii}(t)=-\sum_{j\neq i} q_{ji}(t).
\eqn
Define $Q(t)=\left[q_{ij}(t)\right]_{(i,j\in\mathcal{K})}$ as the \textit{rate matrix} associated with the Markov chain. The Markov chain is called \textit{homogeneous} if the rate matrix is time independent. The results presented in this thesis are equally valid for inhomogeneous models where the rate matrix is time dependent and so we allow for this generality throughout. It is also common to impose further symmetries upon the rate matrix such as the Jukes Cantor and  Kimura 3ST models \cite{nei2000}. However, the results presented here are again valid for any rate matrix satisfying (\ref{ratespars}), and hence no restriction upon the rate parameters is made. This model is referred to as the \textit{general Markov model} \cite{allman2003}.\\
For notational simplicity we will write $\pi_{i}(t):=\mathbb{P}(X(t)=i)$ and, given an initial distribution $\pi_i(0)$, write solutions of (\ref{contmark}) as
\beqn
\pi_i(t)=\sum_{j\in\mathcal{K}}m_{ij}(t,s)\pi_j(s),\qquad 0\leq s< t;\nonumber
\eqn
where $m_{ij}(t,s):=\mathbb{P}(X(t)=i|X(s)=j)$ are the \textit{transition probabilities} of the chain. We define the matrix $M(t,s)=\left[m_{ij}(t,s)\right]_{(i,j\in\mathcal{K})}$ such that in the homogeneous case the transition probabilities only depend on the difference $(t-s)$ and can be represented in terms of the rate matrix as
\beqn
M(t,s)=M(t-s,0)=e^{Q[(t-s)]}:=\sum_{n=0}^{\infty}\frac{Q^n[(t-s)]^n}{n!}.\nonumber
\eqn
In the inhomogeneous case there are several representations available for the matrix of transition probabilities (for details see \cite{isoifescu1980,rindos1995}). The representation that is of most use to us here is the time-ordered product:
\beqn\label{inhom}
M(t,s)=\mathbb{T}\exp{\int_s^tQ(u)du}
\eqn
(see for example \cite{itzykson1980} for the definition of the time-ordering operator $\mathbb{T}$.) For sufficiently small $\delta t$, we can write this in the approximate form
\beqn\nonumber
M(t,s)&\simeq M(t,t-\delta t)\ldots M(s+2\delta t,s+\delta t)M(s+\delta t,s)\\
&=e^{Q(t-\delta t)\delta t}e^{Q(t-2\delta t)\delta t}\ldots e^{Q(s+\delta t)\delta t}e^{Q(s)\delta t}.
\eqn
From these solutions it is clear that
\beqn\label{transdet}
\det[M(t,s)]=\exp{\int_s^ttr[Q(u)]du}.
\eqn 
A more fundamental way to define the transition matrices of a CTMC is to impose the \textit{backward} and \textit{forward} Kolmogorov equations \cite{isoifescu1980}: 
\beqn\label{kalamorov}
\frac{\partial M(t,s)}{\partial s}&=-M(t,s)Q(s),\\
\frac{\partial M(t,s)}{\partial t}&=Q(t)M(t,s).
\eqn
\section{Phylogenetic trees}
The remaining task is to model the case of phylogenetically related molecular sequences evolving under a stochastic process. Effectively the model consists of multiple copies of the random variable $X(t)$ taken as a generalization (via a tree structure) of a cartesian product and then modelled collectively as a CTMC. The reader is referred to \cite{semple2003} for a more extended discussion of the model. Here we keep the presentation to a minimum while allowing for the introduction of some essential notation and concepts.
\\ 
A tree, $T$, is a connected graph without cycles and consists of a set of vertices, $V$, and edges, $E$. Vertices of degree one are called \textit{leaves} and we partition the set of vertices as $V=L\cup N$ where $L$ is the set of leaves and $N$ is the set of internal vertices. We work with \textit{orientated} trees, which are defined by directing each edge of $T$ away from a distinguished vertex, $\pi$, known as the \textit{root} of the tree. Consequently, a given edge lying between vertices $u$ and $v$ is specified as an ordered pair $e=(u,v)$, where $u$ lies on the (unique) path between $v$ and $\pi$. The general Markov model of a phylogenetic tree is then made by assigning a set of random variables $\{X_s,\ s\in V\}$ to the vertices of the tree; these random variables are assumed to be conditionally independent and individually satisfy the properties of a CTMC. Taking a distribution at the root of the tree, $\{\mathbb{P}(X_\pi=i):=\pi_i, i\in\mathcal{K}\}$, completes the specification of the phylogenetic tree. The interpretation of a phylogenetic tree is that the probability distribution at each leaf is associated with the observed sequence of a single taxon and the joint probability distribution across a number of leaves is associated with the aligned sequences of the same number of molecular sequences.    
\\
For example in Figure \ref{pic:4leaf} we present the tree consisting of four leaves which has probability distribution
\beqn
p_{i_1i_2i_3i_4}=\sum_{j,k}m^{(1)}_{i_1j}m^{(2)}_{i_2j}m^{(3)}_{i_3k}m^{(4)}_{i_4k}m^{(5)}_{kj}\pi_j,\nonumber
\eqn
where 
\beqn
p_{i_1i_2i_3i_4}:=\mathbb{P}(X_1=i_1,X_2=i_2,X_3=i_3,X_4=i_4),\nonumber
\eqn
and we refer to these quantities as \textit{pattern probabilities}. 

\begin{figure}[t]
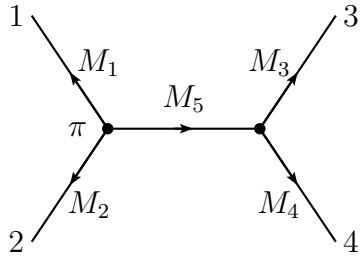

\centering
\pspicture[](0,0)(5,3)
\psset{linewidth=\pstlw,xunit=0.5,yunit=0.5,runit=0.5}
\psset{arrowsize=2pt 2,arrowinset=0.2}
\psline{-}(2,3)(0,0)
\rput(-0.4,0){2}
\psline{-}(2,3)(0,6)
\rput(-0.4,6){1}
\psline{-}(2,3)(6,3)
\psline{-}(6,3)(8,6)
\rput(8.4,6){3}
\psline{-}(6,3)(8,0)
\rput(8.4,0){4}
\psline{->}(2,3)(1,4.5)
\rput(1.75,4.75){$M_1$}
\psline{->}(2,3)(1,1.5)
\rput(1.5,1.0){$M_2$}
\psline{->}(6,3)(7,4.5)
\rput(6.25,4.75){$M_3$}
\psline{->}(6,3)(7,1.5)
\rput(6.5,1.0){$M_4$}
\psline{->}(3.75,3)(4.25,3)
\rput(4,3.8){$M_5$}
\pscircle[linewidth=0.8pt,fillstyle=solid,fillcolor=black](2,3){.15} 
\pscircle[linewidth=0.8pt,fillstyle=solid,fillcolor=black](6,3 ){.15} 
\rput(1.2,3){$\pi $}
\endpspicture 
\caption{Phylogenetic tree of four taxa}
\label{pic:4leaf}
\end{figure}


\section{Tensor presentation}
Setting $\mathbb{P}(X(t)=i)=p_i(t)$, we introduce the $n$-dimensional vector space $V$ with preferred basis \{$e_1,e_2,\ldots,e_n$\} and associate the probabilities uniquely with the vector
\beqn
p(t)=p_1(t)e_1+p_2(t)e_2+\ldots+p_n(t)e_n.\nonumber
\eqn  
The time evolution of this vector is then governed by equation (\ref{contmark}) written in operator form as
\beqn
\frac{d}{dt}p(t)=Q(t)p(t).\nonumber
\eqn
The solution of this equation is written as
\beqn
p(t)=M(t,s)p(s).\nonumber
\eqn
The probabilities can be recovered by taking the inner product
\beqn
p_i(t)=(e_i,p(t)),\nonumber
\eqn
and defining 
\beqn\label{def:theta}
\theta=\sum_{i=1}^n e_i,
\eqn
we have
\beqn
(\theta,p(t))=1,\quad\forall t.\nonumber
\eqn
In analogy we label the joint probabilities as 
\beqn
p_{i_1i_2\ldots i_m}(t):=\mathbb{P}(X_1=i_1,X_2=i_2,\ldots,X_m=i_m;t),\nonumber
\eqn
and by introducing the tensor product space $V^{\otimes m}$ we associate these probabilities with the unique tensor
\beqn
P(t):=\sum p_{i_1i_2\ldots i_m}(t)e_{i_1}\otimes e_{i_2}\otimes\ldots\otimes e_{i_m}.\nonumber
\eqn
Again the probabilities are recovered from the inner product:
\beqn
p_{i_1i_2\ldots i_m}(t)=(e_{i_1}\otimes e_{i_2}\otimes\ldots\otimes e_{i_m},P(t)),\nonumber
\eqn
and we define $\Omega=\sum e_{i_1}\otimes e_{i_2}\otimes\ldots\otimes e_{i_m}$ so that
\beqn
(\Omega,P(t))=1,\quad\forall t.\nonumber
\eqn
We now introduce the branching events into this formalism.\\
Consider a vertex on a phylogenetic tree where the stochastic evolution of a single random variable branches into that of two random variables. The corresponding mathematical operation is a mapping $V\rightarrow V\otimes V$. In order to formalize this we introduce the \textit{branching} operator $\delta:V\rightarrow V\otimes V$. The most general action of a (linear) operator $\delta$ upon the basis elements of $V$ can be expressed as
\beqn\label{splitting}
\delta e_i=\sum_{j,k}\Gamma _i^{jk}e_j\otimes e_k,
\eqn
where $\Gamma _i^{jk}$ are an arbitrary set of coefficients set by the assumption of conditional independence across branches of the tree.\\
To this end it is only necessary to consider initial probability distributions of the form 
\beqn 
\pi^{(\gamma)}&=\sum_i \delta_i^{\gamma}e_i,\nonumber\\\gamma &=1,2,\ldots,n.
\eqn 
Directly subsequent to the branching event the two leaf state is given by
\beqn
P^{(\gamma)}&=\delta \pi^{(\gamma)},\nonumber\\
&=\sum_{i,j,k}\delta^\gamma_i\Gamma_i^{jk}e_j\otimes e_k.
\eqn
We implement the conditional independence upon the branches by setting
\beqn\label{condindep}
\mathbb{P}&(X_1=i_1,X_2=i_2,t=t'|X_1=X_2=\gamma,t=0)\\
&=\mathbb{P}(X_1=i_1,t=t'|X_1=\gamma,t=0)\mathbb{P}(X_2=i_2,t=t'|X_2=\gamma,t=0).
\eqn
Using the tensor formalism the transition probabilities can be expressed as
\beqn
&\mathbb{P}(X_1=i_1,t=t'|X_1=\gamma,t=0)=\sum_{k_1}m^{(1)}_{i_1k_1}(t')\delta_{k_1}^{\gamma},\nonumber\\
&\mathbb{P}(X_2=i_2,t=t'|X_2=\gamma,t=0)=\sum_{k_2}m^{(2)}_{i_2k_2}(t')\delta_{k_2}^{\gamma},\nonumber\\
&\mathbb{P}(X_1=i_1,X_2=i_2,t=t'|X_1=X_2=\gamma,t=0)\nonumber\\
&\hspace{100pt}
=\sum_{k_1,k_2,k_3}m^{(1)}_{i_1k_1}(t')m^{(2)}_{i_2k_2}(t')\delta_{k_3}^\gamma\Gamma_{k_3}^{k_1k_2}.\nonumber
\eqn
Implementing (\ref{condindep}) leads to the requirement that 
\beqn
\Gamma_\gamma^{k_1k_2}=\delta_{k_1}^\gamma\delta_{k_2}^\gamma,\nonumber
\eqn
and the basis dependent definition of the branching operator
\beqn
\delta e_i=e_i\otimes e_i.\nonumber
\eqn
From this construction we can express the phylogenetic tree Figure \ref{pic:4leaf} as
\beqn
P=(1\otimes 1\otimes M_3\otimes M_4)1\otimes 1\otimes\delta(M_1\otimes M_2\otimes M_5)1\otimes\delta\cdot\delta\pi,\nonumber
\eqn
which can also be written in the more convenient form
\beqn
P=(M_1\otimes M_2\otimes M_3\otimes M_4)1\otimes 1\otimes \delta(1\otimes 1\otimes M_5)1\otimes \delta\cdot\delta\pi.\nonumber
\eqn
This form can be generalized so that any phylogenetic tree can be expressed in the form
\beqn\label{convform}
P:=M_1\otimes M_2\otimes\ldots\otimes M_m\tilde{P},
\eqn
with $M_a\in GL(n)$, $1\leq a\leq m$, and $\tilde{P}$ is found by taking $P$ and setting the Markov operators on the leaf edges, $M_1,M_2,\ldots,M_m$, all equal to the identity operator. This representation will be of importance to us as we consider invariant theory in terms of phylogenetics.
\section{Entanglement and phylogenetics}
In this final section we will study the properties of a phylogenetic tensor evaluated on invariant functions of the general linear group. Recalling (\ref{transdet}), we see that in all reasonable cases the determinant of the transition matrices of a phylogenetic tree is non-zero. This implies that the transition matrices are elements of $GL(n)$. Thus in the case of a phylogenetic tensor of the form (\ref{convform}), an invariant will take the form
\beqn
f(P)=\prod_{a=1}^m \det(M_{a})^kf(\tilde{P}).\nonumber
\eqn
Presently we study the case where $|\mathcal{K}|=2$ and the phylogenetic tensor occurs in the tensor product space relevant to two qubits and three qubits respectively.
\subsection{Two qubits}
For the case of two qubits the most general phylogenetic tensor is given by
\beqn\label{eq:twotaxa}
P=(M_{1}\otimes M_{2})\delta\pi,
\eqn
which corresponds to the tree of Figure \ref{pic:fixtree}.
\begin{figure}[b]
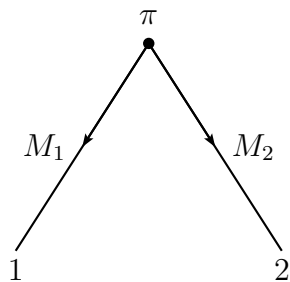

\centering
\pspicture[](0,0)(3,3)
\psset{linewidth=\pstlw,xunit=0.5,yunit=0.5,runit=0.5}
\psset{arrowsize=2pt 2,arrowinset=0.2}
\psline{->}(4,6)(2.25,3.25)
\psline{->}(4,6)(5.75,3.25)
\psline{-}(4,6)(0.5,0.5)
\psline{-}(4,6)(7.5,0.5)
\pscircle[linewidth=0.8pt,fillstyle=solid,fillcolor=black](4,6){0.15} 
\rput(4,6.7){$\pi$}
\rput(1.25,3.25){$M_1$}
\rput(6.75,3.25){$M_2$}
\rput(0.5,0){1}
\rput(7.5,0){2}
\endpspicture 
\caption{Phylogenetic tree with two leaves}
\label{pic:fixtree}
\end{figure}
Following (\ref{convform}) we have 
\beqn
\tilde{P}=\delta\pi.\nonumber
\eqn
As will be discussed in detail in Chapter \ref{chap4}, the concurrence can be used to establish the magnitude of divergence between a pair of sequences derived from a single branching event. The concurrence of the phylogenetic state (\ref{eq:twotaxa}) is given by
\beqn
\mathcal{C}(P)&=\det[M_{1}]\det[M_{2}]\mathcal{C}(\delta\pi).\nonumber
\eqn
Explicitly we have
\beqn
\mathcal{C}(\tilde{P})&=\sum \delta_{i_1i_2}\pi_{i_2}\delta_{i_3i_4}\pi_{i_4}\epsilon_{i_1i_3}\epsilon_{i_2i_4}\\\nonumber
&=\pi_1 \pi_2,
\eqn
and find that 
\beqn\label{concurrence1}
\mathcal{C}(P)&=\det[M_1]\det[M_2]\pi_1\pi_2.\nonumber
\eqn
Assuming that the determinants of the Markov operators are non-zero we see that the phylogenetic tensor is on the entangled orbit.\\
In comparison, if there is no stochastic dependence between the random variables the phylogenetic state can be expressed as
\beqn
P=p_1\otimes p_2,\nonumber
\eqn
which is a product state, such that the random variables $X_1$ and $X_2$ are stochastically independent, and the concurrence vanishes. Thus the non-vanishing of the concurrence can be used as a test of stochastic dependence between any two molecular sequences. In Chapter \ref{chap4} we will show that the determinants of the Markov operators tend to zero as $t$ tends to infinity and we conclude that the phylogenetic (\ref{eq:twotaxa}) state tends to a product state after an infinite amount of divergence. This is what one would expect as the case of infinite divergence should correspond exactly to the case of stochastic independence.
\subsection{Three qubits}
In this section we study the phylogenetic state
\beqn\label{eq:threetaxa}
P=(M_1\otimes M_2\otimes M_3)1\otimes\delta(1\otimes M_4)\delta\pi,
\eqn
which corresponds to the tree Figure \ref{pic:threetree}. Again following (\ref{convform}) we have 
\beqn
\tilde{P}=1\otimes\delta(1\otimes M_4)\delta\pi.\nonumber
\eqn
We now determine which orbit the phylogenetic state (\ref{eq:threetaxa}) lies in. By the general properties of the tangle we find that
\beqn
\mathcal{T}(P)=\prod_{i=1}^3(\det{M_i})^2\mathcal{T}(\tilde{P}),\nonumber
\eqn
and by explicit computation
\beqn
\mathcal{T}(\tilde{P})=(\det{M_4})^2(\pi_1\pi_2)^2,\nonumber
\eqn
to conclude that
\beqn
\mathcal{T}(P)=(\det{M_1}\det{M_2}\det{M_3}\det{M_4})^2(\pi_1\pi_2)^2.\nonumber
\eqn
From this we can conclude that the phylogenetic state (\ref{eq:threetaxa}) lies on the $GHZ$ orbit and the evaluation of the tangle upon three aligned sequence can be used as a test of triplet stochastic dependence.
\begin{figure}[t]
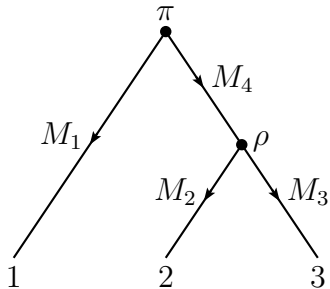

\centering
\pspicture[](0,0)(3,3)
\psset{linewidth=\pstlw,xunit=0.5,yunit=0.5,runit=0.5}
\psset{arrowsize=2pt 2,arrowinset=0.2}
\psline{->}(4,6)(2,3)
\rput(1.25,3.25){$M_1$}
\pscircle[linewidth=0.8pt,fillstyle=solid,fillcolor=black](4,6){0.15}
\rput(4,6.5){$\pi$} 
\psline{-}(2.2,3.3)(0,0)
\rput(0,-.5){1}
\psline{->}(4,6)(5,4.5)
\rput(5.75,4.75){$M_4$}
\psline{-}(4.8,4.8)(6,3)
\pscircle[linewidth=0.8pt,fillstyle=solid,fillcolor=black](6,3){0.15}
\rput(6.5,3.1){$\rho$}
\psline{->}(6,3)(5,1.5)
\rput(4.25,1.75){$M_2$}
\psline{-}(5.2,1.8)(4,0)
\rput(4,-.5){2}
\psline{->}(6,3)(7,1.5)
\rput(7.75,1.75){$M_3$}
\psline{-}(6.8,1.8)(8,0)
\rput(8,-.5){3}
\endpspicture  
\caption{Phylogenetic tree with three leaves}
\label{pic:threetree}
\end{figure}

\subsection{Phylogenetic relation}
Referring to (\ref{transdet}), we see that for continuous time Markov chains the determinants of the transition matrices satisfy:
\beqn\label{limits}
0<\det{M(s,t)}&\leq 1,\quad\forall 0\leq t< \infty,\\
\lim_{t\rightarrow\infty}\det M(t)&=0.
\eqn
Above we have seen that for phylogenetic data of three aligned sequences derived from a tree the tangle polynomial is non-zero, and for two aligned sequences derived from a tree the concurrence is also non-zero. But taking (\ref{limits}) into account we see that, if any one of the branches of a phylogenetic tree is extended to infinite length this will induce the vanishing of these invariant functions which implies that the corresponding part of the phylogenetic tensor decouples from the overall state to form a partial product state. Thus the case of no stochastic dependence directly corresponds to entanglement of the tensor state and stochastic dependence can be tested for using invariant functions.\\
Introducing independent time parameters for each external branch, we can express the phylogenetic tree (\ref{pic:threetree}) as
\beqn
P(t_1,t_2,t_3):=[M_1(0,t_1)\otimes M_2{(0,t_2)}\otimes M_3(0,t_3)]1\otimes\delta[1\otimes M_4]\delta\pi.\nonumber
\eqn
Now, as we have seen, the tangle polynomial will satisfy
\beqn
\lim_{t_a\rightarrow\infty}&\mathcal{T}(P(t_1,t_2,t_3))=0,\\
&\forall a=1,2,3.\nonumber
\eqn
For the concurrence we have
\beqn
\lim_{t_a\rightarrow\infty}&\mathcal{C}_b(P(t_1,t_2,t_3))=0,
\eqn
if and only if $a=b$. From these observations we can conclude that we have the limit:
\beqn
\lim_{t_1\rightarrow\infty}p_{i_1i_2i_3}(t_1,t_2,t_3)=p^{(1)}_{i_1}p^{(23)}_{i_2i_3}(t_2,t_3),\nonumber
\eqn 
and similar for $t_2,t_3$. The phylogenetic state decouples into a partial product state after an infinite amount of stochastic divergence. This is what one would expect, as the branch lengths of the tree become so large that it is impossible to observe the branching event which relates to leaves.\\
From these observations we define a \textit{phylogenetic relation} to exist whenever the relevant phylogenetic tensor cannot be written as a product state.
\section{Closing remarks}
In this chapter we have established the mathematical connection between the notion of entanglement and that of phylogenetic relation. We showed that simple group invariant functions used to quantify entanglement can be utilized in the phylogenetic case. We focused on the invariant function known as the tangle, but considered only the case of two character states. In the next chapter we will study the properties of the tangle in the case of three and four character states.

\chapter{Using the tangle}\label{chap4}
The distance based approach to phylogenetic reconstruction using the neighbor joining algorithm is a commonly used technique \cite{gascuel1997,lake1994,pearson1999,saitou1987}. Under the assumptions of a Markov model of sequence evolution, the phylogenetic relationship is uniquely reconstructible from (suitably defined) pairwise distances \cite{steel1998}. The approach relies crucially upon the calculation of distance matrices from aligned sequence data which give a measure of the pairwise evolutionary distance between the extant taxa under consideration. As far as tree building algorithms are concerned it is required that the distances are strictly \textit{linearly} related to the sum of the (theoretical) edge lengths of the phylogenetic tree, and that the parameters of the linear relation do not vary across the tree. It is essential to the analysis that the measure of distance chosen has both biological and statistical as well as mathematical significance. If one assumes the standard Markov model, the edge lengths of a phylogenetic tree can be taken mathematically to be a quantity that we refer to as the \textit{stochastic distance}. (For mathematical discussion of this quantity see Goodman \cite{goodman1973} who refers to the stochastic distance as \textit{intrinsic time}, and see also Barry and Hartigan \cite{barry1987} who gave a biological interpretation.) Under the assumptions of a general Markov model the $\log\det$ formula is commonly used to obtain pairwise distances. Further, if one may assume a stationary process then the $\log\det$ formula can be modified to give an estimate of the actual stochastic distance \cite{lockhart1994}. (That is, the constants of the linear relation are set by the stationarity assumption.) \\
Distance based methods and, consequently, the $\log\det$ formula are often used in favour of other methods (such as maximum likelihood) in cases where there has been significant compositional heterogeneity during the evolutionary history. The theoretical basis which motivates this usage was presented by Steel \cite{steel1994} and is discussed in Lockhart, Steel, Hendy and Penny \cite{lockhart1994} and Gu and Li \cite{gu1996}. More recently, Jermiin, Ho, Ababneh, Robinson and Larkum published a simulation study which confirms that the $\log\det$ outperforms other techniques in this case \cite{jermiin2004}. Lockhart \textit{et al.} showed that by using the assumption that the base composition remains close to constant, the $\log\det$ formula can be modified to give an estimate of the actual stochastic distance.  However, as will be shown, in both its original and modified form the $\log\det$ formula includes an approximation crucially dependent upon the compositional heterogeneity remaining minimal. The effectiveness of the $\log\det$ formula to correctly reconstruct the phylogenetic history when there has been significant compositional heterogeneity is thus brought into question. Hence there is a contradictory state of affairs between the theoretical basis of the $\log\det$ and the circumstances under which it is implemented. In this chapter we will generalize the $\log\det$ formula in such a way that this dependence upon base composition is truly absent.\\
A disadvantage of the $\log\det$ formula is that it uses only \textit{pairwise} sequence data and is blind to the fact that extra information regarding pairwise distances can be obtained from the sequence data of additional taxa. Felsenstein \cite{felsenstein2004} mentions that it is surprising that distance techniques work at all given that they ignore the extra information in higher order alignments. This chapter details exactly how the $\log\det$ formula can be improved upon by taking functions of aligned sequence data for \textit{three} taxa at a time. It may seem counter-intuitive that consideration of a third taxon can impart information regarding the evolutionary distance between two taxa, but it is the case that by considering a third taxon the $\log\det$ formula can be refined. This result depends crucially upon the fact that, as is somewhat trivially the case for two taxa, there is only one possible (unrooted) tree topology relating three taxa. (For discussion of what a tree topology is see \cite{nei2000}, Chapter 5.) It is possible to refine the $\log\det$ formula by considering the respective distance to an arbitrary third taxon. The reader should note that the use of triplet sequence data to the problem of reconstruction of the Markov model was also considered in \cite{chang1996} and \cite{pearl1986}. The approach discussed in the present chapter is original in the sense that triplets of the aligned sequences are being used explicitly in a distance method, and follows on from the theoretical discussions of \cite{sumner2005}.
\\
A complication arises regarding the total stochastic distance between leaves and the placement of the root of a phylogenetic tree. It turns out that if we define phylogenetic trees of identical topology to be equivalent if they give the identical probability distributions then we find that the total stochastic distance between leaves is not, in general, left unchanged as we move the root of the tree. The so defined equivalence class provides a generalization of Felsenstein's \textit{pulley principle} \cite{felsenstein1981} and was first presented in Steel, Szekely and Hendy \cite{steel1994b}. The fact that the stochastic distance is not left unchanged is a surprising result and has important implications regarding the interpretation of the edge lengths of phylogenetic trees defined under the Markov model. In particular this result implies that the $\log\det$ technique is an inconsistent estimator of pairwise distances on phylogenetic trees. It is the purpose of this chapter to present a new estimator that is consistent in the case of phylogenetic quartets. We are motivated to present this construction of quartet distance matrices by the interest in phylogenetic reconstruction of large trees from the correct determination of the set of $\binom{n}{4}$ quartets \cite{bryant2001,strimmer1996}. 
\\
This chapter will begin by formally defining the stochastic distance. We will then examine how the general linear group invariants, the $det$ (\ref{concurrence}) and the $tangle$ (\ref{tangle4}), can used to estimate the stochastic distance between any two taxa on a phylogenetic tree. As a consequence of this discussion we will examine a generalized pulley principle and finish by showing that by including the tangle in the analysis we can arrive at a consistent estimator. 

\noindent
\textbf{Note:} This chapter follows closely the text of \cite{sumner2006}.

\subsection{Stochastic distance}

In this chapter we will be interested in the assignment of edge lengths to phylogenetic trees. To this end we consider the rate of change of base changes at time $s$: \footnote{It is standard to include a factor of $n^{-1}$ in this definition. However, this factor clutters the consequent formulae and here we do not include it as it has no consequence to the forgoing discussion and can always be incorporated into the analysis later.}
\beqn
\lambda(s):=\sum_{i\in\mathcal{K}}\frac{\partial\mathbb{P}(X(t)=i|X(s)\neq i)}{\partial t}|_{t=s}.\nonumber
\eqn
By considering (\ref{ratespars}) and (\ref{kalamorov}) this quantity can be explicitly expressed using the rate parameters: 
\beqn
\lambda(s)&=-\sum_{i}q_{ii}(s),\\\nonumber
&=-trQ(s).
\eqn
From these considerations we define the \textit{stochastic distance} to be given by the expression
\beqn
\omega(s,t):=\int_s^t\lambda(u)du.\nonumber
\eqn
By considering the time-ordered product representation (\ref{inhom}) and the Jacobi identity $\det e^X=e^{trX}$, we find that the stochastic distance can be directly related to the transition probabilities of the Markov chain:
\beqn\label{omega}
\omega(s,t)=-\log\det M(s,t).
\eqn
Our assignment of edge lengths will take the Markov matrix associated with each edge and set the edge length equal to the stochastic distance. 
\\
The relation (\ref{omega}) is known in various guises in both the mathematical and phylogenetic literature \cite{barry1987,goodman1973} and, as will be confirmed in the next section, is the basis of the $\log\det$ formula. It should also be noted that (\ref{omega}) will remain positive and finite because $\omega(s,s)=0$, $\lambda(s)\geq 0$ and the integral $\int_0^T \lambda(t)dt$ is not expected to diverge.\footnote{There are two cases where the integral may diverge, but we can safely exclude these possibilities as follows. i. $\lambda(t)$ may be a badly behaved function. We can reject this possibility outright in phylogenetics as there is every reason to expect the rate parameters to change smoothly with time. ii. $T\rightarrow\infty$. We can safely ignore this possibility as we will be assuming that the divergence times of the Markov chain are sufficiently small such that the phylogenetic historical signal is still obtainable.}

\subsection{Observability of the stochastic distance}
An interesting consideration (which at first sight is at odds with our aims) is that given a single random process modelled as a CTMC there is simply no way of inferring the value of the stochastic distance from an observed distribution without making restrictive assumptions about the process and the initial distribution. This is best illustrated by considering a stationary CTMC for which the rate-parameters are time-independent and given an initial distribution $\pi_i(0)$ satisfy
\beqn
\sum_{j}q_{ij}\pi_{j}(0)=0,\qquad\forall i.\nonumber
\eqn
Now, although the consequent distribution is time-independent, $\pi_i(t)=\pi_i(0)$, and hence carries zero informative value in comparison to the initial distribution, the stochastic distance itself increases linearly with time
\beqn
\omega(0,t)=-\left(\sum_{i}q_{ii}\right) t.\nonumber
\eqn
From this observation it is clear that in the general case if all we have access to is the final distribution, there is no way we can estimate the stochastic distance unless we make some additional assumptions about the stochastic process.\\
The remarkable fact is that in the case of phylogenetics it is possible to estimate the stochastic distance from the observed distribution. (As we will show in Section \ref{pairwise}, this is true even for the case where the underlying chains are stationary!)

\section{Pairwise distance measures}\label{pairwise}

In this section we will derive and discuss a standard approach to the construction of distance matrices. (For an excellent perspective of the various measures of phylogenetic pairwise distance see \cite{baake1999}.) A distance matrix, $\phi=[\phi_{ab}]_{(a,b)\in L}$, is constructed from the aligned sequence data of multiple extant taxa such that each entry gives a suitable estimate of the distance between a given pair of taxa. The mathematical conditions on the $\phi_{ab}$ are the standard conditions of a distance function as well as the four point condition \cite{steel1998} (which is required for the distance measure to be consistent with the tree structure):
\beqn\label{distance}
\phi_{ab}&\geq 0,\\
\phi_{ab}&=0\text{ iff }a=b,\\
\phi_{ab}&=\phi_{ba},\\
\phi_{ab}+\phi_{cd}&\leq\textit{max}\{\phi_{ac}+\phi_{bd},\phi_{ad}+\phi_{bc}\};\qquad\forall\ a,b,c,d\in L.
\eqn
There are no further conditions required upon $\phi$ for it to give a unique tree reconstruction \cite{steel1998}. However it is of course desirable for the distance measure to have a well defined biological interpretation. To this end, for a given edge $e$, we define the \textit{edge length}, $\omega_e$, which we set to be the stochastic distance (\ref{omega}) taken from the Markov model:
\beqn
\omega_e=-\log\det M_e.\nonumber
\eqn
It is then apparent that any significant estimate of pairwise distance must statistically be expected to converge to a value which is linearly related to the sum of the stochastic distances lying on the (unique) path between the two taxa under consideration. It should be clear that such a measure will satisfy the relations (\ref{distance}). It is crucial to the performance of the distance measure under a tree building algorithm that the parameters of the linear relation are expected to be \textit{constant} for all pairs of taxa. That is, given the unique path between leaf $a$ and $b$, $P(T;a,b)$, we are demanding that statistically we have the following convergence:
\beqn
\phi_{ab}\rightarrow\alpha\omega{(a,b)}+\beta,\nonumber
\eqn
where
\beqn
\omega{(a,b)}:=\sum_{e\in P(T;a,b)}\omega_e,\nonumber
\eqn
and $\alpha$ and $\beta$ are expected to be independent of $a$ and $b$. As we will see, the $\log\det$ formula does not satisfy this property for the most general models.
\\
\subsection{The $\log\det$ formula}
In Figure \ref{pic:twotaxa} we consider the two taxa phylogenetic tree, with pattern probabilities given by
\beqn\label{twotaxa}
p_{i_1i_2}=\sum_{j}m^{(1)}_{i_1j}m^{(2)}_{i_2j}\pi_j.
\eqn
\begin{figure}[t]
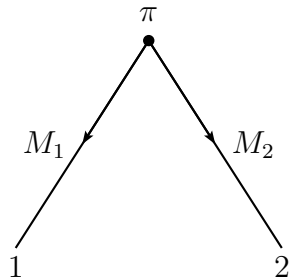

\centering
\pspicture[](0,0)(3,3)
\psset{linewidth=\pstlw,xunit=0.5,yunit=0.5,runit=0.5}
\psset{arrowsize=2pt 2,arrowinset=0.2}
\psline{->}(4,6)(2.25,3.25)
\psline{->}(4,6)(5.75,3.25)
\psline{-}(4,6)(0.5,0.5)
\psline{-}(4,6)(7.5,0.5)
\pscircle[linewidth=0.8pt,fillstyle=solid,fillcolor=black](4,6){0.15} 
\rput(4,6.7){$\pi$}
\rput(1.25,3.25){$M_1$}
\rput(6.75,3.25){$M_2$}
\rput(0.5,0){1}
\rput(7.5,0){2}
\endpspicture 
\caption{Phylogenetic tree of two taxa}
\label{pic:twotaxa}
\end{figure}
By considering the matrices defined as 
\beqn
P^{(1,2)}:&=\left[p_{ij}\right]_{(i,j)\in\mathcal{K}},\\
D_\pi:&=\left[diag(\pi_i)\right]_{i\in\mathcal{K}};\nonumber
\eqn
 it is easy to show that (\ref{twotaxa}) is equivalent to
\beqn
P^{(1,2)}=M_1D_\pi M_2^t.\nonumber
\eqn
Taking the determinant of this expression and considering (\ref{omega}) yields
\beqn\label{conc}
\det P^{(1,2)} &=\det M_1\det M_2\det D_\pi\\
&=e^{-(\omega_1+\omega_2)}\prod_i\pi_i.
\eqn
This expression can be generalized to the case of any two taxa from a given phylogenetic tree:
\beqn\label{det}
\det P^{(a,b)}=e^{-\omega(a,b)}\prod_i\pi_i^{(a,b)},
\eqn
where $\pi_i^{(a,b)}$ is the distribution at the most recent ancestral vertex between taxa $a$ and $b$ determined by the meeting point of the two paths traced backwards along the phylogenetic tree from leaf $a$ and $b$.\\
Now $\omega(a,b)$ is theoretically equal to the total stochastic distance between each of $a$ and $b$ and their most recent ancestral vertex and hence it is clear that $-\log\det P^{(a,b)}$ will be linearly related to this quantity. In the original formulation of the $\log\det$, a distance measure between two taxa was defined as
\beqn\label{logdet}
d_{ab}:&=-\log\det P^{(a,b)}\\
&=\omega(a,b)-\sum_i\log[\pi_i^{(a,b)}],
\eqn
and shown to satisfy the conditions (\ref{distance}) \cite{steel1994}. From this relation it seems that one can take $\alpha=1$ and $\beta=-\sum_i\log[\pi_i^{(a,b)}]$  and evaluate (\ref{logdet}) on the observed pattern frequencies for each pair of taxa to calculate a well defined distance matrix from a set of aligned sequence data (as was presented in \cite{lockhart1994}). This procedure depends crucially upon the shifting term $\beta=\sum_i\log[\pi_i^{(a,b)}]$ being independent of $a$ and $b$. However, this is only true in special circumstances such as star phylogeny or if the base composition is constant (the stationary model). In the general case, one is led to a different shifting term depending on the topology of the tree (this was noted in Sumner and Jarvis \cite{sumner2005} and we reproduce the result here). Consider the phylogenetic tree of three taxa given in Figure \ref{pic:threetaxaroot}
with pattern probabilities given by
\beqn\label{threetaxa}
p_{i_1i_2i_3}=\sum_{j,k}m^{(1)}_{i_1j}m^{(2)}_{i_2k}m^{(3)}_{i_3k}m^{(4)}_{kj}\pi_j.\nonumber
\eqn
By calculating (\ref{logdet}) for the three possible pairs of taxa we find that
\beqn
d_{12}&=(\omega_1+\omega_4+\omega_2)-\sum_i\log\pi_i,\\\nonumber
d_{13}&=(\omega_1+\omega_4+\omega_3)-\sum_i\log\pi_i,\\
d_{23}&=(\omega_1+\omega_3)-\sum_i\log\rho_i,
\eqn
from which it is explicitly clear that the shifting term is \textit{not} constant across this phylogenetic tree. The shifting term is dependent on the base composition at the most recent ancestral node of the two taxa and from the above example it is clear that this depends on the topology of the tree and is not always simply the root of the tree. This means that (\ref{logdet}) does not produce distance matrices whose entries are linearly related to the edge length of the tree because the entries of the matrix will depend essentially upon the topology of the tree. 
\begin{figure}[t]
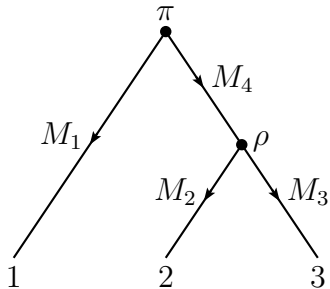

\centering
\pspicture[](0,0)(3,3)
\psset{linewidth=\pstlw,xunit=0.5,yunit=0.5,runit=0.5}
\psset{arrowsize=2pt 2,arrowinset=0.2}
\psline{->}(4,6)(2,3)
\rput(1.25,3.25){$M_1$}
\pscircle[linewidth=0.8pt,fillstyle=solid,fillcolor=black](4,6){0.15}
\rput(4,6.5){$\pi$} 
\psline{-}(2.2,3.3)(0,0)
\rput(0,-.5){1}
\psline{->}(4,6)(5,4.5)
\rput(5.75,4.75){$M_4$}
\psline{-}(4.8,4.8)(6,3)
\pscircle[linewidth=0.8pt,fillstyle=solid,fillcolor=black](6,3){0.15}
\rput(6.5,3.1){$\rho$}
\psline{->}(6,3)(5,1.5)
\rput(4.25,1.75){$M_2$}
\psline{-}(5.2,1.8)(4,0)
\rput(4,-.5){2}
\psline{->}(6,3)(7,1.5)
\rput(7.75,1.75){$M_3$}
\psline{-}(6.8,1.8)(8,0)
\rput(8,-.5){3}
\endpspicture  
\caption{Phylogenetic tree of three taxa}
\label{pic:threetaxaroot}
\end{figure}
\\
It is, however, possible to obtain an estimate of the total stochastic distance between any two taxa by modifying the $\log\det$ formula. The ancestral base composition is approximated by using the harmonic mean
\beqn\label{harmonicmean}
\prod_i\pi^{(a,b)}_i\approx[\prod_{i_1,i_2}\pi^{(a)}_{i_1}\pi^{(b)}_{i_2}]^{\frac{1}{2}},
\eqn
where $\pi^{(a,b)}_k$ is the closest common ancestral base composition between taxa $a$ and $b$ and $\pi^{(a)}_i:=\mathbb{P}(X_a(\tau_a)=i)$ (and similarly for $b$). One is then led to the formula
\beqn\label{logdetharm}
d'_{ab}:=-\log\det P^{(a,b)}+\fra{1}{2}\sum_{i_1,i_2}(\log\pi^{(a)}_{i_1}+\log\pi^{(b)}_{i_2}),\qquad\forall\ a,b\in L.
\eqn
where $d'_{ab}$ is then an estimator of the total stochastic distance between taxa $a$ and $b$. (This form of the $\log\det$ formula was presented in \cite{lockhart1994} and \cite{steel1998}).
\\
In the case of a stationary base composition model the additional assumption is made that
\beqn
\sum_{j}m^{(e)}_{ij}\pi_j=\pi_i;\qquad\forall\ e\in E.\nonumber
\eqn
In this case we have
\beqn
\pi_i^{(a,b)}=\pi^{(a)}_i=\pi^{(b)}_i,\qquad\forall\ a,b\in L,\nonumber
\eqn
and it is clear that the harmonic mean approximation becomes an exact relation and the $\log\det$ formula is expected to converge exactly to the total stochastic distance between the two taxa.

\subsection{The tangle}

In this section we will show how the $\log\det$ formula can be generalized to obtain, for the most general Markov models,  an unbiased estimate of the distance matrix. The basis of the technique is the existence of a measure analogous to (\ref{conc}) which is valid for \textit{triplets}. \\
Sumner and Jarvis \cite{sumner2005} presented a polynomial function $\mathcal{T}$ which is known in quantum physics as the tangle and can be evaluated on phylogenetic data sets of three aligned sequences in the case of $n=2$. Evaluated on the pattern probabilities of any phylogenetic tree of three taxa, $\{a,b,c\}$, the tangle takes on the theoretical value 
\beqn\label{tangle}
\mathcal{T}(a,b,c)=e^{-2\omega(a,b,c)}\left(\prod_{i\in\mathcal{K}}\pi_i\right)^2,
\eqn
where 
\beqn
\omega{(a,b,c)}:=\sum_{e\in T}\omega_e,\nonumber
\eqn
$\pi$ is the common ancestral root of the three taxa and this relation holds independently of the particular tree topology which relates $\{a,b,c\}$. This independence upon the topology is a very nice property and is crucial to the practical use of the tangle as a distance measure. The similarity between (\ref{tangle}) and (\ref{det}) should be noted.\\
In this chapter we report generalized tangles, which are polynomials which satisfy (\ref{tangle}) for the cases of $n=3,4$ in addition to the $n=2$ case which was presented in \cite{sumner2005}. It is possible to infer the existence of the tangles and derive their polynomial form from group theoretical considerations. Here we give forms using the completely antisymmetric (Levi-Civita) tensor, $\epsilon$, which has components $\epsilon_{i_1i_2...i_n}$ and satisfies $\epsilon_{12...n}=1$. For the cases of $n=2,3,4$ the tangles are given by\footnote{This expression for $\mathcal{T}_2$ corrects for the erroneous expression presented in \cite{sumner2005}.}

\noindent
$
\mathcal{T}_2=\fra{1}{2!}\sum_1^2 p_{i_1i_2i_3}p_{j_1j_2j_3}p_{k_1k_2k_3}p_{l_1l_2l_3}\epsilon_{i_1j_1}\epsilon_{i_2j_2}\epsilon_{k_1l_1}\epsilon_{k_2l_2}\epsilon_{i_3l_3}\epsilon_{j_3k_3},$

\noindent
$
\mathcal{T}_3=\fra{1}{3!}\sum_1^3p_{i_1i_2i_3}p_{j_1j_2j_3}p_{k_1k_2k_3}p_{l_1l_2l_3}p_{m_1m_2m_3}p_{n_1n_2n_3}$\\ 
\hspace*{140pt}$\cdot\epsilon_{i_1j_1k_1}\epsilon_{j_2k_2l_2}\epsilon_{k_3l_3m_3}\epsilon_{l_1m_1n_1}\epsilon_{m_2n_2i_2}\epsilon_{n_3i_3j_3},$

\noindent
$
\mathcal{T}_4=\fra{1}{4!}\sum_1^4p_{i_1j_1k_1}p_{i_2j_2k_2}p_{i_3j_3k_3}p_{i_4j_4k_4}p_{i_5j_5k_5}p_{i_6j_6k_6}p_{i_7j_7k_7}p_{i_8j_8k_8}$\\ 
\hspace*{140pt}$\cdot\epsilon_{i_1i_2i_3i_4}\epsilon_{i_5i_6i_7i_8}\epsilon_{j_1j_5j_4j_8}\epsilon_{j_2j_6j_3j_7}\epsilon_{k_1k_5k_2k_6}\epsilon_{k_3k_7k_4k_8};
$

\noindent
respectively, (where the summation is over every index). The expression (\ref{tangle}) can be proved by studying the group theoretical properties of the tangle (see \cite{sumner2005}) and by explicitly expanding the above forms. For the tangle on two characters we find
\beqn
\mathcal{T}_2=&-p_{122}^2p_{211}^2 + 2p_{121}p_{122}p_{211}p_{212} - p_{121}^2p_{212}^2 + 2p_{112}p_{122}p_{211}p_{221} +\nonumber
\\& 2p_{112}p_{121}p_{212}p_{221} - 4p_{111}p_{122}p_{212}p_{221} - p_{112}^2p_{221}^2 - 4p_{112}p_{121}p_{211}p_{222} +\\& 
2p_{111}p_{122}p_{211}p_{222} + 2p_{111}p_{121}p_{212}p_{222} + 
2p_{111}p_{112}p_{221}p_{222} - p_{111}^2p_{222}^2.
\eqn
Substantial computer power is required to explicitly compute $\mathcal{T}_3$ and $\mathcal{T}_4$. These polynomials have 1152 and 431424 terms, respectively. 
\subsection{Star topology}
Consider the phylogenetic tree relating three taxa with a star topology:
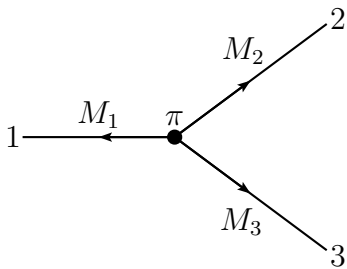
\beqn
\\
\\
\pspicture[](5,0)(1,2.5)
\psset{linewidth=\pstlw,xunit=0.5,yunit=0.5,runit=0.5}
\psset{arrowsize=2pt 2,arrowinset=0.2}
\psline{-}(4,3)(0,3)
\rput(-.25,3){1}
\psline{-}(4,3)(8,0)
\rput(8.3,-0.15){3}
\psline{-}(4,3)(8,6)
\rput(8.3,6.15){2}
\psline{->}(4,3)(2,3)
\rput(2,3.6){$M_1$}
\psline{->}(4,3)(6,4.5)
\rput(5.8,5.3){$M_2$}
\psline{->}(4,3)(6,1.5)
\rput(5.75,.75){$M_3$}
\pscircle[linewidth=0.8pt,fillstyle=solid,fillcolor=black](4,3){.2} 
\rput(4,3.5){$\pi$}
\endpspicture  \nonumber
\\
\\
\eqn
with pattern probabilities given by the formula
\beqn
p_{i_1i_2i_3}=\sum_{j}m^{(1)}_{i_1j}m^{(2)}_{i_2j}m^{(3)}_{i_3j}\pi_j.\nonumber
\eqn
 Here we will use the fact that the root of this tree is also the common ancestral root of any pair of the three taxa. (This is not the case in general if we allow for a general rooting of the tree and/or more than three taxa. The complications arising in these cases will be dealt with in the next section.)
 \\
Considering the formulae (\ref{tangle}) and (\ref{conc}) we are led to introduce the novel distance matrix, $\Delta$, with the pairwise distance between $\{a,b\}$ given by
\beqn\label{dist}
\Delta^{(c)}_{ab}:=-\log\mathcal{T}(a,b,c)+\log{\det{P^{(a,c)}}+\log\det{P}^{(b,c)}},\qquad a,b,c\in L.
\eqn
From (\ref{conc}) and (\ref{tangle}) it follows that 
\beqn
\Delta^{(c)}_{ab}=\omega(a,b),\nonumber
\eqn
such that our new formula will directly give the stochastic distance between the two taxa. There is no need to make the harmonic mean approximation and this distance measure is mathematically and biologically meaningful. This is the main result of this chapter: given a set of aligned sequence data, the tangle formula (\ref{dist}) can be used to compute the \textit{exact} pairwise edge lengths for any triplet. As mentioned above, the explicit polynomial form of the tangle has been computed for the cases of two, three and four bases and it is our intent that (\ref{dist}) will provide a significant improvement over the $\log\det$ formula in the calculation of pairwise distance matrices for these cases.

\subsection{Summary}
Considering the stochastic distance to be the correct way to assign edge lengths to branches of a phylogenetic tree, we have reviewed three different ways of obtaining a distance measure between any two taxa $a$ and $b$:
\begin{enumerate}
\item {$d_{ab}=-\log \det P^{(a,b)}$}
\item {$d'_{ab}=-\log\det P^{(a,b)}+\frac{1}{2}\sum_{i_1,i_2}(\log\pi^{(a)}_{i_1}+\log\pi^{(b)}_{i_2})$}
\item {$\Delta^{(c)}_{ab}=-\log\mathcal{T}(a,b,c)+\log{\det{P^{(a,c)}}+\log\det{P}^{(b,c)}}$}
\end{enumerate}
where one substitutes the observed pattern frequencies into these expressions. From the previous considerations we found that these three distance measures have the following properties:
\begin{enumerate}
\item When $d_{ab}$ is evaluated on a set of observed pattern frequencies, this estimator satisfies the requirements of a distance function (\ref{distance}), but is inconsistent with the general Markov model as the estimate is \textit{not} expected to converge to a value that is linearly related to $\omega(a,b)$.
\item When $d'_{ab}$ is evaluated on a set of observed pattern frequencies, this estimator satisfies the requirements (\ref{distance}) and is expected to converge to a value that is linearly related to $\omega(a,b)$ whenever the compositional heterogeneity is absent. In the heterogeneous case this quantity approximates $\omega(a,b)$ by using (\ref{harmonicmean}).
\item When $\Delta^{(c)}_{ab}$ is evaluated on a set of observed pattern frequencies, this estimator satisfies the requirements of (\ref{distance}) and is expected to converge \text{exactly} to $\omega(a,b)$ in all cases.
\end{enumerate}\noindent
Thus we see that the tangle formula (\ref{dist}) should be a significant improvement as an empirical estimator of $\omega(a,b)$ upon both forms of the $\log\det$ formula. However, the formula (\ref{dist}) depends on taking an arbitrary third taxon, $c$. The question remains as to what to do in the case of constructing pairwise distances for sets of greater than three taxa. The surprising answer to this question will be addressed in the next section where we will bring into question the uniqueness of the theoretical quantity $\omega(a,b)$. The discussion has consequences for the interpretation of each of the estimators of pairwise distances that we have discussed.
\\

\section{Generalized pulley principle}

In this section we generalize the Felsenstein's pulley principle \cite{felsenstein1981}. In its original formulation the pulley principle describes the unrootedness of phylogenetic trees where the underlying Markov model is assumed to be reversible and stationary. Here we show how the pulley principle may be generalized to remain valid under the most general Markov models. Our immediate motivation is to show that (\ref{dist}) remains a valid distance measure under the circumstance of a general phylogenetic tree of multiple taxa. Unfortunately this generalization introduces surprising mathematical complications which have consequences not only for our formula (\ref{dist}), but also for the $\log\det$ technique and any other estimate of the stochastic distance upon a phylogenetic tree. The discussion will lead to the consequence that, for a given tree topology, there are multiple -- actually, infinitely many -- phylogenetic trees with identical probability distributions. (These phylogenetic trees differ by arbitrary rerootings and consequential redirection of edges.)  We will see that the generalized pulley principle shows that as far as inference from the observed pattern frequencies is concerned, there is no theoretical justification behind specifying the root of a phylogenetic tree if the most general Markov model is allowed. Also, we will see that the theoretical value of the stochastic distance is not constant for arbitrary rerootings of a phylogenetic tree. Clearly, if the stochastic distance is not uniquely defined theoretically, then one must be careful in interpreting any formula that gives an estimate thereof from the observed data.\\
Considering a phylogenetic tree as a directed graph shows that a rerooting involves redirecting an edge (or part thereof). The property required is that the Markov chain on the involved edge is taken to progress as if time has been reversed, and we refer to the new chain as the \textit{time-reversed} chain. This should be compared to the requirement of \textit{reversibility} as defined in the mathematical literature, (for example see \cite{isoifescu1980}). In the case of a stationary and reversible Markov chain the time-reversed chain (as we will define) is identical to the original chain.
\\ 
By way of example, we take the rooted tree of three taxa (\ref{threetaxa}) and redirect the relevant internal edge to give the following rerooting:
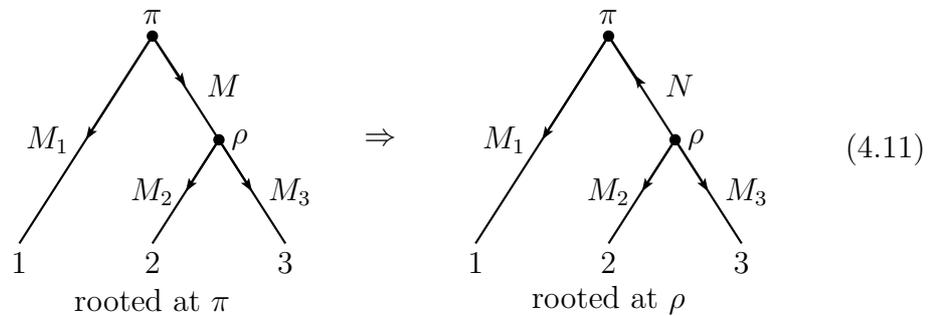
\beqn
\\
\\
\pspicture[](8.75,0)(1,2.5)
\psset{linewidth=\pstlw,xunit=0.5,yunit=0.5,runit=0.5}
\psset{arrowsize=2pt 2,arrowinset=0.2}
\psline{->}(4,6)(2.25,3.25)
\psline{->}(4,6)(4.875,4.675)
\psline{-}(4,6)(5.75,3.25)
\psline{-}(4,6)(0.50,0.50)
\psline{->}(5.75,3.25)(6.625,1.875)
\psline{-}(5.75,3.25)(7.50,0.50)
\psline{->}(5.75,3.25)(4.875,1.875)
\psline{-}(5.75,3.25)(4,0.5)
\pscircle[linewidth=0.8pt,fillstyle=solid,fillcolor=black](4,6){.15}
\pscircle[linewidth=0.8pt,fillstyle=solid,fillcolor=black](5.75,3.25){.15}
\rput(4,6.5){$\pi$}
\rput(1.25,3.25){$M_1$}
\rput(5.875,4.675){$M$}
\rput(4,1.875){$M_2$}
\rput(6.3,3.25){$\rho$}
\rput(7.625,1.875){$M_3$}
\rput(0.5,0){1}
\rput(4,0){2}
\rput(7.5,0){3}
\rput(4,-1){\text{rooted at $\pi$}}
\rput(16,-1){\text{rooted at $\rho$}}
\rput(10,3.25){$\Rightarrow$}
\psline{->}(16,6)(14.25,3.25)
\psline{-<}(16,6)(16.875,4.675)
\psline{-}(16,6)(17.75,3.25)
\psline{-}(16,6)(12.50,0.50)
\psline{->}(17.75,3.25)(18.625,1.875)
\psline{-}(17.75,3.25)(19.50,0.50)
\psline{->}(17.75,3.25)(16.875,1.875)
\psline{-}(17.75,3.25)(16,0.5)
\pscircle[linewidth=0.8pt,fillstyle=solid,fillcolor=black](16,6){.15}
\pscircle[linewidth=0.8pt,fillstyle=solid,fillcolor=black](17.75,3.25){.15}
\rput(16,6.5){$\pi$}
\rput(13.25,3.25){$M_1$}
\rput(17.875,4.675){$N$}
\rput(16,1.875){$M_2$}
\rput(18.3,3.25){$\rho$}
\rput(19.625,1.875){$M_3$}
\rput(12.5,0){1}
\rput(16,0){2}
\rput(19.5,0){3}
\endpspicture  
\label{pic:reroot}
\\
\\
\eqn
Our immediate task is to infer the existence of an appropriate time-reversed Markov chain, $N$, such that these two phylogenetic trees give identical probability distributions. If we equate the pattern probabilities of (\ref{pic:reroot}) and contract all edges except the one we are reversing, we are led to the simple algebraic solution
\beqn\label{reversed}
n_{ij}=\frac{m_{ji}\pi_i}{\rho_j}.
\eqn 
(This solution was presented in \cite{steel1994b}.) Presently we use this result to give an explicit form in the general case.
\\
Given a CTMC $X(t)$ with transition probabilities 
\beqn
m_{ij}(t,s):=\mathbb{P}(X(t)=i|X(s)=j),\nonumber
\eqn
we wish to find a second CTMC, $Y(t)$, such that, given any $T\geq 0$, we have
\beqn
\mathbb{P}(Y(t)=i)=\pi_i(T-t),\qquad\forall\quad 0\leq t\leq T.\nonumber
\eqn
That is, if the direction of time is reversed, the second CTMC $Y(t)$ has identical distribution to $X(t)$. The uniqueness of $Y(t)$ is a technical matter which we do not consider, because in the phylogenetic case there are extra restrictions which led to the unique solution (\ref{reversed}).\\ 
Considering again the general case, we write
\beqn
\mathbb{P}(Y(t)=i|Y(s)=j):=n_{ij}(t,s)\nonumber
\eqn
and use (\ref{reversed}) to infer the general solution
\beqn\label{reverse}
n_{ij}(t,s)=\frac{m_{ji}(T-s,T-t)\pi_{i}(T-t)}{\pi_{j}(T-s)}.
\eqn
It is trivial to show that these transition probabilities satisfy the requirements of a CTMC:
\beqn
\sum_jn_{ji}(t,s)&=1,\qquad\forall\ i,\nonumber\\
N(t,s)N(s,u)&=N(t,u),
\eqn
where $N(t,s)=[n_{ij}(t,s)]_{(i,j\in\mathcal{K})}$.\\
Furthermore, by using (\ref{kalamorov}) we find that the rate parameters of the time-reversed chain can be expressed as
\beqn
f_{ij}(s):&=\frac{\partial n_{ij}(t,s)}{\partial t}|_{t=s}\\
&=\frac{q_{ji}(T-s)\pi_{i}(T-s)}{\pi_{j}(T-s)}-\sum_k\frac{\delta_{ij}q_{ik}(T-s)\pi_{k}(T-s)}{\pi_{j}(T-s)}\nonumber
\eqn
From which it follows that
\beqn
f_{ij}(s)\geq 0,\quad \forall i\neq j;\qquad f_{ii}(s)=-\sum_{j\neq i} f_{ji}(s)\nonumber
\eqn
which confirms that the $f_{ij}(s)$ are a valid set of rate parameters for a CTMC (as expected). It should be noted that even in the case where $X(t)$ is a homogeneous chain it is certainly not the case in general that $Y(t)$ is also homogeneous. Consider, however, the stationary and reversible case, with the respective conditions:
\beqn
\sum_jq_{ij}\pi_j(0)&=0,\\
q_{ij}\pi_j(0)&=q_{ji}\pi_i(0),\nonumber
\eqn
where the stationarity condition ensures that 
\beqn
\pi_i(t)=\pi_i(0),\qquad\forall t.\nonumber
\eqn
 In this circumstance it follows that
\beqn
f_{ij}=q_{ij},\nonumber
\eqn
such that $Y(t)\equiv X(t)$ and is hence also stationary and reversible. This was the basis of Felsenstein's initial formulation of the pulley principle -- if one considers only stationary and reversible Markov chains on a phylogenetic tree, any time-reversed chain is identical to the original Markov chain and hence a phylogenetic tree can be arbitrarily rerooted. We have given a continuous time generalization of Felsenstein's result which removes the stationary and reversible restriction.
\\
Equipped with the solution (\ref{reverse}) it is possible to take any phylogenetic tree and find an alternative tree of identical topology, but rooted in a different place, such that the alternative tree generates an identical probability distribution to that of the original. This is the basis of our generalized pulley principle. 
\\
The reader should note that we have proven, under the assumptions of the most general Markov model, that it is not possible to determine the orientation of a phylogenetic tree by only considering the joint probability distribution it generates at the leaves. Thus, any procedure that attempts to determine the root from the observed pattern frequencies must be justified by making additional assumptions about the underlying stochastic process. Chang \cite{chang1996} showed that the tree topology and (up to permutations of rows) the set of transition matrices, are reconstructible from the set of \textit{triples} of the joint distribution at the leaves. This is consistent with our result as Chang explicitly prohibited internal nodes with two incident edges and worked with unrooted/unorientated trees. Baake \cite{baake1998} showed that (up to similarity transformation) the \textit{return-trip} matrices (in our notation $M(s,t)N(t,s)$) are identifiable from the set of \textit{pairwise} joint distributions at the leaves. Again this is consistent with our result.
\\
The curious aspect of the generalized pulley principle is that the stochastic distance is \textit{not} conserved along the edge of the tree where the directedness was reversed. This is easy to show by considering the determinant of (\ref{reverse})
\beqn\label{consistency}
\det N(t,s)=\det M(T-s,T-t)\prod_i\frac{\pi_i(T-t)}{\pi_i(T-s)}
\eqn
Thus the stochastic distance in the reversed time chain is equal to that of the original chain if and only if
\beqn\label{iff}
\prod_i\frac{\pi_i(T-t)}{\pi_i(T-s)}=1.
\eqn
This property of CTMCs and their time-reversed counterparts was observed by Barry and Hartigan \cite{barry1987}. It can be seen that in the stationary case (\ref{iff}) will certainly be true. There are other cases where (\ref{iff}) may hold but there does not seem to any biologically sound way to interpret the required condition. In the proceeding discussion we will consider the consequences of the generalized pulley principle upon the interpretation of distance matrices. We see that for a given observed distribution we can use the generalized pulley principle to show that there are multiple edge length assignments using the stochastic distance which are consistent with the Markov model on a phylogenetic tree. These edge length assignments differ from one another as a consequence of (\ref{consistency}).

\subsection{Interpretation}

For illustrative purposes we consider the consequence to the stochastic distance of the rerooting of a phylogenetic tree of two taxa. We consider the phylogenetic trees illustrated in Figure \ref{pic:twotaxapulley}, and by using the generalized pulley principle define their respective transition matrices so that their probability distributions are identical: 
\beqn
n_{ij}&=\frac{m_{ji}\pi_{i}}{\rho_j},\nonumber\\
\rho_i&=\sum_jm_{ij}\pi_j.
\eqn
We find in the first case that we have
\beqn
\omega_\pi(1,2)=-\log\det M_1-\log\det M-\log\det M_2,\nonumber
\eqn
and in the second case
\beqn
\omega_\rho(1,2)=-\log\det M_1-\log\det N-\log\det M_2.\nonumber
\eqn
Now in general $\det M\neq \det N$ and we see that the two possible pairwise distances are not expected to be equal. However, from an empirical perspective it is impossible to distinguish these two possible theoretical scenarios because the probability distributions are identical. Now, because any estimator of the pairwise distance must be inferred from the observed distribution, we conclude that one must be careful to consider exactly what theoretical quantity one is obtaining an estimate of. For the case of the $\log\det$ formula we find that the quantity it is estimating depends essentially upon the base composition of the observed sequences as follows: 
\\
Considering the pairwise distance $d'_{ab}$ given by (\ref{logdetharm}), from the generalized pulley principle we see that this formula will give an estimate of the stochastic distance between $a$ and $b$, where the common ancestral node is placed such that the quantity
\beqn
\chi(a,b):=\prod_i\pi^{(a,b)}_i-\left[\prod_{i_1,i_2}\pi^{(a)}_{i_1}\pi^{(b)}_{i_2}\right]^{\frac{1}{2}}\nonumber
\eqn 
is minimized. Thus the $\log\det$ method will be inconsistent in the sense that, if there has been compositional heterogeneity, the pairwise distance it produces will be an estimate for the edge length assignment where $\chi(a,b)$ is minimized. This may have nothing to do with true placement of the common ancestral vertex and it may even be the case that $\chi(a,b)$ has multiple minimum points. The situation amounts to the fact that, for a given phylogenetic tree, one is (potentially) using the $\log\det$ to estimate pairwise distances with a different edge length assignment for each and every pair of taxa. Clearly for the analysis of multiple taxa this could be become a significant problem and any alternative approach which removes this inconsistency would be beneficial to the analysis. 
\\
We see that the consequences of the generalized pulley principle and (\ref{consistency}) to the interpretation of the Markov model of phylogenetics are quite subtle. The generalized pulley principle is telling us that there is no direct way to distinguish the rootedness (and equivalently the directedness of internal edges) of phylogenetic trees. This is due to the fact that there are (infinitely) many phylogenetic trees of identical topology which generate identical probability distributions, differing only by the assignment of stochastic distance and the associated redirection of internal edges.
\begin{figure}[t]
\pspicture[](-1.5,-.4)(3,3)
\psset{linewidth=\pstlw,xunit=0.5,yunit=0.5,runit=0.5}
\psset{arrowsize=2pt 2,arrowinset=0.2}
\psline{->}(4,6)(2.25,3.25)
\psline{->}(4,6)(4.875,4.675)
\psline{-}(4.875,4.675)(5.75,3.25)
\psline{-}(4,6)(0.50,0.50)
\psline{->}(5.75,3.25)(6.625,1.875)
\psline{-}(4,6)(7.50,0.50)
\pscircle[linewidth=0.8pt,fillstyle=solid,fillcolor=black](4,6){.15}
\pscircle[linewidth=0.8pt,fillstyle=solid,fillcolor=black](5.75,3.25){.15}
\rput(4,6.5){$\pi$}
\rput(1.25,3.25){$M_1$}
\rput(5.875,4.675){$M$}
\rput(6.3,3.25){$\rho$}
\rput(7.625,1.875){$M_2$}\rput(0.5,0){1}
\rput(7.5,0){2}
\rput(4,-1){\text{rooted at $\pi$}}
\rput(10.25,3){$\cong$}
\psline{->}(16,6)(14.25,3.25)
\psline{-<}(16,6)(16.875,4.675)
\psline{-}(16,6)(17.75,3.25)
\psline{-}(16,6)(12.50,0.50)
\psline{->}(17.75,3.25)(18.625,1.875)
\psline{-}(17.75,3.25)(19.50,0.50)
\pscircle[linewidth=0.8pt,fillstyle=solid,fillcolor=black](16,6){.15}
\pscircle[linewidth=0.8pt,fillstyle=solid,fillcolor=black](17.75,3.25){.15}
\rput(16,6.5){$\pi$}
\rput(13.25,3.25){$M_1$}
\rput(17.875,4.675){$N$}
\rput(18.3,3.25){$\rho$}
\rput(19.625,1.875){$M_2$}\rput(12.5,0){1}
\rput(19.5,0){2}
\rput(16,-1){\text{rooted at $\rho$}}
\endpspicture
\caption{Using the generalized pulley principle}
\label{pic:twotaxapulley}
\end{figure}

\section{The quartet case}

In this section we will show that in the case of a phylogenetic tree of four taxa, the tangle can be used to construct consistent quartet distance matrices. These distance matrices will be consistent in the sense that theoretically they are constructed from one topology with \textit{one} edge length assignment. This should be compared to the $\log\det$ formula which in the general case can be estimating a different edge length assignment for each and every pairwise distance. 
\\
For analytic purposes we use the generalized pulley principle to root the four taxon tree in two ways, as illustrated in Figure \ref{pic:fourtaxapulley}. The difference between the two cases is simply in the directedness of the internal edge and the generalized pulley principle allows us to calculate the required transition probabilities so that the two trees generate identical probability distributions. The pattern probabilities for the two cases are given by
\beqn\label{fourtaxapulley}
p_{i_1i_2i_3i_4}&=\sum_{j,k}m^{(1)}_{i_1j}m^{(2)}_{i_2j}m^{(3)}_{i_3k}m^{(4)}_{i_4k}m^{(5)}_{kj}\pi_j\\
&=\sum_{j,k}m^{(1)}_{i_1k}m^{(2)}_{i_2k}m^{(3)}_{i_3j}m^{(4)}_{i_4j}n^{(5)}_{kj}\rho_j
\eqn
where to ensure the equality of the two expressions we have
\beqn
n^{(5)}_{ij}=\frac{m^{(5)}_{ji}\pi_{i}}{\rho_j},\nonumber
\eqn
and $\rho_i=\sum_jm^{(5)}_{ij}\pi_j$.
\begin{figure}[t]
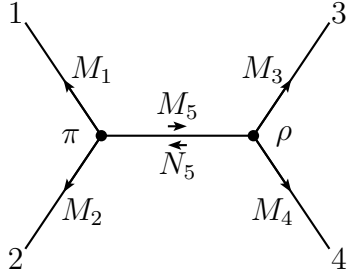

\label{fig:fourtaxa1}
\pspicture[](-5,0)(3,3)
\psset{linewidth=\pstlw,xunit=0.5,yunit=0.5,runit=0.5}
\psset{arrowsize=2pt 2,arrowinset=0.2}
\psline{-}(2,3)(0,0)
\rput(-.25,-.25){2}
\psline{-}(2,3)(0,6)
\rput(-.25,6.25){1}
\psline{-}(2,3)(6,3)
\psline{-}(6,3)(8,6)
\rput(8.25,6.25){3}
\psline{-}(6,3)(8,0)
\rput(8.25,-0.25){4}
\psline{->}(2,3)(1,4.5)
\rput(1.75,4.75){$M_1$}
\psline{->}(2,3)(1,1.5)
\rput(1.5,1.0){$M_2$}
\psline{->}(6,3)(7,4.5)
\rput(6.25,4.75){$M_3$}
\psline{->}(6,3)(7,1.5)
\rput(6.5,1.0){$M_4$}
\psline{->}(3.75,3.25)(4.25,3.25)
\rput(4,3.8){$M_5$}
\psline{<-}(3.75,2.75)(4.25,2.75)
\rput(4,2.2){$N_5$}
\pscircle[linewidth=0.8pt,fillstyle=solid,fillcolor=black](2,3){.15} 
\pscircle[linewidth=0.8pt,fillstyle=solid,fillcolor=black](6,3 ){.15} 
\rput(1.2,3){$\pi $}
\rput(6.8,3){$\rho$}
\endpspicture 
\caption{Four taxa tree with alternative roots}
\label{pic:fourtaxapulley}
\end{figure}
\\\noindent
From these expressions we wish to calculate the theoretical values of the formula (\ref{dist}) for each possible group of three taxa. To obtain these values one simply chooses the form of the tree such that after the deletion of a fourth taxon one is left with a three taxon tree of star topology. By sequentially deleting one taxon at a time we are led to the four star topology subtrees illustrated in Figure \ref{pic:subtrees} and the corresponding pattern probabilities are given by the expressions
\beqn
p^{(123)}_{ijk}&=\sum_{l_1,l_2}m^{(1)}_{il_1}m^{(2)}_{jl_1}m^{(3)}_{kl_2}m^{(5)}_{l_2l_1}\pi_{l_1},\\\nonumber
p^{(124)}_{ijk}&=\sum_{l_1,l_2}m^{(1)}_{il_1}m^{(2)}_{jl_1}m^{(4)}_{kl_2}m^{(5)}_{l_2l_1}\pi_{l_1},\\
p^{(134)}_{ijk}&=\sum_{l_1,l_2}m^{(1)}_{il_2}n^{(5)}_{l_2l_1}m^{(2)}_{jl_1}m^{(4)}_{kl_1}\rho_{l_1},\\
p^{(234)}_{ijk}&=\sum_{l_1,l_2}m^{(2)}_{il_2}n^{(5)}_{l_2l_1}m^{(3)}_{jl_1}m^{(4)}_{kl_1}\rho_{l_1}.\\
\eqn
From this it is easy to calculate the values simply by considering the results of the previous section:
\beqn\label{fourdelta}
\Delta_{12}^{(3)}&=\omega(1,2),\qquad 
&\Delta_{12}^{(4)}&=\omega(1,2),\\
\Delta_{13}^{(2)}&=\omega_\pi(1,3),\qquad 
&\Delta_{13}^{(4)}&=\omega_\rho(1,3),\\
\Delta_{14}^{(2)}&=\omega_\pi(1,4),\qquad 
&\Delta_{14}^{(3)}&=\omega_\rho(1,4),\\
\Delta_{23}^{(1)}&=\omega_\pi(2,3),\qquad 
&\Delta_{23}^{(4)}&=\omega_\rho(2,3),\\
\Delta_{24}^{(1)}&=\omega_\pi(2,4),\qquad 
&\Delta_{24}^{(3)}&=\omega_\rho(2,4),\\
\Delta_{34}^{(1)}&=\omega(3,4),\qquad 
&\Delta_{34}^{(2)}&=\omega(3,4),
\eqn
where
\beqn
\omega(a,b)&=\omega_a+\omega_b,\\
\omega_\pi(a,b)&=\omega_a+\omega_m+\omega_b,\\
\omega_\rho(a,b)&=\omega_a+\omega_n+\omega_b,\\
\omega_m&=-\log\det M,\\
\omega_n&=-\log\det N;
\nonumber
\eqn
and we have made use of (\ref{consistency}) in the form
\beqn
\omega_n=\omega_m-\sum_i(\log\pi_i-\log\rho_i).\nonumber
\eqn
\begin{figure}[t]
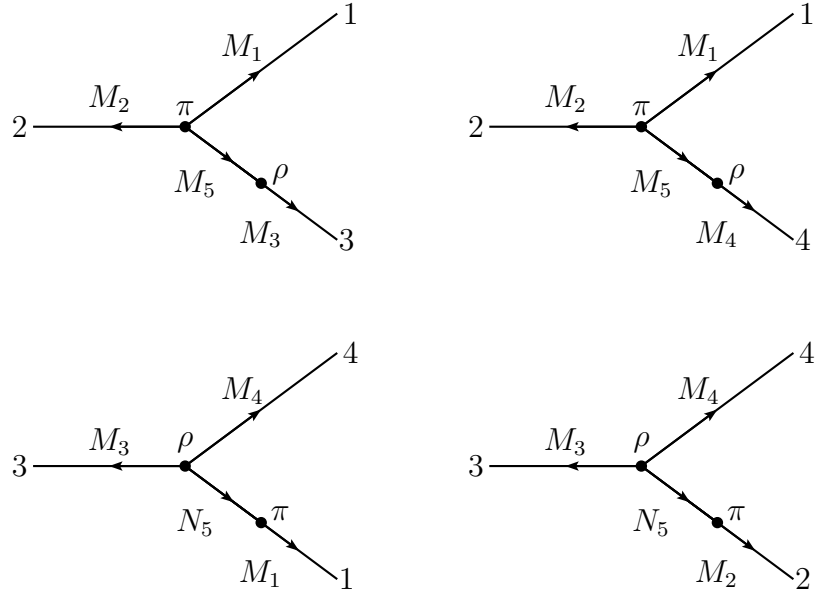

\pspicture[](-2,-5)(3,3)
\psset{linewidth=\pstlw,xunit=0.5,yunit=0.5,runit=0.5}
\psset{arrowsize=2pt 2,arrowinset=0.2}
\rput(4,3.5){$\pi$}
\psline{-}(4,3)(0,3)
\rput(-.35,3){2}
\psline{-}(4,3)(8,0)
\rput(8.25,0){3}
\psline{-}(4,3)(8,6)
\rput(8.35,6){1}
\psline{->}(4,3)(2,3)
\rput(2,3.75){$M_2$}
\psline{->}(4,3)(6,4.5)
\rput(5.5,5.1){$M_1$}
\psline{-}(4,3)(6,1.5)
\psline{>-}(5,2.25)(6,1.5)
\rput(4.25,1.5){$M_5$}
\rput(6.5,1.75){$\rho$}
\psline{->}(4,3)(7,.75)
\rput(5.969,0.15){$M_3$}
\pscircle[linewidth=0.8pt,fillstyle=solid,fillcolor=black](6,1.5){.15} 
\pscircle[linewidth=0.8pt,fillstyle=solid,fillcolor=black](4,3){.15} 
\rput(16,3.5){$\pi$}
\psline{-}(16,3)(12,3)
\rput(11.65,3){2}
\psline{-}(16,3)(20,0)
\rput(20.25,0){4}
\psline{-}(16,3)(20,6)
\rput(20.35,6){1}
\psline{->}(16,3)(14,3)
\rput(14,3.75){$M_2$}
\psline{->}(16,3)(18,4.5)
\rput(17.5,5.1){$M_1$}
\psline{-}(16,3)(18,1.5)
\psline{>-}(17,2.25)(18,1.5)
\rput(16.25,1.5){$M_5$}
\rput(18.5,1.75){$\rho$}
\psline{->}(16,3)(19,.75)
\rput(17.969,0.15){$M_4$}
\pscircle[linewidth=0.8pt,fillstyle=solid,fillcolor=black](18,1.5){.15} 
\pscircle[linewidth=0.8pt,fillstyle=solid,fillcolor=black](16,3){.15} 
\rput(4,-5.375){$\rho$}
\psline{-}(4,-6)(0,-6)
\rput(-.35,-6){3}
\psline{-}(4,-6)(8,-9)
\rput(8.25,-9){1}
\psline{-}(4,-6)(8,-3)
\rput(8.35,-3){4}
\psline{->}(4,-6)(2,-6)
\rput(2,-5.35){$M_3$}
\psline{->}(4,-6)(6,-4.5)
\rput(5.5,-4){$M_4$}
\psline{-}(4,-6)(6,-7.5)
\psline{>-}(5,-6.75)(6,-7.5)
\rput(4.25,-7.5){$N_5$}
\rput(6.5,-7.25){$\pi$}
\psline{->}(4,-6)(7,-8.25)
\rput(5.969,-8.85){$M_1$}
\pscircle[linewidth=0.8pt,fillstyle=solid,fillcolor=black](6,-7.5){.15} 
\pscircle[linewidth=0.8pt,fillstyle=solid,fillcolor=black](4,-6){.15}  
\rput(16,-5.375){$\rho$}
\psline{-}(16,-6)(12,-6)
\rput(11.65,-6){3}
\psline{-}(16,-6)(20,-9)
\rput(20.25,-9){2}
\psline{-}(16,-6)(20,-3)
\rput(20.35,-3){4}
\psline{->}(16,-6)(14,-6)
\rput(14,-5.35){$M_3$}
\psline{->}(16,-6)(18,-4.5)
\rput(17.5,-4){$M_4$}
\psline{-}(16,-6)(18,-7.5)
\psline{>-}(17,-6.75)(18,-7.5)
\rput(16.25,-7.5){$N_5$}
\rput(18.5,-7.25){$\pi$}
\psline{->}(16,-6)(19,-8.25)
\rput(17.969,-8.85){$M_2$}
\pscircle[linewidth=0.8pt,fillstyle=solid,fillcolor=black](18,-7.5){.15} 
\pscircle[linewidth=0.8pt,fillstyle=solid,fillcolor=black](16,-6){.15}  
\endpspicture
\caption{Three taxon subtrees}
\label{pic:subtrees} 
\end{figure}
\\\noindent
We see that for any two taxa we have two options for assigning a pairwise distance. In the cases of the pairs $(12)$ and $(34)$ we see that either choice is consistent with the other, whereas in the case of the pair $(13)$, $(14)$, $(24)$ and $(34)$ the two choices lead to an inconsistent assignment of the internal edge length upon the tree. Effectively what is happening here is that for a four taxa tree there are two possible edge length assignments for the internal edge and for a given pair of taxa $(ab)$ and third taxa $c$, the tangle formula (\ref{dist}) is estimating the distance between $a$ and $b$ by assigning one of the two possible edge lengths to the internal edge depending on the topology of the tree.
\\ 
It is possible to eliminate this inconsistency by using either a \textit{max} or \textit{min} criterion in the construction of the distance matrix:
\beqn
\phi^{max}_{ab}:=max\{\Delta_{ab}^{(c)},\Delta_{ab}^{(c')}\}\nonumber
\eqn
or
\beqn
\phi^{min}_{ab}:=min\{\Delta_{ab}^{(c)},\Delta_{ab}^{(c')}\}.\nonumber
\eqn
By making one of these choices to construct a distance matrix we choose the directedness of the internal edge of the phylogenetic tree (\ref{fig:fourtaxa1}) consistently. This procedure leads to an improvement of consistency upon the $\log\det$ technique for the construction of quartet phylogenetic distance matrices. It is hoped that this technique can be used fruitfully to improve the reconstruction of phylogenetic quartets, which can be used as a first step in the reconstruction of large phylogenetic trees \cite{bryant2001,strimmer1996}.
\\

\section{Closing remarks}
In this chapter we have given a review of the standard assignment of branch weights to phylogenetic trees, reviewed the use of the $\log\det$ formula as an estimator of pairwise distances and shown how a previously unknown polynomial, the tangle, can be used to construct an improved estimator. We have generalized Felsenstein's \textit{pulley principle} and used this result to show exactly how the distance matrix estimates become inconsistent when applied to the reconstruction problem of multiple taxa. We have shown that the tangle formula along with a \textit{max/min} criterion can be used to remove this inconsistency and construct consistent quartet distance matrices.

\chapter{Markov invariants}\label{chap5}
In this chapter we will refine the use of invariant theory on phylogenetic trees by defining \textit{Markov invariants} to be invariant functions specific to the general Markov model of sequence evolution. To achieve this we return to the representation theory introduced in Chapter \ref{chap2} and show how the Schur functions can be used to give a count of the existence of the Markov invariants. A procedure which constructs the explicit polynomial form of these invariants will be developed and we examine, as prompted from Chapter \ref{chap3}, the structure of these invariants once placed on a phylogenetic tree. For the triplet and quartet case we show that there exist Markov invariants which have the additional property of being \textit{phylogenetic invariants} \cite{allman2003,evans1993,steel1993}. These previously unobserved invariants can be used to achieve quartet reconstruction under the assumptions of the general Markov model.
\section{The Markov semigroup}
In Chapter \ref{chap3} we considered the transition matrices of a continuous time Markov chain as a subset of the general linear group, and used this property to study the structure of invariant polynomials (used as measures of entanglement in quantum physics) when evaluated on a phylogenetic tree. In this section we will close the gap between the general linear group and the subset consisting of the transition matrices of a CTMC by formally defining the \textit{Markov semigroup}. (For a detailed discussion of the Lie group properties of the Markov semigroup and its relation to the Affine group see \cite{johnson1985}.)\\
Recalling the vector $\theta=\sum e_{i}$ (\ref{def:theta}), the Markov semigroup on $n$ elements, $\mathcal{M}(n)$, with parameters $s\leq t< \infty$ is defined relative to $\theta$ as the subset of $GL(n)$ which satisfies: 
\begin{enumerate}\label{markovsemi}
\item $M(s,s)=1,$
\item $M(t',t)M(t,s)=M(t',s)\quad\forall s<t<t',$
\item $(\theta,M(t,s)v)=(\theta,v)\quad\forall\ v\in V.$
\end{enumerate}
In general this set does not form a group. Consider the time evolution of a probability vector $p(t)$, defined by
\beqn
p(t)=M(t,s)p(s),\quad s\leq t.\nonumber
\eqn
This time evolution will conserve the total probability
\beqn
\sum p_{i}(t)=(\theta,p(t))=(\theta,M(t,s)p(s))=(\theta,p(s))=\sum p_i(s)=1.\nonumber
\eqn
Defining
\beqn
Q(s):=\frac{\partial M(t,s)}{\partial t}|_{t=s},\nonumber
\eqn
it follows that in the $\{e_i\}$ basis, the matrix elements of $Q(t)=[q_{ij}(t)]$ satisfy
\beqn
q_{ij}(t)\geq 0,\quad \forall i\neq j;\qquad q_{ii}(t)=-\sum_{j\neq i}q_{ji}(t),\nonumber
\eqn
and hence each $M(s,t)$ is a valid transition matrix for a CTMC.\\
In Chapter \ref{chap3} we saw that the Markov model of phylogenetics can be considered in terms of the action $\times^mGL(n)$ on $V^{\otimes m}$ (\ref{convform}). We refine this to the action of $\times^m\mathcal{M}(n)$ on $V^{\otimes m}$ so that any phylogenetic tensor can be written as
\beqn
P=M_1\otimes M_2\otimes\ldots\otimes M_m\widetilde{P},\nonumber
\eqn
with $M_a\in\mathcal{M}(n)$, $1\leq a\leq m$. Our present task will be to define and derive invariant functions, $w:V^{\otimes m}\rightarrow \mathbb{C}$, which satisfy
\beqn
w(P)=\prod_{a=1}^m\det(M_{a})^kw(\widetilde{P}),\nonumber
\eqn 
for all $M_a\in\mathcal{M}(n)$, $1\leq a\leq m$, and analyse their relevance to the problem of phylogenetic tree reconstruction. (It should be noted that an invariant of the general linear group is certainly an invariant of the Markov semigroup, but the converse is not necessarily true.)
\subsection{Invariant functions of the Markov semigroup}
Before considering the more general case of the action $\times^m\mathcal{M}(n)$ on $V^{\otimes m}$ given by
\beqn
\psi\rightarrow M_1\otimes M_2\otimes\ldots\otimes M_m\psi,\nonumber
\eqn
we will first define invariant functions of the action $\mathcal{M}(n)$ on $V^{\otimes m}$ given by
\beqn
\psi\rightarrow \otimes^mM\psi.\nonumber
\eqn
Given that $\mathcal{M}(n)$ does not form a group we have to be careful in our definitions of representations and invariant functions. To this end we define the set of functions $\mathbb{C}[V^{\otimes m}]^{\mathcal{M}(n)}_d$ as the subset $w\in\mathbb{C}[V^{\otimes m}]_d$ which satisfy
\beqn\label{markinv}
w\circ \otimes^mM&=\det(M)^kw,\quad\forall M\in\mathcal{M}(n),\\
md&\!=\!kn,
\eqn
(where we have carefully not invoked the inverse element $M^{-1}$). Presently we will derive a sufficient condition for the existence of such invariant functions.\\
Consider $w\in\mathbb{C}[V]_d$ satisfying (\ref{markinv}). Under the canonical isomorphism $\omega:\mathbb{C}[V^{\otimes m}]_d\rightarrow (V^{\otimes m})^{\{d\}}$ (\ref{omegax}) we have
\beqn
w=\omega(\chi),\nonumber
\eqn
for some $\chi\in (V^{\otimes m})^{\{d\}}$. Carefully taking note of the relations (\ref{omegamap1}) and (\ref{omegamap2}) it follows that
\beqn
\omega(\chi)\circ M^{\otimes m}=\omega(\otimes^{md}M^t\chi).\nonumber
\eqn
Hence $w:=\omega(\chi)$ will satisfy (\ref{markinv}) if and only if
\beqn\label{markonerep}
\otimes^{md}M^t\chi=\det(M)^k\chi,\quad\forall M\in\mathcal{M}(n).
\eqn
Consider the tensor $\phi\in V^{\{k^n\}}\otimes V^{\otimes s}$ expressed as
\beqn
\phi=\eta\otimes (\theta^{\otimes s}),\nonumber 
\eqn
with $\eta\in V^{\{k^n\}}$. Recalling (\ref{onedrep}) and the definition of the Markov semigroup it follows that $\phi$ satisfies (\ref{markonerep}):
\beqn
\otimes^{kn+s}M^t\phi=\det(M)^k\phi,\quad\forall M\in\mathcal{M}(n).\nonumber
\eqn 
Consider the decomposition of $V^{\{k^n\}}\otimes V^{\otimes s}$ into irreducible representation spaces of $GL(n)$:
\beqn
V^{\{k^n\}}\otimes V^{\otimes s}=\sum_{|\lambda|=kn+s}h_\lambda V^\lambda,\nonumber
\eqn
for some unknown multiplicities $h_\lambda$. Our present task is to identify the irreducible representation space in which the tensor $\phi$ is contained. Assume $\phi\in V^{\mu}$ with $|\mu|=kn+s$ and recall that
\beqn
V^\mu=Y_{\mu}V^{\otimes kn+s},\nonumber
\eqn
where $Y_{\mu}$ is the projection operator satisfying
\beqn
Y_{\mu}^2&=Y_{\mu},\nonumber\\
Y_{\mu}Y_{\mu'}&=0,\quad |\mu'|=|\mu|,\\
\eqn
so that $Y_{\mu}$ is the unique Young operator satisfying
\beqn
Y_{\mu}\phi&=\phi.\nonumber
\eqn
Considering the inherent permutation symmetry of $\phi$, it is clear that 
\beqn
\mu=\{k+s,k^{n-1}\}.\nonumber
\eqn
From this we conclude that $\phi\in V^{\{k+s,k^{n-1}\}}$, and there exists $\chi\in (V^{\otimes m})^{\{d\}}$ satisfying $(\ref{markonerep})$ whenever
\beqn
(V^{\otimes m})^{\{d\}}\ni V^{\{k+s,k^{n-1}\}},\nonumber
\eqn
as an irreducible subspace under $GL(n)$.

\begin{prop}\label{thm:markouter}
A sufficient condition for the existence of a Markov invariant $w\in\mathbb{C}[V^{\otimes m}]^{\mathcal{M}(n)}_d$ is that $(\times^m\{1\})^{ \otimes{\{d\}}}\ni \{k+s,k^{n-1}\}$ for some $md\!=\!nk+s$.
\end{prop}\noindent
In direct analogy to the development of Theorem (\ref{thm:inner}) we generalize this to the action of $\times^m\mathcal{M}(n)$ on $V^{\otimes m}$:
\begin{prop}\label{thm:markinner}
A sufficient condition for the existence of a Markov invariant $w\in\mathbb{C}[V^{\otimes m}]^{\times^m\mathcal{M}(n)}_d$ is that $\ast^m\{k+s,k^{n-1}\}\ni \{d\}$ for some $d\!=\!nk+s$.
\end{prop}\noindent
Using the representation theoretical tools we have developed it does not seem trivial to show that these conditions are also necessary. However we now have at our disposal a tool for inferring the existence of Markov invariants in various cases.\\
In the next section we will return to the construction of invariants for the general linear group in order to derive a technique allowing us to compute these Markov invariants.
\section{Alternative computation of invariants of the general linear group}
The construction of invariants of the general linear group was presented in Chapter \ref{chap2} using the properties of the Levi-Civita tensor. Unfortunately this construction does not generalize to the case of the Markov semigroup. In this section we show how Young tableaux can be used to construct the invariant functions of $GL(n)$ directly. In the next section we show how this technique can be generalized to allow for the construction of the Markov invariants. 
\subsection{Action of $GL(n)$ on $V^{\otimes m}$}
Recall that the number of invariants of weight $k$ in $\mathbb{C}[V^{\otimes m}]^{GL(n)}_d$ is equal to the number of occurrences of the partition $\{k^n\}$ in $(\times^m\{1\})^{\otimes{\{d\}}}$ with ${kn\!=\!md}$. This gives us a technique for the proof of existence of invariant polynomials, but leaves us with the problem of their explicit construction. Recall Theorem \ref{thm:pleth} and we see that our task is to identify the one-dimensional representations of the general linear group in the decomposition of $(V^{\otimes m})^{\{d\}}$.\\
Suppose we consider $U=V^{\otimes m}$ as a ($nm$-dimensional) vector space with basis $u_1,u_2,\ldots,u_{nm}$ . As we saw in Chapter \ref{chap2}, if $U$ has a basis $u_1,u_2,\ldots,u_{nm}$ then any $\chi\in U^{\{d\}}$ can be constructed from an arbitrary $\chi\in U^{\otimes d}$ by taking
\beqn
\varphi=Y_{\{d\}}\chi,\nonumber
\eqn
where the Young operator acts on the $\{u_{\alpha_1}\otimes u_{\alpha_2}\otimes\ldots\otimes u_{\alpha_d}\}$ basis of $U^{\otimes d}$, $1\leq\alpha_1,\ldots,\alpha_d\leq nm$. Now we define
\beqn
\phi=Y_{\{k^n\}}\varphi,\nonumber
\eqn
where the Young operator now acts on the $\{e_{i_1}\otimes e_{i_2}\otimes\ldots\otimes e_{i_{dm}}\}$ basis of $(V^{\otimes m})^{\{d\}}\cong U^{\{d\}}$, $1\leq i_1,\ldots,i_{dm}\leq n$. The final step is to construct the single independent component of $\phi$ using the semi-standard tableau:
\begin{figure}[h]
\centering
\pspicture[](0,0)(2.5,2.5)
\psline(0,0)(2.5,0)
\psline(0,.5)(2.5,.5)
\psline(0,1.5)(2.5,1.5)
\psline(0,2)(2.5,2)
\psline(0,2.5)(2.5,2.5)

\psline(0,0)(0,2.5)
\psline(.5,0)(.5,2.5)
\psline(1,0)(1,2.5)
\psline(2,0)(2,2.5)
\psline(2.5,0)(2.5,2.5)

\rput(0.25,.25){n}
\rput(0.75,.25){n}
\rput(2.25,.25){n}
\rput(0.25,1.75){2}
\rput(0.75,1.75){2}
\rput(2.25,1.75){2}
\rput(0.25,2.25){1}
\rput(0.75,2.25){1}
\rput(2.25,2.25){1}

\rput(1.5,2.2){...}
\rput(.2,1){.}
\rput(.2,1.1){.}
\rput(.2,.9){.}
\endpspicture 
\end{figure}
\\
and then map over the invariant ring using $\omega:(V^{\otimes m})^{\{d\}}\rightarrow \mathbb{C}[V^{\otimes m}]_d$. The invariant is then
\beqn
f:=\omega(\phi)=\omega(Y_{\{k_n\}}\varphi)=\omega(Y_{\{k^n\}} Y_{\{d\}}\chi),\nonumber
\eqn
which will satisfy
\beqn
f\circ g=\det(g)^kf,\nonumber
\eqn
for all $g\in GL(n)$.
\\
There is no problem with choosing the operator $Y_{\{d\}}$ as there is only one possible standard tableau:
\begin{figure}[h]
\centering
\pspicture[](0,0)(5,.5)
\psline(0,0)(2.5,0)
\psline(0,.5)(2.5,.5)
\psline(0,0)(0,.5)
\psline(.5,0)(.5,.5)
\psline(1,0)(1,.5)
\psline(2,0)(2,.5)
\psline(2.5,0)(2.5,.5)
\rput(0.25,.25){1}
\rput(0.75,.25){2}
\rput(1.5,.2){...}
\rput(2.25,.25){d}
\rput(3,0){.}
\endpspicture
\end{figure}\\
However there does not seem to be any \textit{a priori} way of deciding which standard tableau to use for the symmetrization $Y_{\{k^n\}}$. In general there are more standard tableaux than one-dimensional representations. This is not a serious issue since the Young symmetrization procedure needs to be implemented in an algebraic computation computer package. Our procedure was to make judicious choices of standard tableaux and check for algebraic independence of the resulting invariants until the correct count was achieved. In what follows we will present the results of these computations.\\
The above outlines the formal procedure. In practice we implement the algorithm as follows. The above is equivalent to computing
\beqn\label{ticktick}
\Psi_{i_1\ldots i_{md}}:=Y_{\{k^n\}}\psi_{i_1\ldots i_m}\psi_{i_{m+1}\ldots i_{2m}}\ldots \psi_{\ldots i_{md}},
\eqn
where
\beqn
(ab)\psi_{i_1\ldots i_{m}}\ldots\psi_{\ldots i_a\ldots }\ldots \psi_{\ldots i_b\ldots }\ldots \psi_{\ldots i_{md}}:=\psi_{i_1\ldots i_{m}}\ldots \psi_{\ldots i_b\ldots }\ldots \psi_{\ldots i_a\ldots }\ldots \psi_{\ldots i_{md}},\nonumber
\eqn
for any $1\leq a,b\leq md$, defines the meaning of (\ref{ticktick}) and there is no need to symmetrize with $Y_{\{d\}}$. (In practice the symmetries inherent in this procedure give us some clue as to how to choose the appropriate standard tableaux for $\{k^n\}$.) We then set the indices of $\Psi_{i_1\ldots i_{md}}$ using the single semi-standard tableaux to get
\beqn\label{invalgo}
w(\psi)=\Psi_{12\ldots n12\ldots n\ldots 12\ldots n}.
\eqn
Now this expression only depends on the choice of standard tableau for $\{k^n\}$. In practice we compute (\ref{invalgo}) for different standard tableaux until we have the correct number of independent invariants.
\subsection{Examples}
We consider the case $m\!=\!2$. We have $(\{1\}\times\{1\})^{\otimes \{1\}}=\{2\}+\{1^2\}\ni\{1^2\}$, and hence there is one invariant of degree $d\!=1$. Of course this invariant can simply be found by symmetrizing $V^{\otimes 2}$ with the only standard tableau of shape $\{1^2\}$:
\begin{figure}[h]
\centering
\pspicture[](0,0)(.5,1)
\psline(0,0)(.5,0)
\psline(0,.5)(.5,.5)
\psline(0,1)(.5,1)
\psline(0,0)(0,1)
\psline(.5,0)(.5,1)
\rput(.25,.25){2}
\rput(.25,.73){1}
\endpspicture
\end{figure}\\
with corresponding Young operator
\beqn
Y_{\{1^2\}}=(e-(12)).\nonumber
\eqn
The symmetrized tensor is 
\beqn
\Psi_{i_1i_2}=Y_{\{1^2\}}\psi_{i_1i_2}=\psi_{i_1i_2}-\psi_{i_2i_1}.\nonumber
\eqn
The invariant is found by inserting index labels from the relevant\\ semi-standard tableau, so that
\beqn
w(\psi)=\Psi_{12}=\psi_{12}-\psi_{21}.\nonumber
\eqn
For $d\!=2$ the output of \textbf{Schur} shows that $(\{1\}\times \{1\})^{\otimes \{2\}}\ni2\{2^2\}$.\newpage\noindent There are two Young operators with shape $\{2^2\}$:
\begin{figure}[h]
\centering
\pspicture[](0,.25)(4,1)
\psline(0,0)(1,0)
\psline(0,.5)(1,.5)
\psline(0,1)(1,1)
\psline(0,0)(0,1)
\psline(.5,0)(.5,1)
\psline(1,0)(1,1)
\rput(.25,.25){3}
\rput(.25,.75){1}
\rput(.75,.25){4}
\rput(.75,.75){2}

\psline(1.5,0)(2.5,0)
\psline(1.5,.5)(2.5,.5)
\psline(1.5,1)(2.5,1)
\psline(1.5,0)(1.5,1)
\psline(2,0)(2,1)
\psline(2.5,0)(2.5,1)
\rput(1.75,.25){3}
\rput(1.75,.75){1}
\rput(2.25,.25){4}
\rput(2.25,.75){2}
\endpspicture .
\end{figure}\\
The invariants are then given by
\beqn
\Psi_{i_1i_2i_3i_4}=Y_{\{2^2\}}\psi_{i_1i_2}\psi_{i_3i_4}.\nonumber
\eqn
For the first tableau we have
\beqn
Y_{\{2^2\}}=(e-(13)-(24)+(13)(24))(e+(12)+(34)+(12)(34)),\nonumber
\eqn
and find explicitly for the semi-standard tableau corresponding to component $\Psi_{1212}$:
\beqn
h_1(\psi)=\psi_{12}^2+2\psi_{12}\psi_{21}+\psi_{21}^2-4\psi_{11}\psi_{22},\nonumber
\eqn
and for the second tableau
\beqn
h_{2}(\psi)=\psi_{12}^2-\psi_{12}\psi_{21}+\psi_{21}^2-\psi_{11}\psi_{22}.\nonumber
\eqn
It is a simple exercise to show that these invariants are linear combinations of the two invariants produced in Chapter \ref{chap2} (\ref{eq:inv1}):
\beqn
h_1&=f_1^2-4f_2,\nonumber\\
h_2&=f_1^2-f_2.
\eqn
For the case of $GL(3)$ on $V^{\otimes 2}$ \textbf{Schur} shows that $(\{1\}\times \{1\})^{\otimes \{3\}}\ni 2\{2^3\}$. The invariants are constructed from arbitrary $\psi\in (V^{\otimes 2})^{\otimes 3}$ as
\beqn
f=\omega(Y_{\{2^3\}}\circ Y_{\{3\}}\psi),\nonumber
\eqn
with the standard tableaux
\begin{figure}[h]
\centering
\pspicture[](0,0)(2.5,1.5)
\psline(0,0)(1,0)
\psline(1.5,0)(2.5,0)
\psline(0,.5)(1,.5)
\psline(1.5,.5)(2.5,.5)
\psline(0,1)(1,1)
\psline(1.5,1)(2.5,1)
\psline(0,1.5)(1,1.5)
\psline(1.5,1.5)(2.5,1.5)

\psline(0,0)(0,1.5)
\psline(1.5,0)(1.5,1.5)
\psline(0.5,0)(0.5,1.5)
\psline(2,0)(2,1.5)
\psline(1,0)(1,1.5)
\psline(2.5,0)(2.5,1.5)

\rput(.25,.25){5}
\rput(1.75,.25){5}
\rput(.25,.75){3}
\rput(1.75,.75){2}
\rput(.25,1.25){1}
\rput(1.75,1.25){1}

\rput(.75,.25){6}
\rput(2.25,.25){6}
\rput(.75,.75){4}
\rput(2.25,.75){4}
\rput(.75,1.25){2}
\rput(2.25,1.25){3}

\endpspicture
\end{figure}\\
generating two independent elements:
\beqn
h_{1}(\psi)&=-\psi_{1 3}^2 \psi_{2 2} + \psi_{1 2} \psi_{1 3} \psi_{2 3} + 
    \psi_{1 3} \psi_{2 1} \psi_{2 3} - \psi_{1 1} \psi_{2 3}^2 - 
    2 \psi_{1 3} \psi_{2 2} \psi_{3 1}\\& + \psi_{1 2} \psi_{2 3} \psi_{3 1} + 
    \psi_{2 1} \psi_{2 3} \psi_{3 1} - \psi_{2 2} \psi_{3 1}^2 + 
    \psi_{1 2} \psi_{1 3} \psi_{3 2}\\& + \psi_{1 3} \psi_{2 1} \psi_{3 2} - 
    2 \psi_{1 1} \psi_{2 3} \psi_{3 2} + \psi_{1 2} \psi_{3 1} \psi_{3 2} + 
    \psi_{2 1} \psi_{3 1} \psi_{3 2} - \psi_{1 1} \psi_{3 2}^2\\& - \psi_{1 2}^2 \psi_{3 3} - 
    2 \psi_{1 2} \psi_{2 1} \psi_{3 3} - \psi_{2 1}^2 \psi_{3 3} + 
    4 \psi_{1 1} \psi_{2 2} \psi_{3 3},\nonumber
\eqn
and
\beqn
h_2(\psi)&=\psi_{1 3}^2 \psi_{2 2} - \psi_{1 2} \psi_{1 3} \psi_{2 3} - 
    \psi_{1 3} \psi_{2 1} \psi_{2 3} + \psi_{1 1} \psi_{2 3}^2 + 
    \psi_{1 2} \psi_{2 3} \psi_{3 1}\\& - \psi_{2 1} \psi_{2 3} \psi_{3 1} + 
    \psi_{2 2} \psi_{3 1}^2 - \psi_{1 2} \psi_{1 3} \psi_{3 2} + 
    \psi_{1 3} \psi_{2 1} \psi_{3 2} - \psi_{1 2} \psi_{3 1} \psi_{3 2}\\& - 
    \psi_{2 1} \psi_{3 1} \psi_{3 2} + \psi_{1 1} \psi_{3 2}^2 + \psi_{1 2}^2 \psi_{3 3} + 
    \psi_{2 1}^2 \psi_{3 3} - 2 \psi_{1 1} \psi_{2 2} \psi_{3 3}.\nonumber
\eqn
Again it is possible to show that these invariants are linear combinations of the corresponding invariants produced in Chapter \ref{chap2} (\ref{eq:inv2}).
\subsection{Action of $\times^m GL(n)$ on $V^{\otimes m}$}
Recalling Theorem \ref{thm:inner}, we note that the number of weight $k$ invariants in $\mathbb{C}[V^{\otimes m}]^{\times^m GL(n)}_d$ is equal to the number of occurrences of $\{d\}$ in the decomposition of $\ast^m\{k^n\}$. For even $m$ we have the identity 
\beqn
\ast^m \{1^n\}=\{n\},\nonumber
\eqn
and for odd $m$
\beqn
\ast^m \{1^n\}=\{1^n\}.\nonumber
\eqn
Thus we see that for even $m$ there is a single invariant function of degree $d\!=\!n$ and for odd $m$ there are none. For even $m$ the invariant is generated from
\beqn
\Psi_{i_1\ldots i_{nm}}=Y^{(1)}_{\{1^n\}} Y^{(2)}_{\{1^n\}}\ldots Y^{(m)}_{\{1^n\}}\psi_{i_1\ldots i_m}\psi_{i_{m+1}\ldots i_{2m}}\ldots\psi_{i_{(n-1)m}..i_{nm}}\nonumber
\eqn
where each standard tableau $Y^{(a)},1\leq a\leq m$, is
\begin{figure}[h]
\centering
\pspicture[](0,0)(2,3)
\psline(0,0)(0,2.5)
\psline(3,0)(3,2.5)
\psline(0,0)(3,0)
\psline(0,.5)(3,.5)
\psline(0,1.5)(3,1.5)
\psline(0,2)(3,2)
\psline(0,2.5)(3,2.5)
\rput(1.5,.25){$(n-1)m+a$}
\rput(1.5,1.75){$m+a$}
\rput(1.5,2.25){$a$}
\endpspicture
\hspace{15mm}. 
\end{figure}\\
We then set the indices of $\Psi_{i_1\ldots i_{nm}}$ using the single semi-standard tableau for each Young operator to obtain
\beqn
w(\psi)=\Psi_{11\ldots 122\ldots 2\ldots nn\ldots n}.\nonumber
\eqn
It should be clear that this procedure is completely equivalent to the invariants obtained using the Levi-Civita tensor Chapter \ref{chap2} (\ref{eq:quangles}). In the case $n=2$, this procedure generates the determinant invariants (\ref{concurrence}) and for $n=4$ the quangles (\ref{eq:quangles}).\\  
However, as we will now see, we need to use the tableaux technique in order to do the same job for the Markov semigroup.
\section{Computation of the Markov invariants}
Here we will generalize the above technique for computing invariants of the general linear group to the case of the Markov semigroup. It should be noted that in the case of the general linear group, the basis in which the calculations are performed is of no consequence as the invariants take on the identical form (up to scaling) in \textit{any} basis. (This is by definition!)  However, in the case of the Markov invariants all calculations with Young operators must be performed in the basis $\{z_0,z_a\}$, see Chapter \ref{chap2} (\ref{zbasis}). This is due to the very definition of the Markov semigroup which depends on a particular choice of the vector $\theta=\sqrt{n}z_0$. Thus, in the subsequent discussion, it should be remembered that all Markov invariants are presented in the form they take in the $\{z_0,z_a\}$ basis.
\subsection{Markov invariants of $\mathcal{M}(n)$ on $V^{\otimes m}$}
In this section we consider the action of $\mathcal{M}(n)$ on $V^{\otimes m}$ given by
\beqn
\psi\rightarrow \otimes^m\psi.\nonumber
\eqn
Recalling Conjecture \ref{thm:markouter}, it follows that if 
\beqn
(\times^m\{1\})^{\otimes \{d\}}\ni \{k+s,k^{n-1}\}\nonumber
\eqn
for some $md\!=nk+s$ there exists a Markov invariant $w\in\mathbb{C}[V^{\otimes m}]^{\mathcal{M}(n)}_d$. (In all that follows it should be noted that the case $s\!=\!0$ reproduces an invariant of the general linear group.) Computing
\beqn
\Psi_{i_1\ldots i_{dm}}:=Y_{\{k+s,k^{n-1}\}}\psi_{i_1\ldots i_m}\psi_{i_{m+1}\ldots i_{2m}}\ldots\psi_{\ldots i_{md}}\nonumber
\eqn 
where the standard tableau of shape ${\{k+s,k^{n-1}\}}$ used to define $Y_{\{k+s,k^{n-1}\}}$ is not fixed, but is chosen judiciously. The final step is to compute $w(\psi)$ by inserting indices into $\Psi$ using the semi-standard tableau:
\begin{figure}[h]
\centering
\pspicture[](0,0)(4,2.5)
\psline(-.25,0)(2.5,0)
\psline(-.25,.5)(2.5,.5)
\psline(-.25,1.5)(2.5,1.5)
\psline(-.25,2)(5,2)
\psline(-.25,2.5)(5,2.5)

\psline(5,2.5)(5,2)
\psline(4.5,2.5)(4.5,2)
\psline(3.5,2.5)(3.5,2)
\psline(3,2.5)(3,2)

\psline(-.25,0)(-.25,2.5)
\psline(.5,0)(.5,2.5)
\psline(1.25,0)(1.25,2.5)
\psline(1.75,0)(1.75,2.5)
\psline(2.5,0)(2.5,2.5)

\rput(0.1275,.25){n-1}
\rput(0.8775,.25){n-1}
\rput(2.1275,.25){n-1}
\rput(0.1275,1.75){1}
\rput(0.8775,1.75){1}
\rput(2.1275,1.75){1}
\rput(0.1275,2.25){0}
\rput(0.8775,2.25){0}
\rput(2.1275,2.25){0}
\rput(2.75,2.25){0}
\rput(3.25,2.25){0}
\rput(4.75,2.25){0}

\rput(1.5,2.2){...}
\rput(.2,1){.}
\rput(.2,1.1){.}
\rput(.2,.9){.}
\endpspicture .
\end{figure}
\subsection{Examples}
We will consider Markov invariants of degree $d\!=\!1$ only. For the case of $n\!=\!2,m\!=\!3$, \textbf{Schur} shows that 
\beqn
(\times^3\{1\})^{\otimes \{1\}}=\times^3\{1\}\ni 2\{21\},\nonumber
\eqn
which implies that there are two Markov invariants corresponding to $\{21\}$ with $k\!=\!s\!=\!1$.\\
There are two standard tableaux of shape $\{2,1\}$:
\begin{figure}[h]
\centering
\pspicture[](0,0)(2,1)
\psline(-.5,0)(-.5,1)
\psline(-.5,1)(.5,1)
\psline(0,0)(0,1)
\psline(-.5,0)(0,0)
\psline(.5,.5)(.5,1)
\psline(-.5,.5)(.5,.5)
\rput(-.25,.75){1}
\rput(.25,.75){2}
\rput(-.25,.25){3}
\endpspicture 
\pspicture[](-1.5,0)(0,1)
\psline(-.5,0)(-.5,1)
\psline(-.5,1)(.5,1)
\psline(0,0)(0,1)
\psline(-.5,0)(0,0)
\psline(.5,.5)(.5,1)
\psline(-.5,.5)(.5,.5)
\rput(-.25,.75){1}
\rput(.25,.75){3}
\rput(-.25,.25){2}
\rput(1,0){.} 
\endpspicture 
\end{figure}\\
The corresponding $d\!=\!1$ Markov invariant follows from computing
\beqn
\Psi_{i_1i_2i_3}:=Y_{\{21\}}\psi_{i_1i_2i_3}\nonumber
\eqn
\newpage\noindent
and then inserting indices according to the single semi-standard tableau:
\begin{figure}[h]
\centering
\pspicture[](0,0)(2,1)
\psline(-.5,0)(-.5,1)
\psline(-.5,1)(.5,1)
\psline(0,0)(0,1)
\psline(-.5,0)(0,0)
\psline(.5,.5)(.5,1)
\psline(-.5,.5)(.5,.5)
\rput(-.25,.75){0}
\rput(.25,.75){0}
\rput(-.25,.25){1}
\endpspicture .
\end{figure}\\
For the first tableau we compute the symmetrized tensor
\beqn
\Psi_{a_1a_2a_3}^{(1)}=\psi_{a_1a_2a_3}+\psi_{a_2a_1a_3}-\psi_{a_3a_2a_1}-\psi_{a_2a_3a_1}.\nonumber
\eqn
The single independent component gives the Markov invariant
\beqn
\Psi_{001}^{(1)}=2\psi_{001}-\psi_{100}-\psi_{010}.\nonumber
\eqn
The second tableau gives the symmetrized tensor
\beqn
\Psi_{a_1a_2a_3}^{(2)}=\psi_{a_1a_2a_3}+\psi_{a_3a_2a_1}-\psi_{a_2a_1a_3}-\psi_{a_3a_1a_2}.\nonumber
\eqn
The single independent component gives the second Markov invariant
\beqn
\Psi_{010}^{(2)}=2\psi_{010}-\psi_{100}-\psi_{001}.\nonumber
\eqn
As a second example, consider the case $n\!=\!2,m\!=\!4$, with \textbf{Schur} giving
\beqn
(\times^4\{1\})^{\otimes \{1\}}=\times^4{\{1\}}\ni 3\{31\},\nonumber
\eqn
so there are three Markov invariants with $k\!=\!s\!=\!1$. There are three standard tableaux and hence three candidate Young operators:
\begin{figure}[h]
\pspicture[](-2,0)(.5,1)
\psline(-.5,0)(-.5,1)
\psline(-.5,1)(1,1)
\psline(1,1)(1,.5)
\psline(0,0)(0,1)
\psline(-.5,0)(0,0)
\psline(.5,.5)(.5,1)
\psline(-.5,.5)(1,.5)
\rput(-.25,.75){1}
\rput(.25,.75){2}
\rput(.75,.75){3}
\rput(-.25,.25){4}
\endpspicture
\pspicture[](-2,0)(.5,1)
\psline(-.5,0)(-.5,1)
\psline(-.5,1)(1,1)
\psline(1,1)(1,.5)
\psline(0,0)(0,1)
\psline(-.5,0)(0,0)
\psline(.5,.5)(.5,1)
\psline(-.5,.5)(1,.5)
\rput(-.25,.75){1}
\rput(.25,.75){2}
\rput(.75,.75){4}
\rput(-.25,.25){3}
\endpspicture
\pspicture[](-2,0)(.5,1)
\psline(-.5,0)(-.5,1)
\psline(-.5,1)(1,1)
\psline(1,1)(1,.5)
\psline(0,0)(0,1)
\psline(-.5,0)(0,0)
\psline(.5,.5)(.5,1)
\psline(-.5,.5)(1,.5)
\rput(-.25,.75){1}
\rput(.25,.75){3}
\rput(.75,.75){4}
\rput(-.25,.25){2}
\rput(1.5,0){.}
\endpspicture 
\end{figure}\\
The associated semi-standard tableau is
\begin{figure}[h]
\centering
\pspicture[](0,0)(.5,1)
\psline(-.5,0)(-.5,1)
\psline(-.5,1)(1,1)
\psline(1,1)(1,.5)
\psline(0,0)(0,1)
\psline(-.5,0)(0,0)
\psline(.5,.5)(.5,1)
\psline(-.5,.5)(1,.5)
\rput(-.25,.75){0}
\rput(.25,.75){0}
\rput(.75,.75){0}
\rput(-.25,.25){1}
\rput(1,0){.}
\endpspicture
\end{figure}\\
For the first tableau we have the symmetrized tensor:
\beqn
\Psi_{a_1a_2a_3a_4}&=\psi_{a_1a_2a_3a_4}+\psi_{a_2a_1a_3a_4}+\psi_{a_3a_2a_1a_4}+\psi_{a_1a_3a_2a_4}\\&+\psi_{a_3a_1a_2a_4}+\psi_{a_2a_3a_1a_4}-\psi_{a_4a_2a_3a_1}-\psi_{a_2a_4a_3a_1}\\&-\psi_{a_3a_2a_4a_1}-\psi_{a_4a_3a_2a_1}-\psi_{a_3a_4a_2a_1}-\psi_{a_2a_3a_4a_1}.\nonumber
\eqn
By inserting the indices we get the Markov invariant
\beqn
6\psi_{0001}-2\psi_{1000}-2\psi_{0100}-2\psi_{0010}.\nonumber
\eqn
And by analogy for the remaining two Young operators (with the same semi-standard tableau) we have the Markov invariants
\beqn
6\psi_{0010}-2\psi_{1000}-2\psi_{0100}-2\psi_{0001},\nonumber
\eqn
and
\beqn
6\psi_{0100}-2\psi_{1000}-2\psi_{0010}-2\psi_{0001}.\nonumber
\eqn
Our final example is the case $n\!=\!3,m\!=\!4$, with \textbf{Schur} giving
\beqn
(\times^4\{1\})^{\otimes 1}=\times^4\{1\}\ni 3\{21^2\},\nonumber
\eqn
so there are two Markov invariants with $k\!=\!s\!=\!1$. Again, there are three standard tableaux
\begin{figure}[h]
\pspicture[](-2,0)(.5,2)
\psline(-.5,0)(-.5,1.5)
\psline(0,0)(0,1.5)
\psline(.5,1)(.5,1.5)
\psline(-.5,0)(0,0)
\psline(-.5,.5)(0,.5)
\psline(-.5,1)(.5,1)
\psline(-.5,1.5)(.5,1.5)
\rput(-.25,1.25){1}
\rput(.25,1.25){2}
\rput(-.25,.75){3}
\rput(-.25,.25){4}
\endpspicture
\pspicture[](-2,0)(.5,2)
\psline(-.5,0)(-.5,1.5)
\psline(0,0)(0,1.5)
\psline(.5,1)(.5,1.5)
\psline(-.5,0)(0,0)
\psline(-.5,.5)(0,.5)
\psline(-.5,1)(.5,1)
\psline(-.5,1.5)(.5,1.5)
\rput(-.25,1.25){1}
\rput(.25,1.25){3}
\rput(-.25,.75){2}
\rput(-.25,.25){4}
\endpspicture
\pspicture[](-2,0)(.5,2)
\psline(-.5,0)(-.5,1.5)
\psline(0,0)(0,1.5)
\psline(.5,1)(.5,1.5)
\psline(-.5,0)(0,0)
\psline(-.5,.5)(0,.5)
\psline(-.5,1)(.5,1)
\psline(-.5,1.5)(.5,1.5)
\rput(-.25,1.25){1}
\rput(.25,1.25){4}
\rput(-.25,.75){2}
\rput(-.25,.25){3}
\endpspicture
\end{figure}\\
with associated semi-standard tableau
\begin{figure}[h]
\pspicture[](-2,0)(.5,2)
\psline(-.5,0)(-.5,1.5)
\psline(0,0)(0,1.5)
\psline(.5,1)(.5,1.5)
\psline(-.5,0)(0,0)
\psline(-.5,.5)(0,.5)
\psline(-.5,1)(.5,1)
\psline(-.5,1.5)(.5,1.5)
\rput(-.25,1.25){0}
\rput(.25,1.25){0}
\rput(-.25,.75){1}
\rput(-.25,.25){2}
\endpspicture .
\end{figure}\\
From the first standard tableau we compute the symmetrized tensor:
\beqn
\Psi_{a_1a_2a_3a_4}^{(1)}=&\psi_{a_1a_2a_3a_4}+\psi_{a_2a_1a_3a_4}-\psi_{a_3a_2a_1a_4}-\psi_{a_2a_3a_1a_4}-\psi_{a_4a_2a_3a_1}\\&-\psi_{a_2a_4a_3a_1}-\psi_{a_1a_2a_4a_3}-\psi_{a_2a_1a_4a_3}+\psi_{a_4a_2a_1a_3}\\&+\psi_{a_2a_4a_1a_3}+\psi_{a_3a_2a_4a_1}+\psi_{a_2a_3a_4a_1}.\nonumber
\eqn
Again by filling the indices according to the semi-standard tableau we get the Markov invariant
\beqn
\Psi_{0012}^{(1)}=2\psi_{0012}&-\psi_{1002}-\psi_{0102}-\psi_{2010}-\psi_{0210}-2\psi_{0021}\\&+\psi_{2001}+\psi_{0201}+\psi_{1020}+\psi_{0120}.\nonumber
\eqn
Similarly we find for the remaining two standard tableaux:
\beqn
\Psi_{0102}^{(2)}=-\psi_{0 0 1 2}& + \psi_{0 0 2 1} + 2 \psi_{0 1 0 2} - \psi_{0 
    1 2 0} - 2 \psi_{0 2 0 1} \\&+ \psi_{0 2 1 0} -
   \psi_{1 0 0 2} + \psi_{1 2 0 0} + \psi_{2 0 0 1} - \psi_{2 1 0 0}\nonumber
\eqn
and
\beqn
\Psi_{0120}^{(3)}=\psi_{0 0 1 2}& - \psi_{0 0 2 1} - 
  \psi_{0 1 0 2} + 2 \psi_{0 1 2 0} + \psi_{0 2 0 1} - 2 \psi_{0 2 1
   0}\\& - \psi_{1 0 2 0} + \psi_{1 2 0 0} + \psi_{2 0 1 0} - \psi_{2 1 0 0}.\nonumber
\eqn
These three invariants are linearly independent, as required.
\section{Markov invariants of $\times^m\mathcal{M}(n)$ on $V^{\otimes m}$}
We now consider invariants of the group action $\times^m\mathcal{M}(n)$ on $V^{\times m}$ given by
\beqn
\psi\rightarrow M_1\otimes M_2\otimes\ldots\otimes M_m\psi;\quad M_a\in\mathcal{M}(n),\ 1\leq a\leq m.\nonumber
\eqn
According to Conjecture \ref{thm:markinner} there exists a Markov invariant, $w$, of degree $d$ of this group action if
\beqn
\ast^m\{k+s,k^{n-1}\}\ni \{d\},\nonumber
\eqn
for some $nk+s\!=\!d$. These Markov invariants will satisfy
\beqn
w(M_1\otimes M_2\otimes\ldots\otimes M_m\psi)=(\det(M_1)\det(M_2)\ldots\det(M_m))^kw(\psi)\nonumber
\eqn
for all $\psi\in V^{\otimes m},\forall M_a\in\mathcal{M}(n)\ 1\leq a\leq m$. The inner product multiplications computed for various cases by \textbf{Schur} are given in Table \ref{tab:inner}.
\begin{table}[b]
\centering
\begin{tabular}[h]{|c||c|c||c|c||c|c||}
\hline
n & 2 & 2 & 3 & 3 & 4 & 4 \\
\hline
m & $\{21\}$ & $\{31\}$ & $\{21^2\}$ & $\{31^2\}$ & $\{21^3\}$ & $\{31^3\}$ \\
\hline
2 & 1 & 1 & 1 & 1 & 1 & 1\\
3 & 1 & 1 & 1 & 1 & 0 & 1\\
4 & 3 & 4 & 4 & 13 & 4 & 16\\
5 & 5 & 10 & 10 & 61 & 6 & 137\\
6 & 11 & 31 & 31 & 397 & 40 & 1396\\
\hline
\end{tabular}
\caption{Occurrences of $\{d\}$ in $\ast^m\{k+s,k^{n-1}\}$ with $nk\!+\!s\!=\!d$}
\label{tab:inner}
\end{table}
\\
The Markov invariants can then be computed from
\beqn
\Psi_{i_1\ldots i_{dm}}:=Y^{(1)}_{\{k+s,k^{n-1}\}}Y^{(2)}_{\{k+s,k^{n-1}\}}\ldots Y^{(m)}_{\{k+s,k^{n-1}\}}\psi_{i_1\ldots i_m}\psi_{i_{m+1}\ldots i_{2m}}\ldots \psi_{i_{(n-1)m}\ldots i_{dm}},\nonumber
\eqn
where each Young operator $Y^{(a)}_{\{k+s,k^{n-1}\}}$, $1\leq a\leq m$, is generated from a standard tableau of shape $\{k+s,k^{n-1}\}$ with integers chosen from the set $\{a,m+a,\ldots,(d-1)m+a\}$. The final step is to insert indices into $\Psi$ using the semi-standard tableau:
\vspace{-0mm}
\begin{figure}[h]
\centering
\pspicture[](0,0)(4,3)
\psline(-.25,0)(2.5,0)
\psline(-.25,.5)(2.5,.5)
\psline(-.25,1.5)(2.5,1.5)
\psline(-.25,2)(5,2)
\psline(-.25,2.5)(5,2.5)

\psline(5,2.5)(5,2)
\psline(4.5,2.5)(4.5,2)
\psline(3.5,2.5)(3.5,2)
\psline(3,2.5)(3,2)

\psline(-.25,0)(-.25,2.5)
\psline(.5,0)(.5,2.5)
\psline(1.25,0)(1.25,2.5)
\psline(1.75,0)(1.75,2.5)
\psline(2.5,0)(2.5,2.5)

\rput(0.1275,.25){n-1}
\rput(0.8775,.25){n-1}
\rput(2.1275,.25){n-1}
\rput(0.1275,1.75){1}
\rput(0.8775,1.75){1}
\rput(2.1275,1.75){1}
\rput(0.1275,2.25){0}
\rput(0.8775,2.25){0}
\rput(2.1275,2.25){0}
\rput(2.75,2.25){0}
\rput(3.25,2.25){0}
\rput(4.75,2.25){0}

\rput(1.5,2.2){...}
\rput(.2,1){.}
\rput(.2,1.1){.}
\rput(.2,.9){.}
\endpspicture .
\end{figure}\\
Again, the correct set of standard tableaux needed to generate a particular invariant is not certain, and we proceed by computing for different cases and checking for algebraic dependence until we get the correct number of algebraically independent invariants. 
\\
In what follows, we will adopt a notation where a Young operator corresponding to a certain tableau is written as $Y_{a_1,a_2,\ldots;b_1,b_2,\ldots;c_1,\ldots}$, where the commas separate column entries in the tableau and semi-colons separate the rows.
\subsection{The stochastic invariant}
For the group action of $\times^m\mathcal{M}(n)$ there is always what is known as the degree $d\!=\!1$ \textit{stochastic invariant}, $\Phi$, for all $m,n$ given by:
\beqn
\Phi:=\omega(\otimes^m\theta).\nonumber
\eqn
This corresponds to the trivial inner product multiplication
\beqn
\ast^m\{1\}=\{1\},\nonumber
\eqn
with $k=0,s=1$. Evaluated on any tensor $\psi\in V^{\otimes m}$ the stochastic invariant is simply the sum of the tensor components:
\beqn
\Phi(\psi)=\sum_{i_1,i_2,\ldots,i_m}\psi_{i_1i_2\ldots i_m}.\nonumber
\eqn
In particular, evaluated on a phylogenetic tensor $P$:
\beqn
\Phi(P)=\sum_{i_1,i_2,\ldots,i_m}p_{i_1i_2\ldots i_m}=1,\nonumber
\eqn
which motivates the terminology.
\subsection{The $n\!=\!2$ case}
From Table \ref{tab:inner} we see that for $m\!=2$ there is a single Markov invariant for each of $d\!=\!3$ and $d\!=\!4$. These can be generated by simply taking pointwise products of the stochastic invariant with the general linear group invariant ${D}_2$ (\ref{concurrence}):
\beqn
{\Phi\cdot{D}_2},\qquad \Phi^2\cdot{D}_2.\nonumber
\eqn
For $m\!=\!3$ there is a Markov invariant generated from $\{21\}$. We coin this invariant the \textit{stangle} (stochastic tangle). By directed trial and error with various tableaux, this invariant was found by taking the composition of the three Young tableaux:
\begin{figure}[h]
\centering
\pspicture[](0,0)(2,.75)
\psline(0,0)(.5,0)
\psline(0,.5)(1,.5)
\psline(0,1)(1,1)
\psline(0,0)(0,1)
\psline(.5,0)(.5,1)
\psline(1,.5)(1,1)
\rput(.25,.75){1}
\rput(.75,.75){7}
\rput(.25,.25){4}
\endpspicture
\pspicture[](0,0)(2,.75)
\psline(0,0)(.5,0)
\psline(0,.5)(1,.5)
\psline(0,1)(1,1)
\psline(0,0)(0,1)
\psline(.5,0)(.5,1)
\psline(1,.5)(1,1)
\rput(.25,.75){2}
\rput(.75,.75){8}
\rput(.25,.25){5}
\endpspicture
\pspicture[](0,0)(2,.75)
\psline(0,0)(.5,0)
\psline(0,.5)(1,.5)
\psline(0,1)(1,1)
\psline(0,0)(0,1)
\psline(.5,0)(.5,1)
\psline(1,.5)(1,1)
\rput(.25,.75){3}
\rput(.75,.75){9}
\rput(.25,.25){6}
\endpspicture .
\end{figure}\\
This is written in our new notation as
\beqn\label{stochtang2}
\Psi_{i_1i_2i_3i_4i_5i_6i_7i_8i_9}:=Y_{1,7;4}Y_{2,8;5}Y_{3,9;6}\psi_{i_1i_2i_3}\psi_{i_4i_5i_6}\psi_{i_7i_8i_9} 
\eqn
and we find that the stangle is
\beqn
\mathcal{T}_2^{s}=\Psi_{000111000}=-2\psi_{001}\psi_{010}\psi_{100}&+\psi_{000}\psi_{011}\psi_{100}+\psi_{000}\psi_{010}\psi_{101}\\
&+\psi_{000}\psi_{001}\psi_{110}-\psi_{000}^2\psi_{111}.\nonumber
\eqn
For $m\!=\!4$ there are three Markov invariants which we call the \textit{squangles} (stochastic quangles). One of these Markov invariants can be generated simply by taking the pointwise product of the quangle multiplied by the stochastic invariant:
\beqn
\Phi\cdot Q_2.\nonumber
\eqn
By directed trial and error the other two squangles have been found to be generated from
\beqn\label{squangle2}
Y_{1,5;9}Y_{2,6;10}Y_{3,11;7}Y_{4,12;8}\psi_{i_1i_2i_3i_4}\psi_{i_5i_6i_7i_8}\psi_{i_9i_{10}i_{11}i_{12}}
\eqn
and
\beqn
Y_{1,5;9}Y_{2,10;6}Y_{3,11;7}Y_{4,12;8}\psi_{i_1i_2i_3i_4}\psi_{i_5i_6i_7i_8}\psi_{i_9i_{10}i_{11}i_{12}}.\nonumber
\eqn
Explicitly the first squangle is
\beqn
Q^{s_1}_2&=\psi_{0 0 1 1} \psi_{0 1 0 0} \psi_{1 0 0 0} + \psi_{0 0 1 0} \psi_{
        0 1 0 1} \psi_{1 0 0 0} + \psi_{0 0 0 
    1} \psi_{0 1 1 0} \psi_{1 0 0 0} - \psi_{0 0 0 0} \psi_{0 1 1 1} \psi_{1 
        0 0 0}\\& + \psi_{0 0 1 
    0} \psi_{0 1 0 0} \psi_{1 0 0 1} + \psi_{0 0 0 1} \psi_{0 1 0 0} \psi_{1 
        0 1 0} - \psi_{0 0 0 
    0} \psi_{0 1 0 0} \psi_{1 0 1 1} - 2\ \psi_{0 0 0 1} \psi_{0 0 1 0} \psi_{
        1 1 0 0}\\& + 3\ \psi_{0 
    0 0 0} \psi_{0 0 1 1} \psi_{1 1 0 0} - \psi_{0 0 0 0} \psi_{0 0 1 0} 
        \psi_{1 1 0 1} - \psi_{0 
    0 0 0} \psi_{0 0 0 1} \psi_{1 1 1 0} + \psi_{0 0 0 0}^2\ \psi_{1 1 1 1},\nonumber
\eqn
and the second
\beqn
Q^{s_2}_2&=\psi_{0 0 1 1} \psi_{0 1 0 0} \psi_{1 0 0 0} - 2 \psi_{0 0 1 0} \psi_{0 1 0
         1} \psi_{1 0 0 0} + 
    \psi_{0 0 0 1} \psi_{0 1 1 0} \psi_{1 0 0 0} - \psi_{0 0 0 0} \psi_{0 1 1
         1} \psi_{1 0 0 0}\\& + \psi_{0 0
     1 0} \psi_{0 1 0 0} \psi_{1 0 0 1} - 2 \psi_{0 0 0 1} \psi_{0 1
         0 0} \psi_{1 0 1 0} + 
    3 \psi_{0 0 0 0} \psi_{0 1 0 1} \psi_{1 0 1 0} - \psi_{0 0 0 0} \psi_{0 1 0
         0} \psi_{1 0 1 1}\\& + 
    \psi_{0 0 0 1} \psi_{0 0 1 0} \psi_{1 1 0 0} - \psi_{0 0 0 0} \psi_{0 0 1
         0} \psi_{1 1 0 1} - 
    \psi_{0 0 0 0} \psi_{0 0 0 1} \psi_{1 1 1 0} + \psi_{0 0 0 0}^2 \psi_{1 1
         1 1}.\nonumber
\eqn
The three degree $d=3$ Markov invariants $\{\Phi\cdot Q_2,Q^{s_1}_2,Q^{s_2}_2\}$ have been shown by explicit computation to be linearly independent, as required.
\subsection{The $n\!=\!3$ case}
From Table \ref{tab:inner}, there are two Markov invariants for $n\!=\!3,m\!=\!2$ of degree $d\!=\!4,5$. Again these invariants can be easily produced by taking products of the stochastic invariant with the determinant invariant (\ref{concurrence}):
\beqn
\Phi\cdot {D}_3,\qquad\Phi^2\cdot{D}_3.\nonumber
\eqn
In the case $m\!=\!3$ there is a single Markov invariant, which we also refer to as the \textit{stangle}:
\beqn\label{stochtang3}
\Psi_{i_1i_2i_3i_4i_5i_6i_7i_8i_9i_{10}i_{11}i_{12}}&\\:=Y_{1,4;7;10}&Y_{2,8;5;11}Y_{3,12;6,9}\psi_{i_1i_2i_3}\psi{i_4i_5i_6}\psi_{i_7i_8i_9}\psi_{i_{10}i_{11}i_{12}},
\eqn
so that
\beqn
\mathcal{T}_3^{(s)}&=\Psi_{000011102220}\\
&=\psi_{0 1 2} \psi_{0 2 0} \psi_{1 0 1} \psi_{2 0 0} - \psi_{0 1 0} \psi_{0 2 2} \psi_{1 0 
      1} \psi_{2 0 0} - 
  \psi_{0 1 1} \psi_{0 2 0} \psi_{1 0 2} \psi_{2 0 0}\\& + \psi_{0 1 0} \psi_{0 2 1} \psi_{1 0 
      2} \psi_{2 0 0} - 
  \psi_{0 0 2} \psi_{0 2 1} \psi_{1 1 0} \psi_{2 0 0} + \psi_{0 0 1} \psi_{0 2 2} \psi_{1 1 
      0} \psi_{2 0 0}\\& + 
  \psi_{0 0 2} \psi_{0 1 1} \psi_{1 2 0} \psi_{2 0 0} - \psi_{0 0 1} \psi_{0 1 2} \psi_{1 2 
      0} \psi_{2 0 0} - 
  \psi_{0 1 2} \psi_{0 2 0} \psi_{1 0 0} \psi_{2 0 1}\\& + \psi_{0 1 0} \psi_{0 2 2} \psi_{1 0 
      0} \psi_{2 0 1} + 
  \psi_{0 0 2} \psi_{0 2 0} \psi_{1 1 0} \psi_{2 0 1} - \psi_{0 0 0} \psi_{0 2 2} \psi_{1 1 
      0} \psi_{2 0 1}\\& - 
  \psi_{0 0 2} \psi_{0 1 0} \psi_{1 2 0} \psi_{2 0 1} + \psi_{0 0 0} \psi_{0 1 2} \psi_{1 2 
      0} \psi_{2 0 1} + 
  \psi_{0 1 1} \psi_{0 2 0} \psi_{1 0 0} \psi_{2 0 2}\\& - \psi_{0 1 0} \psi_{0 2 1} \psi_{1 0 
      0} \psi_{2 0 2} - 
  \psi_{0 0 1} \psi_{0 2 0} \psi_{1 1 0} \psi_{2 0 2} + \psi_{0 0 0} \psi_{0 2 1} \psi_{1 1 
      0} \psi_{2 0 2}\\& + 
  \psi_{0 0 1} \psi_{0 1 0} \psi_{1 2 0} \psi_{2 0 2} - \psi_{0 0 0} \psi_{0 1 1} \psi_{1 2 
      0} \psi_{2 0 2} + 
  \psi_{0 0 2} \psi_{0 2 1} \psi_{1 0 0} \psi_{2 1 0}\\& - \psi_{0 0 1} \psi_{0 2 2} \psi_{1 0 
      0} \psi_{2 1 0} - 
  \psi_{0 0 2} \psi_{0 2 0} \psi_{1 0 1} \psi_{2 1 0} + \psi_{0 0 0} \psi_{0 2 2} \psi_{
      1 0 1} \psi_{2 1 0}\\& + \psi_{0 0 1} \psi_{0 2 0} \psi_{1 0 2} \psi_{2 1 0} - \psi_{
      0 0 0} \psi_{0 2 1} \psi_{1 0 2}
   \psi_{2 1 0} - \psi_{0 0 2} \psi_{0 1 1} \psi_{1 0 0} \psi_{2 2 0}\\& + \psi_{0 0 1} \psi_{0 
      1 2} \psi_{1 0 0}
   \psi_{2 2 0} + \psi_{0 0 2} \psi_{0 1 0} \psi_{1 0 1} \psi_{2 2 0} - \psi_{0 0 0} \psi_{0 
      1 2} \psi_{1 0 1}
   \psi_{2 2 0}\\& - \psi_{0 0 1} \psi_{0 1 0} \psi_{1 0 2} \psi_{2 2 0} + \psi_{0 0 0} \psi_{0 
      1 1} \psi_{1 0 2} \psi_{2 2 0}.\nonumber
\eqn
In the case of $m\!=\!4$, Table \ref{tab:inner} predicts four Markov invariants, which we again refer to as squangles. One of the squangles can be inferred directly as the pointwise product:
\beqn
\Phi\cdot Q_3,\nonumber
\eqn
and by directed trial and error we have shown that the other three can be generated from the Young operators:
\beqn\label{squangle3}
Q^{s_1}_3&\leftarrow Y_{1,5;9;13}Y_{2,6;10;14}Y_{3,7;11;15}Y_{4,8;12;16},\\
Q^{s_2}_3&\leftarrow Y_{1,9;5;13}Y_{2,14;6;10}Y_{3,7;11;15}Y_{4,8;12;16},\\
Q^{s_3}_3&\leftarrow Y_{1,9;5;13}Y_{2,10;6;14}Y_{3,7;11;15}Y_{4,8;12;16},
\eqn
where $\leftarrow$ indicates the implementation of our procedure with the indices of $\Psi$ filled out to create the only semi-standard tableau of shape $\{21^2\}$ using the integers $\{0,1,2\}$. The four invariants $\{\Phi\cdot Q_3,Q^{s_1}_3,Q^{s_2}_3,Q^{s_3}_3\}$ have been shown by explicit computation to be linearly independent.
\subsection{The $n\!=\!4$ case}
In the case of $n\!=\!4,m\!=\!2$, Table \ref{tab:inner} predicts a Markov invariant of degree $d=5,6$. Again, these invariants can be generated easily as the pointwise products:
\beqn
\Phi\cdot D_4,\qquad\Phi^2\cdot D_4.\nonumber
\eqn
In the case of $m\!=\!3$ Table \ref{tab:inner} predicts a degree $d\!=\!6$ Markov invariant which we again refer to as the \textit{stangle}. It is generated from the Young operator
\beqn\label{eq:stochtang4}
\mathcal{T}_4^{s}\leftarrow Y_{1,4,13;7,10,16}Y_{2,8,17;5,11,14}Y_{3,12,18;6,9,15}.
\eqn
Explicitly this polynomial has 1404 terms.
\\
In the case of $m\!=\!4$ there are four degree $d=5$ Markov invariants which we again refer to as squangles. One of these is generated easily as
\beqn
\Phi\cdot Q_4,\nonumber
\eqn
and by directed trial and error the other three have been found to be given by the Young operators:
\beqn\label{squangle4}
Q^{s_1}_4&\leftarrow Y_{1,5;9;13;17} Y_{2,6;10;14;18} Y_{3,7;11;15;19} Y_{4,8;12;16;20},\\
Q^{s_2}_4&\leftarrow Y_{1,9;5;13;17} Y_{2,14;6;10;18} Y_{3,7;11;15;19} Y_{4,8;12;16;20},\\
Q^{s_3}_4&\leftarrow Y_{1,9;5;13;17} Y_{2,14;6;10;18} Y_{3,19;7;11;15} Y_{4,8;12;16;20}.
\eqn
The four degree $d\!=\!5$ Markov invariants $\{\Phi\cdot Q_4,Q^{s_1}_4,Q^{s_2}_4,Q^{s_3}_4\}$ have been shown by explicit computation to be linearly independent, as required.
\section{What happens on a phylogenetic tree?}
In this section we will examine the structure of the invariant functions we have discovered on phylogenetic trees. We will focus on the case of four characters $n=4$ and three and four leaves $m=3,4$.\\
We have discovered invariant functions which satisfy
\beqn
w(g\psi)=\det(g)^kw(\psi),\nonumber
\eqn
for all $g\in \times^m \mathcal{M}(n)$ and $\psi\in V^{\otimes m}$. If we consider the case where these invariants are evaluated on the phylogenetic tensor $P$, the invariant takes the form
\beqn
w(P)=\prod_{a=1}^m \det(M_{a})^kw(\widetilde{P}).\nonumber
\eqn
Our task is to examine the structure of the Markov invariants when evaluated on the phylogenetic tensor $\widetilde{P}$ corresponding to the various possible trees. 

\subsection{The stangle}
As we saw in Chapter \ref{chap4}, we need only consider \textit{unrooted} phylogenetic trees. For the case of three taxa the most general phylogenetic tree is:
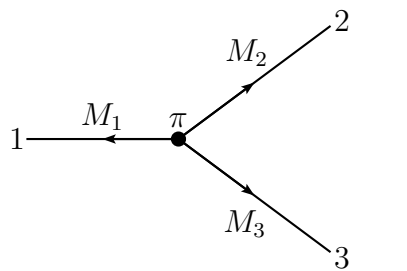
\beqn
\\
\\
\pspicture[](5,0)(1,2.5)
\psset{linewidth=\pstlw,xunit=0.5,yunit=0.5,runit=0.5}
\psset{arrowsize=2pt 2,arrowinset=0.2}
\psline{-}(4,3)(0,3)
\rput(-.25,3){1}
\psline{-}(4,3)(8,0)
\rput(8.3,-0.15){3}
\psline{-}(4,3)(8,6)
\rput(8.3,6.15){2}
\psline{->}(4,3)(2,3)
\rput(2,3.6){$M_1$}
\psline{->}(4,3)(6,4.5)
\rput(5.8,5.3){$M_2$}
\psline{->}(4,3)(6,1.5)
\rput(5.75,.75){$M_3$}
\pscircle[linewidth=0.8pt,fillstyle=solid,fillcolor=black](4,3){.2} 
\rput(4,3.5){$\pi$}
\endpspicture . \nonumber
\\
\\
\eqn
The corresponding phylogenetic tensor can be expressed as
\beqn
P=(M_1\otimes M_2\otimes M_3)1\otimes \delta\cdot\delta\cdot\pi,\nonumber
\eqn
where
\beqn
\delta^2:=1\otimes\delta\cdot\delta=\delta\otimes 1\cdot\delta.\nonumber
\eqn
From the general properties of the Markov invariants we find that
\beqn
\mathcal{T}^{(s)}(P)=\det(M_1)\det(M_2)\det(M_3)\mathcal{T}^{(s)}(\widetilde{P}),\nonumber
\eqn
and by direct computation 
\beqn
\mathcal{T}^{(s)}(\widetilde{P})=0.\nonumber
\eqn
It follows that evaluating the stangle on the general phylogenetic tensor of four leaves satisfies
\beqn
\mathcal{T}^{(s)}(P)=0.\nonumber
\eqn
This equation is independent of all the model parameters contained in the phylogenetic tree. This observation implies that this Markov invariant also satisfies the properties of a \textit{phylogenetic invariant} for the general Markov model \cite{allman2003}.
\subsection{The squangles}
For the case of four taxa there are three inequivalent unrooted phylogenetic trees as presented in Figure \ref{pic:threequartets}.
\begin{figure}[t]
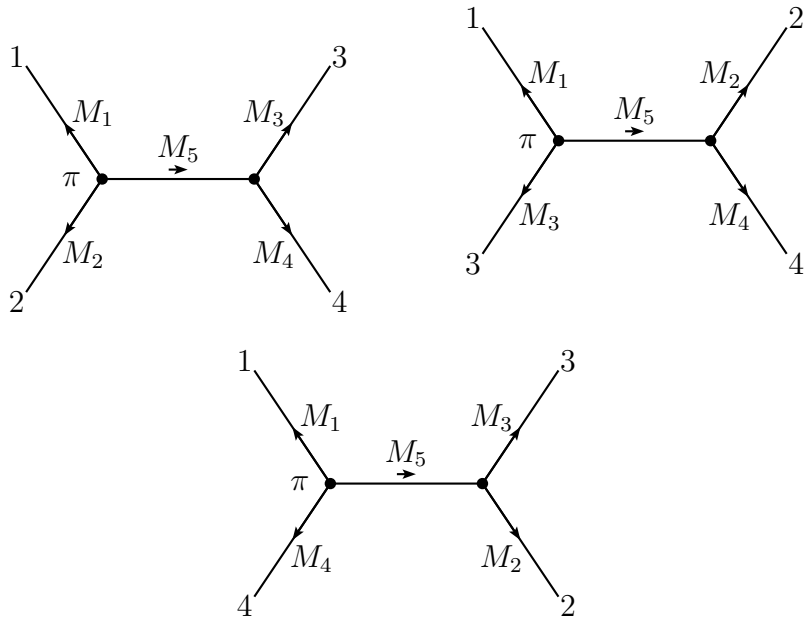

\label{fig:fourtaxa}
\pspicture[](-7,0)(0,3)
\psset{linewidth=\pstlw,xunit=0.5,yunit=0.5,runit=0.5}
\psset{arrowsize=2pt 2,arrowinset=0.2}
\psline{-}(2,3)(0,0)
\rput(-.25,-.25){3}
\psline{-}(2,3)(0,6)
\rput(-.25,6.25){1}
\psline{-}(2,3)(6,3)
\psline{-}(6,3)(8,6)
\rput(8.25,6.25){2}
\psline{-}(6,3)(8,0)
\rput(8.25,-0.25){4}
\psline{->}(2,3)(1,4.5)
\rput(1.75,4.75){$M_1$}
\psline{->}(2,3)(1,1.5)
\rput(1.5,1.0){$M_3$}
\psline{->}(6,3)(7,4.5)
\rput(6.25,4.75){$M_2$}
\psline{->}(6,3)(7,1.5)
\rput(6.5,1.0){$M_4$}
\psline{->}(3.75,3.25)(4.25,3.25)
\rput(4,3.8){$M_5$}
\pscircle[linewidth=0.8pt,fillstyle=solid,fillcolor=black](2,3){.15} 
\pscircle[linewidth=0.8pt,fillstyle=solid,fillcolor=black](6,3 ){.15} 
\rput(1.2,3){$\pi $}
\endpspicture

\pspicture[](-1,-3)(3,0)
\psset{linewidth=\pstlw,xunit=0.5,yunit=0.5,runit=0.5}
\psset{arrowsize=2pt 2,arrowinset=0.2}
\psline{-}(2,3)(0,0)
\rput(-.25,-.25){2}
\psline{-}(2,3)(0,6)
\rput(-.25,6.25){1}
\psline{-}(2,3)(6,3)
\psline{-}(6,3)(8,6)
\rput(8.25,6.25){3}
\psline{-}(6,3)(8,0)
\rput(8.25,-0.25){4}
\psline{->}(2,3)(1,4.5)
\rput(1.75,4.75){$M_1$}
\psline{->}(2,3)(1,1.5)
\rput(1.5,1.0){$M_2$}
\psline{->}(6,3)(7,4.5)
\rput(6.25,4.75){$M_3$}
\psline{->}(6,3)(7,1.5)
\rput(6.5,1.0){$M_4$}
\psline{->}(3.75,3.25)(4.25,3.25)
\rput(4,3.8){$M_5$}
\pscircle[linewidth=0.8pt,fillstyle=solid,fillcolor=black](2,3){.15} 
\pscircle[linewidth=0.8pt,fillstyle=solid,fillcolor=black](6,3 ){.15} 
\rput(1.2,3){$\pi $}
\endpspicture

\pspicture[](-4,0)(0,1)
\psset{linewidth=\pstlw,xunit=0.5,yunit=0.5,runit=0.5}
\psset{arrowsize=2pt 2,arrowinset=0.2}
\psline{-}(2,3)(0,0)
\rput(-.25,-.25){4}
\psline{-}(2,3)(0,6)
\rput(-.25,6.25){1}
\psline{-}(2,3)(6,3)
\psline{-}(6,3)(8,6)
\rput(8.25,6.25){3}
\psline{-}(6,3)(8,0)
\rput(8.25,-0.25){2}
\psline{->}(2,3)(1,4.5)
\rput(1.75,4.75){$M_1$}
\psline{->}(2,3)(1,1.5)
\rput(1.5,1.0){$M_4$}
\psline{->}(6,3)(7,4.5)
\rput(6.25,4.75){$M_3$}
\psline{->}(6,3)(7,1.5)
\rput(6.5,1.0){$M_2$}
\psline{->}(3.75,3.25)(4.25,3.25)
\rput(4,3.8){$M_5$}
\pscircle[linewidth=0.8pt,fillstyle=solid,fillcolor=black](2,3){.15} 
\pscircle[linewidth=0.8pt,fillstyle=solid,fillcolor=black](6,3 ){.15} 
\rput(1.2,3){$\pi $}
\endpspicture
 
\caption{Three alternative quartet trees}
\label{pic:threequartets}
\end{figure}
 The corresponding phylogenetic tensors are
\begin{itemize}
\item $P^{(1)}=M_1\otimes M_2\otimes M_3\otimes M_4(1\otimes 1\otimes \delta M_5)\delta^2\pi$
\item $P^{(2)}=M_1\otimes M_2\otimes M_3\otimes M_4(1\otimes \delta M_5 \otimes 1)\delta^2\pi$
\item $P^{(3)}=M_1\otimes M_2\otimes M_3\otimes M_4(\delta M_5\otimes 1\otimes 1)\delta^2\pi$.
\end{itemize}
For any linear combination of the Markov invariants: 
\beqn
w=c\Phi\cdot Q_4 +c_1Q^{s_1}_4+c_2Q^{s_2}_4+c_3Q^{s_3}_4\,\nonumber
\eqn
we have
\beqn
w(P^{(1)})=\det(M_1)\det(M_2)\det(M_3)\det(M_4)w(\widetilde{P}^{(a)}),\quad a=1,2,3.\nonumber
\eqn
Defining the linearly independent combinations
\beqn
L_1&=-\fra{3}{2}Q^{s_1}_4+Q^{s_2}_4+2Q^{s_3}_4,\nonumber\\
L_2&=-\fra{3}{2}Q^{s_1}_4+2Q^{s_2}_4+Q^{s_3}_4,\\
L_3&=-Q^{s_2}_4+Q^{s_3}_4.
\eqn
it is possible to show by direct computation that the following relations hold:
\begin{itemize}
\item $L_1(\widetilde{P}^{(1)})=0,\quad L_2(\widetilde{P}^{(1)})=-L_3(\widetilde{P}^{(1)})>0;$
\item $L_2(\widetilde{P}^{(2)})=0,\quad L_1(\widetilde{P}^{(1)})=L_3(\widetilde{P}^{(1)})>0;$
\item $L_3(\widetilde{P}^{(3)})=0,\quad L_1(\widetilde{P}^{(1)})=L_2(\widetilde{P}^{(1)})<0$.
\end{itemize}
This implies that these linear combinations of the squangles are not only Markov invariants, but also phylogenetic  invariants \cite{allman2003}. They are actually phylogenetically \textit{informative} invariants because they can be used to distinguish between the three quartet topologies. Studying the statistical properties of this technique is a topic of ongoing work (see Appendix \ref{app1}).
\section{Review of important invariants}
We tabulate the invariant functions that have been of interest in this thesis in Table \ref{fig:dong}. It should be noted that in the case of the squangles the invariants of the general linear group are included with the invariants of the Markov semigroup.
\begin{table}[ht]
\begin{tabular}[h]{|c|c|c|c|c|c|}
\hline
Name & Symbol & Schur multi. & Group & $(d,k)$ & Ref. \\
\hline
$det$& $\det_2$ & $\ast^2\{1^2\}=\{2\}$ & $\times^2GL(2)$ & (2,1) & (\ref{concurrence})\\
&$\det_3$ & $\ast^2\{1^3\}=\{3\}$ & $\times^2GL(3)$ & (3,1) & (\ref{concurrence})\\
&$\det_4$ & $\ast^2\{1^4\}=\{4\}$ & $\times^2GL(4)$ & (4,1) & (\ref{concurrence})\\
$tangle$ &$\mathcal{T}_2$ & $\ast^3\{2^2\}\ni\{4\}$ & $\times^3GL(2)$ & (4,2) & (\ref{tangle2})\\
&$\mathcal{T}_3$ & $\ast^3\{2^3\}\ni\{6\}$ & $\times^3GL(3)$ & (6,2) & (\ref{tangle3})\\
&$\mathcal{T}_4$ & $\ast^3\{2^4\}\ni\{8\}$ & $\times^3GL(4)$ & (8,2) & (\ref{tangle4})\\
$stangle$&$\mathcal{T}_2^{s}$ & $\ast^3\{21\}\ni\{3\}$ & $\times^3\mathcal{M}(2)$ & (3,1) & (\ref{stochtang2})\\
$$&$\mathcal{T}_3^{s}$ & $\ast^3\{21^2\}\ni\{4\}$ & $\times^3\mathcal{M}(3)$ & (4,1) & (\ref{stochtang3})\\
&$\mathcal{T}_4^{s}$ & $\ast^3\{31^3\}\ni\{6\}$ & $\times^3\mathcal{M}(4)$ & (6,1) & (\ref{eq:stochtang4})\\
$quangle$&$Q_2$ & $\ast^4\{1^2\}\ni\{2\}$ & $\times^4GL(2)$ & (2,1) & (\ref{eq:quangles})\\
&$Q_3$ & $\ast^4\{1^3\}\ni\{3\}$ & $\times^4GL(3)$ & (3,1) & (\ref{eq:quangles})\\
&$Q_4$ & $\ast^4\{1^4\}\ni\{4\}$ & $\times^4GL(4)$ & (4,1) & (\ref{eq:quangles})\\
$squangle$&$(Q_2,Q_2^{s_1},Q_2^{s_2})$ & $\ast^4\{21\}\ni 3\{3\}$ & $\times^4\mathcal{M}(2)$ & (3,1) & (\ref{squangle2})\\
$$&$(Q_3,Q_3^{s_1},Q_3^{s_2},Q_3^{s_3})$ & $\ast^4\{21^2\}\ni 4\{4\}$ & $\times^4\mathcal{M}(3)$ & (3,1) & (\ref{squangle3})\\
&$(Q_4,Q_4^{s_1},Q_4^{s_2},Q_4^{s_3})$ & $\ast^4\{21^3\}\ni 4\{5\}$ & $\times^4\mathcal{M}(4)$ & (5,1) & (\ref{squangle4})\\
\hline
\end{tabular}
\caption{Invariant functions satisfying $f\circ g=\det(g)^kf$}
\label{fig:dong}
\end{table}
\section{Closing remarks}
In this chapter we have defined and proved the existence of Markov invariants. We have shown how to derive their explicit polynomial form in interesting cases. We examined the structure of several invariants in the context of phylogenetic trees. Finally, we derived a novel technique of quartet tree reconstruction which is valid under the assumptions of the general Markov model of sequence evolution.


\chapter[CONCLUSION]{Conclusion} \label{concl}

In this thesis we have examined the mathematical analogy between quantum physics and the Markov model of a phylogenetic tree.\\
In Chapter \ref{chap2} we gave a review of group representation theory, established the Schur/Weyl duality and went on to show how one-dimensional representations and invariant functions of the general linear group can be put into coincidence. We also presented several examples of the explicit polynomial form of these invariants.\\
In Chapter \ref{chap3} we concretely established the mathematical analogy between entanglement and that of phylogenetic relation. We showed that group invariant functions can be used to quantify a measure of phylogenetic relation.\\
In Chapter \ref{chap4} we gave a review of pairwise phylogenetic distance measures and examined the use of the tangle in improving the calculation of pairwise distance measures from observed sequence data.\\
In Chapter \ref{chap5} we defined and showed how to derive Markov invariant functions. We studied their properties in cases relevant to the problem of phylogenetic tree reconstruction. We derived a new technique for reconstruction of quartets which is valid under the assumptions of a general Markov model.\\
\subsubsection{Future investigations}
There are several clear paths for continuing the work that has been presented in this thesis.\\
Rather than use the tangle to give improved pairwise distances it seems judicious to examine how the tangle could be used in more direct ways. The Neighbour-Joining (NJ) algorithm for tree reconstruction has at its core the concept of pairwise distances and in opposition to this the tangle polynomial actually gives a measure of the sum of the branch lengths for a triplet. Hence it seems that one possibility is to generalize the NJ algorithm in such a way that the tangle is incorporated explicitly into the procedure. Additionally, biologists are interested in the evolutionary distance between taxa and another possibility would be to use the tangle as a measure of the evolutionary distance between triplets of taxa without decomposing this distance into pairs. Given a set of multiple taxa one could construct interesting questions comparing different triplets using the value of the tangle as a quantifier. \\
The stochastic tangle is a very interesting mathematical object as it simultaneously satisfies the properties of a Markov invariant and that of a phylogenetic invariant. In this thesis we have not investigated the potential of finding a practical role for the stochastic tangle in the problem of phylogenetic reconstruction. The possibilities of practical roles are similar to that of the tangle and we leave this as an open problem.
\\
The squangles have been shown to give a new tree reconstruction algorithm for the case of quartets. The main path for future investigation is to study the statistical properties of such an algorithm. It is theoretically clear how to calculate unbiased forms of the squangles (see Appendix \ref{app1}) and this would be a desirable practical outcome as it will improve the performance of the quartet reconstruction in the case where the sequence data is of relatively short length. Unfortunately this calculation of an unbiased form is computationally difficult and has not been achieved. To further the complete statistical understanding it is necessary to calculate the variance of the squangles. Again this is theoretically clear but computationally difficult as one is required to square the polynomials.\\
In this thesis we have used the concept of a tree in a rather ad hoc way. Our procedure was to compute the explicit polynomial form of the invariant functions and then to impose a given tree structure onto the polynomial by choosing coordinates for the tensors selected to be consistent with the tree. Given that the existence of the invariant functions was proved using the Schur functions series, a natural corollary would be to ask if it is possible to identify the relationships between the invariant functions that occur on particular trees by simply studying the properties of the Schur functions in more detail. The branching operator $\delta$ is technically an invertible linear operator on the expanded linear space known as a $Fock$ space and it follows that the character theory of this action together with that of the Markov semigroup should introduce the possibility of ``seeing'' the tree structure within the Schur functions. Hence it seems feasible to identify the relationships between the invariant functions that occur on particular trees by simply studying the properties of the Schur functions in more detail. \\
The other clear course for theoretical investigation is to completely classify the ring of invariants for the Markov semigroup. This is not an easy problem as the Hilbert basis theorem states that the ring of invariants is guaranteed to be finitely generated if the group action is completely reducible \cite{kraft2000}. However, the Markov group has an invariant subspace with no complementary invariant subspace and is hence not completely reducible. Further study is required to fully characterize the ring of Markov invariants. Additionally, the exact connection between the ring of Markov invariants and the ideal of phylogenetic invariants should be established concretely. In this thesis this connection was only made for the particular cases that were of interest. A well defined and complete description of the connection is required before one can speak with confidence on this matter. 

\appendix


\chapter[Bias correction of invariant functions]{Bias correction of invariant functions}\label{app1}
\section{Multinomial distribution}
Let $X_a,\ 1\leq a\leq n$, be the random variable which counts the occurrences of character $a$ in a finite subset of an infinite sequence consisting of the characters $\{1,2,...,n\}$. If each character occurs with probability $p_a$, then for a subset of length $N$ we have the standard multinomial distribution
\beqn\label{multinom}
\mathbb{P}(X_1=k_1,X_2=k_2,...,X_n=k_n)=\frac{N!}{k_1!k_2!...k_n!}p_1^{k_1}p_2^{k_2}...p_n^{k_n}.
\eqn
Defining the vector valued random variable $X=(X_1,X_2,...,X_n)\in\mathbb{N}^n$, we can express (\ref{multinom}) as
\beqn
\mathbb{P}(X=k)=\frac{N!}{\prod_{a=1}^nk_a!}\prod_{b=1}^np_{b}^{k_b},\nonumber
\eqn
with $k=(k_1,...,k_n)\in\mathbb{N}^n$ and $k_1+k_2+...+k_n=N$. Consider any function
\beqn
\phi:\mathbb{C}^n\rightarrow\mathbb{C}^q,\ q\in\mathbb{N}.\nonumber
\eqn
The expectation value of $\phi(X)$ is then defined as
\beqn
E[\phi(X)]=\sum_{k\in\mathbb{N}: k_1+k_2+...+k_n=N}\mathbb{P}(X=k)\phi(k).\nonumber
\eqn
\section{Generating function}
For every $s\in\mathbb{R}^n$ we define the generating function $G:\mathbb{R}^n\rightarrow\mathbb{C}$ as
\beqn
G(s)=E[e^{i(s,X)}],\nonumber
\eqn
where we have considered $X\in\mathbb{N}^n\subset\mathbb{R}^n$ and $(s,X)=s_1X_1+s_2X_2+...+s_nX_n$\newpage\noindent and convergence is ensured by $|e^{i(s,X)}|=1$ and the triangle inequality.\\
Observe that
\beqn
\frac{\partial G(s)}{\partial s_j}=E[iX_je^{i(s,X)}].\nonumber
\eqn
In particular we have
\beqn
\frac{\partial G(s)}{\partial s_j}|_{s=0}=iE[X_j].\nonumber
\eqn
We simplify notation by taking the Laplace transform
\beqn
s\rightarrow is,\nonumber
\eqn
and find that in general
\beqn
\frac{\partial^{b_1+b_2+...+b_m}G(s)}{\partial s_{a_1}^{b_1}\partial s_{a_2}^{b_2}...\partial s_{a_m}^{b_m}}|_{s=0}
=E[X_{a_1}^{b_1}X_{a_2}^{b_2}...X^{b_m}_{a_m}].\nonumber
\eqn
Computing a closed form of $G(s)$ follows easily given the identity
\beqn
(x_1+x_2+...+x_n)^N=\sum_{k\in\mathbb{N}^n:k_1+k_2+...+k_n=N}\frac{N!}{k_1!k_2!...k_n!}x_1^{k_1}x_{2}^{k_2}...x_n^{k_n},\nonumber
\eqn
so that
\beqn
G(s)=(p_1e^{s_1}+p_2e^{s_2}+...+p_ne^{s_n})^N.\nonumber
\eqn
In particular $G(0)=1$.
\section{Expectations of polynomials}
We are particularly interested in the case when 
\beqn
\phi\in\mathbb{C}[V]_d,\quad V\cong \mathbb{C}^n.\nonumber
\eqn
In general we have 
\beqn
E[(\phi_1+c\phi_2)(X)]=E[\phi_1(X)]+cE[\phi_2(X)],\nonumber
\eqn
but
\beqn
E[\phi_1\cdot\phi_2(X)]\neq E[\phi_1(X)]E[\phi_2(X)].\nonumber
\eqn
Thus in order to calculate the expected value of a polynomial we need only study expectation values of monomials:
\beqn
E[X_{a_1}^{b_1}X_{a_2}^{b_2}...X_{a_m}^{b_m}],\quad m\leq n.\nonumber
\eqn
In particular we have
\beqn\label{monoexpect}
E[X_a]&=\frac{\partial G(s)}{\partial s_a}|_{s=0}\\
&=Np_a,\\
E[X_aX_b]&=\frac{\partial^2 G(s)}{\partial s_a\partial s_b}|_{s=0}\\
&=N(N-1)p_ap_b+Np_a\delta_{ab},\\
E[X_aX_bX_c]=&\frac{\partial^3 G(s)}{\partial s_a\partial s_b\partial s_c}|_{s=0}\\
=&N(N-1)(N-2)p_ap_bp_c\\&+N(N-1)(p_ap_b\delta_{ac}+p_ap_c\delta_{ab}+p_bp_c\delta_{ab})+Np_a\delta_{ab}\delta_{ac},\eqn
and for a set of distinct integers $1\leq a_1,a_2,...,a_d\leq n\}$ we have
\beqn\label{monoform}
E[X_{a_1}X_{a_2}...X_{a_d}]=\frac{N!}{(N-d)!}p_{a_1}p_{a_2}...p_{a_m}.
\eqn
\section{Bias correction}
For a given homogeneous polynomial $\phi$ of degree $d$, we would like to find a polynomial $\widetilde{\phi}$ such that
\beqn
E[\widetilde{\phi}(X)]=\phi(p).\nonumber
\eqn
We refer to $\widetilde{\phi}$ as the unbiased form of $\phi$.\\
By looking at the general form of the invariants ${\det}_n$ it can be seen that every monomial term is of the form (\ref{monoform}). It follows easily that
\beqn
E[{\det}_n(X)]=\frac{N!}{(N-n)!}{\det}_n(p),\nonumber
\eqn
so that the unbiased version is given simply by
\beqn
\widetilde{{\det}_n}:=\frac{(N-d)!}{N!}{\det}_n.\nonumber
\eqn
It should be noted that this says nothing about what to do about finding an unbiased form of $\log\det$, because the $\log$ function is \textit{not} polynomial. For discussion on the bias correction of the $\log\det$ function see \cite{barry1987}.\\
We leave the computation of unbiased forms of the other invariants presented in this thesis as an open problem. However, the process is exemplified in the following.\\
Consider the expectation:
\beqn
E[X_1X_2X_3]=N(N-1)(N-2)p_1p_2p_3.\nonumber
\eqn
Thus the unbiased form of this monomial is simply
\beqn
\frac{(N-3)!}{N!}X_1X_2X_3.\nonumber
\eqn
Consider
\beqn
E[X_1^2X_2]=N(N-1)(N-2)p_1^2p_2+N(N-1)p_1p_2.\nonumber
\eqn
The unbiased form of this monomial is then
\beqn
\frac{(N-3)!}{N!}(X_1^2X_2-X_1X_2),\nonumber
\eqn
since
\beqn
E[\frac{(N-3)!}{N!}(X_1^2X_2-X_1X_2)]=p_1^2p_2.\nonumber
\eqn
By generalizing (\ref{monoexpect}) for a set of distinct integers $1\leq a,b_1,b_2,...,b_m\leq n$ it follows that 
\beqn
E[\fra{N!}{(N-(m+1))!}(X_a^2X_{b_1}X_{b_2}...X_{b_m}-X_aX_{b_1}X_{b_2}...X_{b_m})]=p_a^2p_{b_1}p_{b_2}...p_{b_m}.\nonumber
\eqn
This is the first step to computing the unbiased form of general monomials. Clearly the process becomes more complicated as the degree of a given random variable within each monomial becomes larger. 

\include{app2}
\small
\bibliographystyle{plain}
\bibliography{thesis}

\printindex

\index{invariant}

\end{document}